\documentclass[format=acmsmall,screen=false,review=false]{acmart}

\usepackage[utf8]{inputenc}
\usepackage{amssymb}
\usepackage{color}
\usepackage{ifthen}
\usepackage{tikz}
\usepgflibrary{arrows}
\usetikzlibrary{fit}
\usetikzlibrary{decorations.pathreplacing}
\usetikzlibrary{decorations.markings}
\usetikzlibrary{patterns}
\newcommand\numberthis{\addtocounter{equation}{1}\tag{\theequation}}
\usepackage{mathdots}
\usepackage{mathtools}
\usepackage{stmaryrd}
\usepackage{dsfont}
\usepackage{longtable}
\usepackage{booktabs}
\usepackage{savefnmark}
\usepackage{todonotes}

\newcommand\COMMENTESAUFSILATEXDIFF[1]{}
\newcommand\PASCOMMENTESAUFSILATEXDIFF[1]{#1}

\newcommand{\A}{\mathbb{A}}
\newcommand{\R}{\mathbb{R}}
\newcommand{\Q}{\mathbb{Q}}
\newcommand{\N}{\mathbb{N}}
\newcommand{\D}{\mathbb{D}}
\COMMENTESAUFSILATEXDIFF{}
\newcommand{\Z}{\mathbb{Z}}
\newcommand{\K}{\mathbb{K}}

\newcommand{\MAT}[3]{M_{#1\ifthenelse{\equal{#2}{}}{}{,#2}}\ifthenelse{\equal{#3}{}}{}{\left(#3\right)}}
\newcommand{\Rgen}{\R_{G}}
\newcommand{\Rpoly}{\R_{P}}
\newcommand{\Rp}{\R_{+}}
\newcommand{\Rs}{\R^{*}}
\newcommand{\Rps}{\R_{+}^{*}}
\newcommand{\intinterv}[2]{\llbracket#1,#2\rrbracket}
\newcommand{\powerset}[1]{\mathcal{P}(#1)}
\newcommand{\dom}{\operatorname{dom}}
\newcommand{\card}[1]{{\##1}}
\newcommand{\PTIME}{\ensuremath{\operatorname{P}}}

\newcommand{\NP}{\ensuremath{\operatorname{NP}}}
\newcommand{\NPTIME}{\NP}
\newcommand{\FP}{\ensuremath{\operatorname{FP}}}
\newcommand{\PSPACE}{\ensuremath{\operatorname{PSPACE}}}
\newcommand{\inorm}[2]{\left\lVert{#1}\right\rVert_{#2}}

\newcommand{\infnorm}[1]{\inorm{#1}{}}
\newcommand{\sigmap}[1]{{\Sigma{#1}}}
\newcommand{\degp}[1]{{\operatorname{deg}(#1)}}
\newcommand{\poly}{\operatorname{poly}}
\newcommand{\sgn}[1]{\operatorname{sgn}(#1)}
\newcommand{\intp}{\operatorname{int}}
\newcommand{\fracp}{\operatorname{frac}}
\newcommand{\fiter}[2]{#1^{[#2]}}
\newcommand{\frestrict}[2]{{#1}_{\restriction_{#2}}}
\newcommand{\indicator}[1]{\mathds{1}_{#1}}
\newcommand{\idfun}{\operatorname{id}}
\newcommand{\ceil}[1]{\left\lceil#1\right\rceil}
\newcommand{\round}[1]{\left\lfloor#1\right\rceil}

\newcommand{\myclass}[1]{\operatorname{#1}}
\newcommand{\pcab}[2]{\ensuremath{\PTIME_{C[#1,#2]}}}

\newcommand\myOmega{\amalg}

\COMMENTESAUFSILATEXDIFF{\newcommand{\gval}[2][]{\ensuremath{\myclass{GVAL}_{#1}\ifthenelse{\equal{#2}{}}{}{[#2]}}}}
\newcommand{\gpval}[1][]{\ensuremath{\myclass{GPVAL}_{#1}}}
\newcommand{\gc}[3][]{\ensuremath{\myclass{ATSC}_{#1}(#2,#3)}}
\newcommand{\mygpc}[1][]{\ensuremath{\myclass{ALP}_{#1}}}
\newcommand{\gpc}[1][]{\ensuremath{\myclass{ATSP}_{#1}}}
\COMMENTESAUFSILATEXDIFF{}
\newcommand{\gwc}[3][]{\ensuremath{\myclass{AWC}_{#1}(#2,#3)}}
 \newcommand{\gpwc}[1][]{\ensuremath{\myclass{AWP}_{#1}}}

\COMMENTESAUFSILATEXDIFF{ }
\COMMENTESAUFSILATEXDIFF{ }
\newcommand{\goc}[3]{\ensuremath{\myclass{AOC}(#1,#2,#3)}}
\COMMENTESAUFSILATEXDIFF{ }
\COMMENTESAUFSILATEXDIFF{ }
\newcommand{\guc}[4]{\ensuremath{\myclass{AXC}(#1,#2,#3,#4)}}
\COMMENTESAUFSILATEXDIFF{ }
\COMMENTESAUFSILATEXDIFF{ }
 \newcommand{\glc}[2][]{\ensuremath{\myclass{ALC}_{#1}(#2)}}
 \newcommand{\gplc}[1][]{\ensuremath{\myclass{ALP}_{#1}}}
 \newcommand{\cglc}[1][]{\ensuremath{\myclass{ALC}_{#1}}}

\newcommand{\unaware}{extreme}

\newcommand{\argmin}{\operatornamewithlimits{argmin}}

\newcommand{\jacobian}[1]{J_{#1}}
\newcommand{\grad}[1]{\nabla{#1}}
\newcommand{\scalarprod}[2]{{#1}\cdot{#2}}
\newcommand{\bigO}[1]{\mathcal{O}\left(#1\right)}
\newcommand{\softO}[1]{\tilde{\mathcal{O}}\left(#1\right)}
\newcommand{\taylor}[3]{{T_{#1}^{#2}#3}}
\newcommand{\transpose}[1]{{#1}^T}
\newcommand{\pastsup}[2]{{\sup}_{#1}#2}
\newcommand{\mtt}[1]{\mathtt{#1}}
\newcommand{\ovl}[1]{\overline{#1}}

\newcommand{\PIVP}{{PIVP}}

\newcommand{\myop}[1]{\operatorname{#1}}

\COMMENTESAUFSILATEXDIFF{}
\COMMENTESAUFSILATEXDIFF{}
\COMMENTESAUFSILATEXDIFF{}

\newcommand{\lxh}{\myop{lxh}}
\newcommand{\hxl}{\myop{hxl}}

\newcommand{\norm}{\myop{norm}}

\newcommand{\sample}{\myop{sample}}

\newcommand{\crnd}{\myop{rnd}^*}

\newcommand{\glen}[1]{\myop{len}_{#1}}
\newcommand{\mix}[3]{\myop{mix}(#1,#2,#3)}
\newcommand{\lagrange}[1]{{\mathds{1}}_{#1}}
\newcommand{\lagreq}[2]{{\mathds{D}}_{#1=#2}}
\newcommand{\lagrneq}[2]{{\mathds{D}}_{#1\neq #2}}
\newcommand{\emptyword}{\lambda}

\newcommand{\machcfg}[1]{\mathcal{C}_{#1}}
\newcommand{\machstep}[1]{#1}
\newcommand{\realenc}[1]{\left\langle#1\right\rangle}
\newcommand{\idealrealstep}[1]{\realenc{#1}_\infty}
\newcommand{\realstep}[1]{\realenc{#1}}

\newcommand{\functionextract}{\myop{extract}}
\newcommand{\functionrnd}{\myop{rnd}}

\newcommand{\LenI}{\operatorname{PsLen}}

\definecolor{darkgreen}{rgb}{0.1,0.6,0.1}
\colorlet{myyellow}{yellow!80!blue}

\allowdisplaybreaks

\newcommand\param[1]{{#1}}

\theoremstyle{acmdefinition}
\newtheorem{remark}[theorem]{Remark}

\def\INFORMATIONANDCOMPUTATION{2016arXiv160200546B}
\def\JOURNALOFCOMPLEXITY{BournezGP16b}

\def\PAPIERODETCS{PoulyG16}
\def\PAPIERODE{Pouly16}

\acmJournal{JACM}

\acmDOI{0000001.0000001}

\received{September 2016}
\received[revised]{May 2016}

\begin{document}

\title[Polynomial time corresponds to solutions of polynomial ODEs of polynomial length]
{Polynomial Time Corresponds to Solutions of Polynomial Ordinary
    Differential Equations of Polynomial Length \textcolor{red}{(SUBMITTED)}}

\author{Olivier Bournez}
\affiliation{%
  \institution{Ecole Polytechnique, LIX}
  \streetaddress{}
  \city{Palaiseau}
  \postcode{91128 Cedex}
  \country{France}}
\author{Daniel S. Gra\c{c}a}
\affiliation{%
  \institution{Universidade do Algarve}
  \country{Portugal}}
\affiliation{%
  \institution{Instituto de Telecomunica\c{c}\~{o}e}
  \country{Portugal}}
\author{Amaury Pouly}
\affiliation{%
  \institution{Ecole Polytechnique, LIX}
  \streetaddress{}
  \city{Palaiseau}
  \postcode{91128 Cedex}
  \country{France}}
\affiliation{%
  \institution{Department of Computer Science, University of Oxford}
    \streetaddress{Wolfson Building, Parks Rd}
    \city{Oxford}
    \postcode{OX1 3QD},
  \country{United Kingdom}}

\begin{abstract}
The outcomes of this paper are twofold.
\bigskip

\noindent\textbf{Implicit complexity.} 
    We provide an implicit characterization of polynomial time computation in terms of
    ordinary differential equations: we characterize  the class $\PTIME$
    of languages computable in polynomial time in terms
    of differential equations with polynomial right-hand side.
    This result gives a purely continuous elegant and
    simple characterization of $\PTIME$. We believe it is the first time complexity classes
    are characterized using only ordinary differential equations. Our
    characterization extends to functions computable in polynomial time
    over the reals in the sense of Computable Analysis. 

    Our results may provide a new perspective on classical complexity, by
    giving a way to define complexity classes, like $\PTIME$, in a very simple
    way, without any reference to a notion of (discrete) machine. This
    may also provide ways to state classical questions about computational complexity 
    via ordinary differential equations.

    \bigskip

    \noindent\textbf{Continuous-Time Models of Computation.} 
    Our results can also be interpreted in terms of analog computers or
    analog models of computation: As a side
    effect, we get that the 1941 General Purpose Analog Computer (GPAC) of
    Claude Shannon is provably equivalent to Turing machines both in terms of
    computability and complexity, a fact that has never been
    established before. This result provides arguments in favour of a generalised form of
    the Church-Turing Hypothesis, which states that any physically realistic (macroscopic)
    computer is equivalent to Turing machines both in terms of computability and complexity.
\end{abstract}


\thanks{Daniel Gra\c{c}a was partially supported by \emph{Funda\c{c}\~{a}o para a Ci\^{e}ncia e a Tecnologia}
    and EU FEDER POCTI/POCI via SQIG - Instituto de
    Telecomunica\c{c}\~{o}es through the FCT project
    UID/EEA/50008/2013.
   Olivier Bournez and Amaury Pouly were partially supported by
    \emph{DGA Project CALCULS} and \emph{French National Research
      Agency (ANR)  Project
      ANR-15-CE040-0016-01}.
}

\maketitle

\newpage
\tableofcontents
\newpage

\section{Introduction}

The current article is a journal extended version of our paper
presented at 43rd International Colloquium on Automata, Languages and
Programming  ICALP'2016 (Track B best paper
award). \\

The outcomes of this paper are twofold, and concern a priori
not closely related topics.

\subsection{Implicit Complexity}
Since the introduction of the $\PTIME$ and $\NPTIME$ complexity classes, much work
has been done to build a well-developed complexity theory based on Turing
Machines. In particular, classical computational complexity theory is based on
limiting resources used by Turing machines, such as time and space.
Another approach is implicit computational complexity. The term ``implicit'' in
this context can be understood in various ways, but a common point of
these characterizations is that they provide (Turing or equivalent)
machine-independent alternative definitions of classical complexity. 

Implicit complexity theory has gained enormous interest in the last decade. This has led
to many alternative characterizations of complexity classes using
recursive functions, function algebras, rewriting systems, neural
networks, lambda calculus and so on. 

However, most of --- if not all --- these models or
characterizations are essentially discrete: in particular they are
based on underlying discrete-time models working on objects which are
essentially discrete, such as words, terms, etc.

Models of computation working on a continuous space have also been
considered: they include Blum Shub Smale machines \cite{BCSS98},
Computable Analysis \cite{Wei00}, and quantum computers \cite{Fey82}
which usually feature discrete-time and
continuous-space. Machine-independent characterizations of the
corresponding complexity classes have also been devised: see
e.g. \cite{JLC04,GM95}. However, the resulting characterizations are
still essentially discrete, since time is still considered to be
discrete.

In this paper, we provide a purely analog machine-independent
characterization of the class $\PTIME$. Our characterization
relies only on a simple and natural class of
ordinary differential equations: $\PTIME$ is characterized using ordinary
differential equations (ODEs) with polynomial right-hand side. This shows
first that (classical) complexity theory can be presented in terms of
ordinary differential equations problems. This opens the way to state
classical questions, such as $\PTIME$ vs $\NPTIME$, as questions about
ordinary differential equations, assuming one can also express $\NPTIME$ this way.

\subsection{Analog Computers} 
Our results can also be interpreted in the context of
analog models of computation and actually originate as a side effect
of an attempt to understand the power of continuous-time analog models
relative to classical models of computation.
Refer to \cite{LivreAnalogcomputing} for a very
instructive historical account of the history of Analog computers. See
also \cite{maclennan2009analog,CIEChapter2007} for further discussions.


Indeed, in 1941, Claude Shannon introduced in \cite{Sha41} the
General Purpose Analog Computer  (GPAC) model as a
model for the Differential Analyzer \cite{Bus31}, a mechanical programmable
machine, on which he worked
as an operator.
The GPAC model was later refined in \cite{Pou74}, \cite{GC03}.
Originally it was presented as a model based on circuits {(see Figure~\ref{fig:gpac_circuit})},
where several units performing basic operations (e.g.~sums,
integration) are interconnected {(see Figure~\ref{fig:gpac_example_sin}
)}.

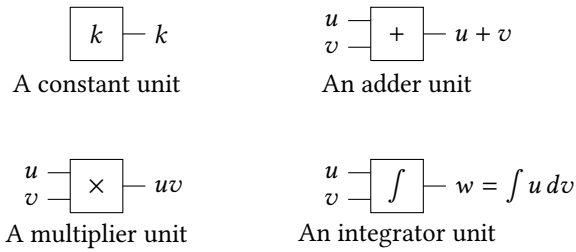
\begin{figure}[h]
\begin{center}
 \setlength{\unitlength}{1200sp}%
\begin{tikzpicture}
 \begin{scope}[shift={(0,0)},rotate=0]
  \draw (0,0) -- (0.7,0) -- (0.7,0.7) -- (0,0.7) -- (0,0);
  \node at (.35,.35) {$k$};
  \draw (.7,.35) -- (1,.35);
  \node[anchor=west] at (1,.35) {$k$};
  \node at (.35, -.3) {A constant unit};
 \end{scope}
 \begin{scope}[shift={(4,0)},rotate=0]
  \draw (0,0) -- (0.7,0) -- (0.7,0.7) -- (0,0.7) -- (0,0);
  \node at (.35,.35) {$+$};
  \draw (.7,.35) -- (1,.35); \draw (-.3,.175) -- (0,.175); \draw (-.3,.525) -- (0,.525);
  \node[anchor=west] at (1,.35) {$u+v$};
  \node at (.35, -.3) {An adder unit};
  \node[anchor=east] at (-.3,.525) {$u$};
  \node[anchor=east] at (-.3,.175) {$v$};
 \end{scope}
 \begin{scope}[shift={(0,-2)},rotate=0]
  \draw (0,0) -- (0.7,0) -- (0.7,0.7) -- (0,0.7) -- (0,0);
  \node at (.35,.35) {$\times$};
  \draw (.7,.35) -- (1,.35); \draw (-.3,.175) -- (0,.175); \draw (-.3,.525) -- (0,.525);
  \node[anchor=west] at (1,.35) {$uv$};
  \node at (.35, -.3) {A multiplier unit};
  \node[anchor=east] at (-.3,.525) {$u$};
  \node[anchor=east] at (-.3,.175) {$v$};
 \end{scope}
 \begin{scope}[shift={(4,-2)},rotate=0]
  \draw (0,0) -- (0.7,0) -- (0.7,0.7) -- (0,0.7) -- (0,0);
  \node at (.35,.35) {$\int$};
  \draw (.7,.35) -- (1,.35); \draw (-.3,.175) -- (0,.175); \draw (-.3,.525) -- (0,.525);
  \node[anchor=west] at (1,.35) {$w=\int u\thinspace dv$};
  \node at (.35, -.3) {An integrator unit};
  \node[anchor=east] at (-.3,.525) {$u$};
  \node[anchor=east] at (-.3,.175) {$v$};
 \end{scope}
\end{tikzpicture}
\end{center}
\caption{Circuit presentation of the GPAC: a circuit built from basic units}
\label{fig:gpac_circuit}
\end{figure}%

\begin{figure}[h]
\begin{center}
\begin{tikzpicture}
 \begin{scope}[shift={(-4.2,0)},rotate=0]
  \draw (0,0) -- (0.7,0) -- (0.7,0.7) -- (0,0.7) -- (0,0);
  \node at (.35,.35) {$-1$};
 \end{scope}
 \draw (-3.5,.35) -- (-3.15,.35) -- (-3.15,.175) -- (-2.8,.175);
 \begin{scope}[shift={(-2.8,0)},rotate=0]
  \draw (0,0) -- (0.7,0) -- (0.7,0.7) -- (0,0.7) -- (0,0);
  \node at (.35,.35) {$\times$};
 \end{scope}
 \draw (-2.1,.35) -- (-1.75,.35) -- (-1.75,.525) -- (-1.4,.525);
 \begin{scope}[shift={(-1.4,0)},rotate=0]
  \draw (0,0) -- (0.7,0) -- (0.7,0.7) -- (0,0.7) -- (0,0);
  \node at (.35,.35) {$\int$};
 \end{scope}
 \draw (-.7,.35) -- (-.35,.35) -- (-.35,.525) -- (0,.525);
 \begin{scope}[shift={(0,0)},rotate=0]
  \draw (0,0) -- (0.7,0) -- (0.7,0.7) -- (0,0.7) -- (0,0);
  \node at (.35,.35) {$\int$};
 \end{scope}
 \draw (.7,.35) -- (1.4,.35);
 \node[anchor=west] at (1.4,.35) {$\sin(t)$};
 \node[anchor=north] at (-1, -0.5) {$\left\lbrace
\begin{array}{@{}c@{}l}
y'(t)&=z(t)\\
z'(t)&=-y(t)\\
y(0)&=0\\
z(0)&=1
\end{array}
\right.\Rightarrow\left\lbrace\begin{array}{@{}c@{}l}
y(t)&=\sin(t)\\
z(t)&=\cos(t)
\end{array}\right.$};
 \draw (1,.35) -- (1,1) -- (-3.15,1) -- (-3.15,.525) -- (-2.8,.525);
 \fill (1,.35) circle[radius=.07];
 \draw (-4.9,.35) -- (-4.55,.35) -- (-4.55,-.3) -- (-.3,-.3) -- (-.3,.175) -- (0,.175);
 \draw (-1.75,-.3) -- (-1.75,.175) -- (-1.4,.175);
 \fill (-1.75,-.3) circle[radius=.07];
 \node[anchor=east] at (-4.9,.35) {$t$};
\end{tikzpicture}
\end{center}
\caption{Example of GPAC circuit: computing sine and cosine with two variables}
\label{fig:gpac_example_sin}
\end{figure}
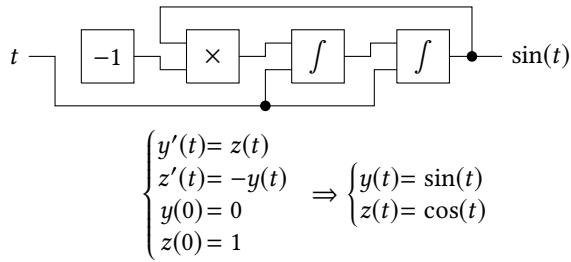

However, Shannon himself realized that functions computed by a GPAC are nothing more
than solutions of a special class of polynomial differential
equations. In particular it can be shown that a function is computed
by a GPAC if and only if it is a
(component of the) solution of a system of ordinary
differential equations (ODEs) with polynomial right-hand side \cite{Sha41}, \cite{GC03}.
In this paper, we consider the refined version presented in \cite{GC03}.

We note that the original notion of computation in the model of the GPAC presented in \cite{Sha41}, \cite{GC03}
is known  not to be equivalent to Turing machine based models, like Computable
Analysis. However, the original GPAC model only allows for functions in
one continuous variable and in real-time: at time $t$ the output is $f(t)$,
which is different from the notion used by Turing machines. In \cite{Gra04} a
new notion of computation for the GPAC, which uses ``converging computations''
as done by Turing machines was introduced and it was shown in \cite{BCGH06}, \cite{BCGH07}
that using this new notion of computation, the GPAC and Computable Analysis are two equivalent
models of computation, at the computability level.

Our paper extends this latter result and proves that the GPAC and Computable Analysis
are two equivalent models of computation, both in terms of computability and complexity.
We also provide as a side effect a robust way to
measure time in the GPAC, or more generally in computations performed by
ordinary differential equations: essentially by considering the length
of the solution curve. 

\section{Results and discussion}\label{sec:our_results}

\subsection{Our results}\label{subsec:our_results}

The first main result of this paper shows that the class \PTIME~can be characterized using ODEs. In particular this result uses the following class of differential
equations:
\begin{equation}\label{eq:gpac}
y(0)=y_0\qquad y'(t)=p(y(t))
\end{equation}
where $p$ is a vector of polynomials and $y: I \to \R^d$ for some interval $I \subset \R$.
Such systems are sometimes called \PIVP, for polynomial initial value
problems \cite{GBC09}. Observe that there is always a unique solution
 to the PIVP, which is analytic, defined on a maximum domain $I$ containing $0$, which we
refer to as ``the solution''.

To state complexity results via ODEs, we need to introduce some kind of complexity measure for ODEs and, more concretely, for PIVPs. This is a non-trivial task since, contrarily to discrete models of computation, continuous models of computation (not only the GPAC, but many others) usually exhibit the so-called ``Zeno phenomena'', where time can be arbitrarily contracted in a continuous system, thus allowing an arbitrary speed-up of the computation, if we take the naive approach of using the time variable of the ODE as a measure of ``time complexity'' (see Section \ref{subsec:continuous_models} for more details).

Our crucial and key idea to solve this problem is that, when using PIVPs (in principle this idea can also be used for others continuous models of computation) to compute a function $f$,
the cost should be measured as a function of the length of the solution curve of the PIVP computing the function $f$.
We recall that the length of a curve $y\in C^1(I,\R^n)$ defined over
some interval $I=[a,b]$ is given by
$\glen{y}(a,b)=\int_I\infnorm{y'(t)}dt,$
where $\infnorm{y}$ refers to the infinity norm of $y$. 

Since a language is made up of words, we need to discuss how to represent (encode) a word into a real number to decide a language with a PIVP.
We fix a finite alphabet $\Gamma=\{0, .., k-2\}$ 
and define the encoding\footnote{Other encodings may be
  used, however, two crucial properties are necessary: (i) $\psi(w)$
  must provide a way to recover the length of the word, (ii)
  $\infnorm{\psi(w)}\approx\poly(|w|)$ in other words, the norm of the
  encoding is roughly the length of the word. For technical reasons, we need to encode
  the number in basis one more than the number of symbols.}
$\psi(w)=\left(\sum_{i=1}^{|w|}w_ik^{-i},|w|\right)$ for a word
$w=w_1w_2\dots w_{|w|}$. We also take $\Rp=[0,+\infty[$.


\begin{definition}[Discrete recognizability]\label{def:discrete_rec_q}
A language $\mathcal{L}\subseteq\Gamma^*$  is called  
\emph{poly-length-analog-recognizable}  
if there
exists a vector $q$ of bivariate polynomials and a vector $p$ of polynomials with $d$ variables,
both with coefficients in $\Q$, 
and a polynomial $\myOmega: \Rp \to \Rp$,
such that for all $w\in\Gamma^*$,
there is a (unique) $y:\Rp\rightarrow\R^d$ such that for all $t\in\Rp$:
\begin{itemize}
\item $y(0)=q(\psi_k(w))$ and $y'(t)=p(y(t))$
\hfill$\blacktriangleright$ $y$ satisfies a differential equation
\item if $|y_1(t)|\geqslant1$ then 
$|y_1(u)|\geqslant1$ for all $u\geqslant t$
\hfill$\blacktriangleright$ decision is stable
\item if $w\in\mathcal{L}$ (resp. $\notin\mathcal{L}$) and $\glen{y}(0,t)\geqslant\myOmega(|w|)$
then $y_1(t)\geqslant1$ (resp. $\leqslant-1$)
\hfill$\blacktriangleright$ decision
\item $\glen{y}(0,t)\geqslant t$
\hfill$\blacktriangleright$ technical condition\footnote{This could be
  replaced by only assuming that we have somewhere the additional
  ordinary differential equation $y'_0=1$.}
\end{itemize}
\end{definition}

Intuitively (see Fig.~\ref{fig:analog_recognizability}) this definition says that a language is poly-length-analog-recognizable if
there is a PIVP such that, if the initial condition is set to be (the
encoding of) some word $w\in\Gamma^*$, then by using a \emph{polynomial length} 
portion of the curve, we are able to tell if this word should be accepted or rejected,
by watching to which region of the space the trajectory goes: the value of $y_1$ determines if
the word has been accepted or not, or if the computation is still in progress. See
Figure~\ref{fig:analog_recognizability} for a graphical representation of Definition~\ref{def:discrete_rec_q}.

\begin{theorem}[A characterization of $\PTIME$]\label{th:p_gpac_un}
A decision problem (language) $\mathcal{L}$  belongs to the class $\PTIME$
if and only if it is poly-length-analog-recognizable.
\end{theorem}

A slightly more precise version of this statement is given at the end of the paper,
in Theorem~\ref{th:p_gpac}.  A characterization of the class $\FP$ of polynomial-time
computable functions is also given in Theorem~\ref{th:fp_gpac}.

{
\begin{figure}
\begin{center}
\begin{tikzpicture}[domain=1:10.1,samples=300,scale=1]
\fill[pattern=north west lines, pattern color=darkgreen!25!white
] (1,1) rectangle (10.1,2.1);
\fill[pattern= north east lines, pattern color=yellow!25!white
] (1,-1) rectangle (10.1,-2.1);
\fill[pattern= grid,pattern color =blue!25!white
] (1,-1) rectangle (7,1);
\fill[pattern= crosshatch, pattern color=black!25!white
] (7,-1) rectangle (10.1,1);
\draw[very thin,color=gray] (0.9,-2.1) grid (10.1,2.1);
\draw[->] (1,-2.1) -- (1,2.2);
\draw[->] (0.9,0) -- (10.2,0) node[right] {
    $\displaystyle\begin{array}{@{}r@{}l@{}}
        &=\int_0^t\infnorm{y'}\\
        \ell(t)&=\text{length of }y\\
        &\hspace{1em}\text{over }[0,t]\end{array}$};
\draw[very thick] (1,1) -- (0.9,1) node[left] {$1$};
\draw[very thick] (1,-1) -- (0.9,-1) node[left] {$-1$};
\draw[very thick] (7,0.1) -- (7,-0.1) node[below=-.3em,xshift=1em] {$\scriptstyle\poly(|w|)$};
\draw[thick
] (4.5,1.5) node {accept: $w\in\mathcal{L}$};
\draw[thick
] (4.5,-1.5) node {reject: $w\notin\mathcal{L}$};
\draw[thick
] (4.5,-.5) node {computing};
\draw[thick
] (9,.5) node {forbidden};
\draw[color=blue,very thick] (1,-.52) -- (0.9,-.52) node[left=-.3em] {$\scriptstyle q(\psi(w))$};
\draw[color=black!75!white,domain=1:7,
        postaction={decorate,decoration={markings,mark=at position 1 with {
            \draw[red,thick] (-.3em,-.3em) -- (.3em,.3em); \draw[red,thick] (-.3em,.3em) -- (.3em,-.3em);}}}]
        plot[id=ar_timeout] function{-0.52-(1+tanh((x-2)*2))/7+(1+tanh((x-3.2)*5))/2+(1+tanh((x-5)*2))/7} node[right] {$\scriptstyle y_1(t)$};
\draw[color=black!75!white,domain=1:2.2,
        postaction={decorate,decoration={markings,mark=at position 1 with {
            \draw[red,thick] (-.3em,-.3em) -- (.3em,.3em); \draw[red,thick] (-.3em,.3em) -- (.3em,-.3em);}}}]
        plot[id=ar_stop] function{.75-(0.75+0.52)*exp((1-x)*5)} node[right] {$\scriptstyle y_1(t)$};
\draw[color=darkgreen] plot[id=ar_accept] function{(x/3+sin(10*x))/(1+exp(x-3))/2+3./2.*tanh(x/5-1./2.)} node[right=-.3em] {$\scriptstyle y_1(t)$};
\draw[color=red] plot[id=ar_reject] function{((-1.1-x/10-sin(5*x-1.5))/(1+exp(x-3))/2+1./2.*tanh(x/5-1./2.))*(1+tanh((5-x)*10.))/2.
    -(1-tanh((6-x)*10))/2.*3./2.*tanh(x/8.5)} node[right=-.3em] {$\scriptstyle y_1(t)$};
\end{tikzpicture}
\end{center}
\caption{Graphical representation of poly-length-analog-recognizability (Definition~\ref{def:discrete_rec_q}).
The green trajectory represents an accepting computation, the red a rejecting one, and the gray are invalid
computations. An invalid computation is a trajectory that is too slow (or converges) (thus violating the technical condition),
or that does not accept/reject in polynomial length. Note that we only represent the first
component of the solution, the other components can have arbitrary behaviors.}
\label{fig:analog_recognizability}
\end{figure}
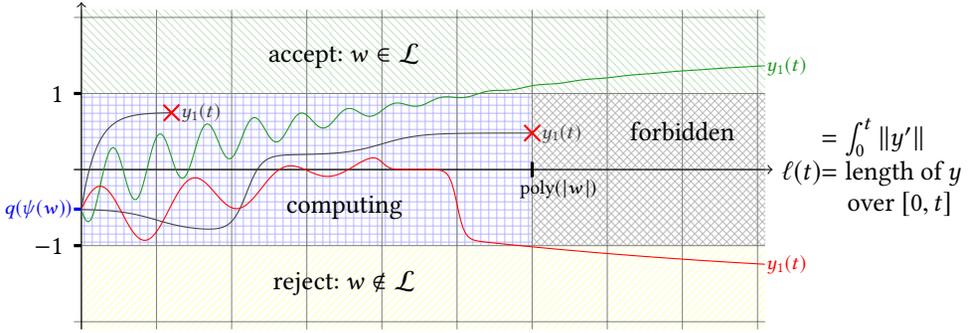
}

\begin{figure}
\centering
\begin{tikzpicture}[domain=1:11,samples=500,scale=1]
    \draw[very thin,color=gray] (0.9,-0.1) grid (11.1,4.1);
    \draw[->] (1,-0.1) -- (1,4.2);
    \draw[->] (0.9,0) -- (11.2,0) node[right] {$\glen{y}$};
    \draw[color=red,thick] plot[id=fn_1] function{
        1+exp((x-1)*(7-x)/10)+(1+sin(10*x))/(1+exp(x-3))-2*exp(-(x-1))
        };
    \draw[blue] (0.9,1) -- (11.1,1);
    \draw[right,blue] (11,1) node {$f(\textcolor{blue}{x})$};
    \draw[color=red,thick] (1.0,0.4) -- (0.9,0.4);
    \draw[red] (0.95,0.4) node[left] {$q_1(\textcolor{blue}{x})$};
    \draw[color=red] (3.3,2.7) node {$y_1$};
    \draw[very thick] (7.0,0) -- (7.0,-0.2);
    \draw[<->,darkgreen,thick] (7.0,1) -- (7,2) node[midway,black,left] {$e^{-\textcolor{darkgreen}{0}}$};
    \draw (7,-0.1) node[below] {$\myOmega(\textcolor{blue}{x},\textcolor{darkgreen}{0})$};
    \draw[<->,darkgreen,thick] (8.4,1) -- (8.4,1.36) node[midway,black,left] {$\scriptstyle e^{-\textcolor{darkgreen}{1}}$};
    \draw[dotted] (8.4,0) -- (8.4,1);
    \draw[very thick] (8.4,0) -- (8.4,-0.2);
    \draw (8.4,-0.1) node[below] {$\myOmega(\textcolor{blue}{x},\textcolor{darkgreen}{1})$};
\end{tikzpicture}
\caption{Poly-length-computability: on input $x$, starting from initial condition $q(x)$,
the PIVP $y'=p(y)$ ensures that $y_1(t)$ gives $f(x)$ with accuracy better than $e^{-\mu}$
as soon as the length of $y$ (from $0$ to $t$) is greater than $\myOmega(\infnorm{x},\mu)$.
Note that we did not plot the other variables $y_2,\ldots,y_d$ and the horizontal
axis measures the length of $y$ (instead of the time $t$).\label{fig:glc}}
\end{figure}
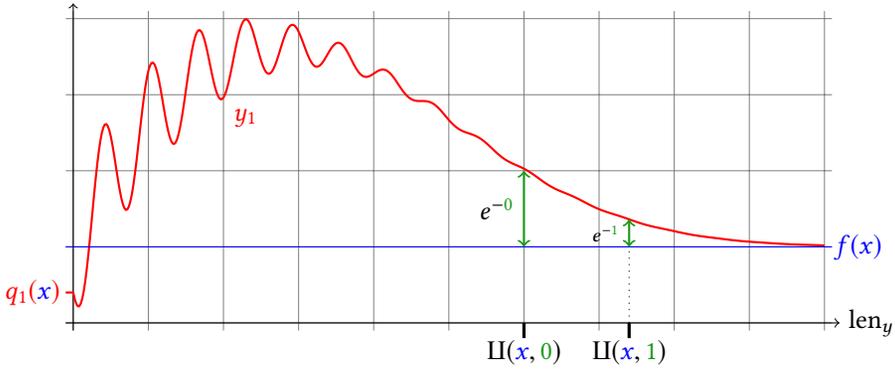

Concerning the second main result of this paper, we assume the reader is familiar with the notion of a polynomial-time
computable  function $f: [a,b] \to \R$ (see
\cite{Wei00} for an introduction to Computable Analysis). 
We denote by $\Rpoly$ the set of
polynomial-time computable reals.
For any vector $y$, $y_{i\ldots j}$ refers to the vector $(y_i, y_{i+1}, \ldots, y_j)$. For any sets $X$ and $Z$, $f:\subseteq X\rightarrow Z$
refers to any function $f:Y\rightarrow Z$ where $Y\subseteq X$ and $\dom{f}$ refers
to the domain of definition of $f$.

Our second main result is an analog characterization of polynomial-time computable real functions.
More precisely, we show that the class of poly-length-computable functions (defined below), when restricted to domains of the form $[a,b]$, is
the same as the class of polynomial-time computable real
functions of Computable Analysis over $[a,b]$, sometimes denoted by $\pcab{a}{b}$, as defined in \cite{Ko91}.
It is well-known that all computable functions (in the Computable Analysis setting) are continuous. Similarly,
all poly-length-computable functions (and more generally GPAC-computable functions) are continuous
(see Theorem~\ref{th:comp_implies_cont}).


\begin{definition}[Poly-Length-Computable Functions]\label{def:gplc}
We say that $f:\subseteq\R^n\rightarrow\R^m$ is poly-length-computable
 if and only if there exists
a vector $p$ of polynomials with $d\geqslant m$ variables and a
vector $q$ of polynomials with $n$ variables, both with coefficients
in $\Q$, and a bivariate polynomial $\myOmega$
such that for any $x\in\dom{f}$, there exists (a unique) $y:\Rp\rightarrow\R^d$ satisfying
for all $t\in\Rp$:
\begin{itemize}
\item $y(0)=q(x)$ and $y'(t)=p(y(t))$ \hfill$\blacktriangleright$ $y$ satisfies a PIVP
\item $\forall\mu\in\Rp$, if $\glen{y}(0,t)\geqslant \myOmega(\infnorm{x},\mu)$ then
$\infnorm{y_{1..m}(t)-f(x)}\leqslant e^{-\mu}$\hfill$\blacktriangleright$ $y_{1..m}$ converges to $f(x)$
\item $\glen{y}(0,t)\geqslant t$
\hfill$\blacktriangleright$ technical condition: the length grows at least linearly with time\footnote{This is a
    technical condition required for the proof. This can be weakened,
    for example to $\infnorm{y'(t)}=\infnorm{p(y(t))}\geqslant\frac{1}{\poly(t)}$. The
    technical issue is that if the speed of the system
    becomes extremely small, it might take an exponential time to reach a
    polynomial length, and we want to avoid such ``unnatural'' 
    cases. This is satisfied by all examples of computations we know \cite{LivreAnalogcomputing}.
    It also avoids pathological cases where the system would ``stop'' (i.e. converge) before accepting/rejecting,
    as depicted in
    Figure~\ref{fig:analog_recognizability}.}\footnote{This could also be
  replaced by only assuming that we have somewhere the additional
  ordinary differential equation $y'_0=1$.}
\end{itemize}
\end{definition}

Intuitively, a function f is poly-length-computable if there is a PIVP
that approximates f with a polynomial length to reach a given level of approximation.
See Figure~\ref{fig:glc} for a graphical representation of Definition~\ref{def:gplc}
and Section~\ref{sec:computable} 
for more background on analog computable functions.

\begin{theorem}[Equivalence with Computable Analysis]\label{th:eq_comp_analysis_un}
For any $a,b\in\Rpoly$ and $f\in C^0([a,b],\R)$, $f$ is polynomial-time computable
if and only if it is poly-length-computable. 
\end{theorem}

A slightly more precise version of this statement is given at the end of the paper,
in Theorem~\ref{th:eq_comp_analysis}.

\subsection{Applications to computational complexity}

We believe these characterizations to open  a new perspective on classical complexity,
as we indeed provide a natural definition (through previous definitions) of $\PTIME$ for
decision problems and of polynomial time for functions over the reals
using analysis only i.e.~ordinary differential equations and
polynomials, no need to talk about any (discrete) machinery like
Turing machines. 
This may open ways to characterize other
complexity classes like $\NP$ or $\PSPACE$. In the current settings of
course $\NP$ \label{dis:np} can be viewed as an existential quantification over our
definition, but we are obviously talking about ``natural''
characterizations, not involving unnatural quantifiers (for e.g.  a
concept of analysis like ordinary differential inclusions).

As a side effect, we also establish that solving
ordinary differential equations with polynomial right-hand side leads
to $\PTIME$-complete problems, when the length of the solution curve is taken into
account. In an less formal way, this is stating that ordinary
differential equations can be solved by following
the solution curve (as most numerical analysis method do), but that for general (and even right-hand side polynomial)
ODEs, no better method can work. Note that our results only deal with ODEs with a polynomial
right-hand side and that we do not know what happens for ODEs with analytic right-hand
sides over unbounded domains. There are some results (see e.g.~\cite{MM93}) which show
that ODEs with analytic right-hand sides can be computed locally in polynomial time.
However these results do not apply to our setting since we need to compute the solution
of ODEs over arbitrary large domains, and not only locally.


\subsection{Applications to continuous-time analog models}\label{subsec:continuous_models}

\PIVP s are known to correspond to
functions that can be generated by the GPAC of Claude Shannon \cite{Sha41}, which is itself a model of the analog computers (differential analyzers) in use in the first half of the XXth century \cite{Bus31}. 

As we have mentioned previously, defining a robust (time) complexity notion for continuous time systems was a
well known open problem \cite{CIEChapter2007} with no
generic solution provided to this day. In short, the difficulty is that the naive idea
of using the time variable of the ODE as a measure of ``time complexity'' is problematic,
since time can be arbitrarily contracted in a continuous system due to the ``Zeno
phenomena''. For example, consider a continuous system defined by an ODE
\[
y^{\prime}=f(y)
\]
where $f:\R\to\R$ and with solution $\theta:\R\to\R$. Now consider the following system
\[
\left\{
\begin{array}
[c]{l}%
y^{\prime}=f(y)z\\
z^{\prime}=z
\end{array}
\right.
\]
with solution $\phi:\R^2\to\R^2$. It is not difficult to see that this systems re-scales the time variable and that its solution $\phi=(\phi_1,\phi_2)$ is given by $\phi_2(t)=e^t$ and $\phi_1(t)=\theta(e^t)$ (see Figure \ref{fig:zeno}). Therefore, the second ODE simulates the first ODE, with an exponential acceleration. In a similar manner, it is also possible to present an ODE which has a solution with a component $\varphi_1:\R\to\R$ such that $\varphi_1(t) = \phi(\tan t)$, i.e.~it is possible to contract the
whole real line into a bounded set. Thus any language computable by the first system (or, in general, by a continuous system) can be computed by another continuous system in time $O(1)$. This problem appears not only for PIVPs (or, equivalently, GPACs), but also for many continuous models (see e.g.~ \cite{Ruo93},
\cite{Ruo94}, \cite{Moo95b}, \cite{Bournez97}, \cite{Bou99},
\cite{AD91}, \cite{CP01}, \cite{Davies01}, \cite{Copeland98},
\cite{Cop02}).

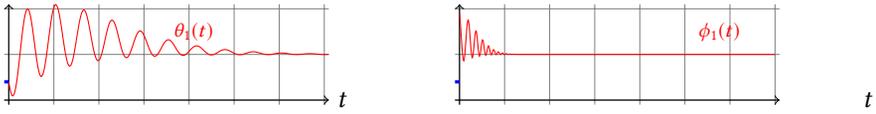
\begin{figure}
\begin{center}
\begin{tikzpicture}[domain=1:8.1,samples=300,scale=0.6]
\begin{scope}[shift={(-5,0)}]
        \draw[very thin,color=gray] (0.9,-0.1) grid (8.1,2.1);
        \draw[->] (1,-0.1) -- (1,2.1);
        \draw[->] (0.9,0) -- (8.1,0) node[right] {$t$};
        \draw[color=red] plot[id=zeno_before] function{1-exp(1-x)+(1+sin(10*x))/(1+exp(x-3))};
        \draw[color=blue,very thick] (0.9,0.4) -- (1,0.4);
        \draw[color=red] (5.1,1.5) node {$\scriptstyle \theta_1(t)$};
  \end{scope}
  \hspace{10mm}
  \begin{scope}[shift={(5,0)}]
        \draw[very thin,color=gray] (0.9,-0.1) grid (8.1,2.1);
        \draw[->] (1,-0.1) -- (1,2.1);
        \draw[->] (0.9,0) -- (8.1,0) node[right] {$t$};
        \draw[color=red,domain=1:5] plot[id=zeno_after_1] function{1-exp(1-exp(x))+(1+sin(10*exp(x)))/(1+exp(exp(x)-3))};
        \draw[color=red,domain=5:8] plot[id=zeno_after_2] function{1};
        \draw[color=blue,very thick] (0.9,0.4) -- (1,0.4);
        \draw[color=red] (5.1,1.5) node {$\scriptstyle \phi_1(t)$};
  \end{scope}
        \end{tikzpicture}
\end{center}
\caption{A continuous system before and after an exponential speed-up.}
\label{fig:zeno}
\end{figure}

With that respect, we solve this open problem by stating that the ``time complexity''
should be measured by the length of the solution curve of the
ODE. Doing so, we get a robust notion of time complexity for PIVP
systems. 
 Indeed, the length is a geometric property of the
 curve and is thus ``invariant'' by rescaling. 
 

Using this notion of complexity, we are then able to show that functions computable by a GPAC in polynomial time are exactly the functions computable in polynomial time in the sense of Computable Analysis (see Section \ref{subsec:our_results}). It was already previously shown in \cite{BCGH06}, \cite{BCGH07} that functions computable by a GPAC are exactly those computable in the sense of Computable Analysis. However this result was only pertinent for computability. Here we show that this equivalence holds also at a computational complexity level.

Stated otherwise, analog computers (as used before the advent of the digital computer) are theoretically equivalent to (and not more powerful than) Turing machine based models, both at a computability and complexity level. Note that this is a new result since, although digital computers are usually more powerful than analog computers at our current technological stage, it was not known what happened at a theoretical level.

This result leave us to conjecture the following generalization of the Church-Turing thesis: any physically realistic (macroscopic) computer is equivalent to Turing machines both in terms of computability and computational complexity.

\subsection{Applications to algorithms} We believe that
transferring the notion of time complexity to a simple consideration
about length of curves allows for very elegant and nice proofs of
polynomiality of many methods for solving both continuous and discrete
problems. For example, the zero of a function $f$ can easily be computed
by considering the solution of $y'=- f(y)$ under reasonable hypotheses
on $f$. More interestingly, this may also cover many interior-point
methods or barrier methods where the problem can be transformed into
the optimization of some continuous function (see e.g.
\cite{karmarkar1984new,Ref6-BFFS03,BFFS03,kojima1991unified}).

\subsection{Related work}

We believe that no purely continuous-time definition of $\PTIME$ 
has ever been stated before.
One direction of our characterization is based on a polynomial-time algorithm (in the length of the curve)
to solve PIVPs over unbounded time domains, and strengthens all existing results on the complexity
of solving ODEs over unbounded time domains. In the
converse direction, our proof requires a way to simulate a Turing
machine using \PIVP{} systems of polynomial length, a task whose
difficulty is discussed below, and still something that has never been
done up to date.

Attempts to derive a complexity theory for continuous-time systems
include \cite{GM02}. However, the theory developed there is not
intended to cover generic dynamical systems but only specific systems
that are related to Lyapunov theory for dynamical systems. The global
minimizers of particular energy functions are supposed to give
solutions of the problem. The structure of such energy functions leads
to the introduction of problem classes $U$ and $NU$, with the
existence of complete problems for theses classes. 

Another attempt is \cite{BSF02}, which also focused on a very specific type
of systems: dissipative flow models. The proposed theory is nice but
non-generic. This theory has been used in several papers from the same
authors to study a particular class of flow dynamics \cite{BFFS03} for
solving linear programming problems.

Neither of the previous two approaches is intended to cover generic ODEs, and
none of them is able to relate the obtained classes to classical
classes from computational complexity.

To the best of our knowledge, the most up to date surveys about continuous time
computation are \cite{CIEChapter2007,maclennan2009analog}. 

Relating computational complexity problems (such as the $\PTIME$ vs $\NP$
question) to problems of analysis has already been the motivation of other papers.  In particular, Félix Costa and Jerzy
Mycka have a series of work (see e.g. \cite{MC06}) relating the $\PTIME$ vs
$\NPTIME$ question to questions in the context of real and complex
analysis. 
Their approach is very different: they do so at the price of introducing a whole
hierarchy of functions and operators over functions. In particular,
they can use multiple times an operator which solves ordinary
differential equations before defining an element of $DAnalog$ and
$NAnalog$ (the counterparts of $\PTIME$ and $\NPTIME$ introduced in their
paper), while in our case we do not need the multiple application of
this kind of operator: we only need to use \emph{one} application of
such an operator (i.e.~we only need to solve one ordinary differential
equations with polynomial right-hand side).

It its true that one can sometimes convert the
multiple use of operators solving ordinary differential equations into
a single application \cite{GC03}, but this happens only in very
specific cases, which do not seem to include the classes $DAnalog$ and
$NAnalog$. In particular, the application of nested continuous
recursion (i.e. nested use of solving ordinary differential equations)
may be needed using their constructions, whereas we define $\PTIME$ using only a
simple notion of acceptance and only \emph{one}
system of ordinary differential equations. 

We also mention that Friedman and Ko (see
\cite{Ko91}) proved that  polynomial time computable functions
are closed under maximization and integration if and only if some open
problems of computational complexity (like $\PTIME=\NP$ for the
maximization case) hold. The complexity of solving
Lipschitz continuous ordinary differential equations has been proved to
be polynomial-space complete by Kawamura \cite{Kaw10}.

This paper contains mainly original contributions.
We however make references to results established in:
\begin{enumerate}
\item \cite{\INFORMATIONANDCOMPUTATION}, under revision for
  publication in \emph{Information and
  Computation}, devoted to properties of generable functions.
\item \cite{\JOURNALOFCOMPLEXITY}, published in \emph{Journal of Complexity},
  devoted to the proof of Proposition~\ref{th:main_eq}. 
\item \cite{\PAPIERODETCS}, published in \emph{Theoretical Computer Science},
  devoted to providing a polynomial time complexity algorithm for
  solving polynomially bounded polynomial ordinary differential
  equations. 
\end{enumerate}

None of these papers establishes relations between
polynomial-length-analog-computable-functions and classical
computability/complexity. This is precisely the objective of the
current article.

\subsection{Organization of the remainder of the paper}

In Section~\ref{sec:generablecomputable}, we introduce generable
functions and computable functions.  Generable functions are functions
computable by PIVPs (GPACs) in the classical sense of
\cite{Sha41}. They will be a crucial tool used in the paper to simplify
the construction of polynomial differential equations. Computable functions were introduced in
\cite{\JOURNALOFCOMPLEXITY}.  This section does not contain any new
original result, but only recalls already known results about these
classes of functions. 

Section~\ref{sec:original} establishes some original preliminary
results needed in the rest of the paper: First we relate generable
functions to computable functions under some basic conditions about
their domain. Then  we show that the class of computable functions 
is closed under arithmetic operations
and composition. We then provide several growth and continuity
properties. We then prove that absolute value, min, max, and some
rounding functions, norm, and bump function  are computable.

In Section~\ref{sec:pivp_turing:mt}, we show how to efficiently encode
the step function of Turing
machines using a computable function. 

In Section~\ref{sec:pivp_turing:equiv}, we provide a characterization
of $\FP$. To obtain this characterization, the idea is basically to iterate the functions of the previous
section using ordinary differential equations in one direction, and to
use a numerical method for solving polynomial ordinary differential
equations in the reverse direction.

In Section~\ref{sec:ptime}, we provide a characterization of
$\PTIME$. 

In Section~\ref{sec:ComputableAnalysis}, we provide a characterization
of polynomial time computable functions over the real in the sense of
Computable Analysis.

On purpose, to help readability of the main arguments of the proof, we postpone the most technical
proofs to Section~\ref{sec:proofs}. This latter section is
devoted to proofs of some of the results used in order to establish
previous characterizations.  

Up to Section~\ref{sec:avant:rationalelimination}, we allow
coefficients that maybe possibly non-rational numbers. 
In Section~\ref{sec:rationalelimination}, we prove that all non
-rationnal
coefficients  can be eliminated. This proves our main results stated
using only rational coefficients.

A list of notations used in this paper as well as in the above mentioned
related papers 
can be found in Appendix~\ref{sec:notations}.

\section{Generable and Computable Functions}
\label{sec:generablecomputable}

In this section we define the main classes of functions considered in
this paper and state some of their properties.  Results and
definitions from this section have already been obtained in other articles: They are taken from
\cite{\INFORMATIONANDCOMPUTATION},\cite{\JOURNALOFCOMPLEXITY}. The
material of this section is however needed for what follows.

\subsection{Generable functions}\label{sec:generable}

The following concept can be attributed to \cite{Sha41}: a function $f: \R \to \R$
is said to be a \PIVP{} function if there exists a system of the form \eqref{eq:gpac}
with $f(t)=y_1(t)$ for all $t$, where $y_1$
denotes the first component of the vector $y$ defined in $\R^d$. In our proofs, we needed
to extend Shannon's notion to talk about (i)
multivariable functions and (ii) the growth of these functions. To this end,
we introduced an extended class of generable functions in \cite{BournezGP16a}.

We will basically be interested with the case $\K=\Q$ in following
definition. However, for reasons explained in a few lines, we will need to
consider larger fields $\K$.

\begin{definition}[Polynomially bounded generable
  function]\label{def:gpac_generable_ext}
Let $\K$ be a field.
Let $I$ be an open and connected subset of $\R^d$
and $f:I\rightarrow\R^e$. We say that $f\in\gpval{}_\K$ if and only if
there exists a polynomial  $\mtt{sp}:\R\rightarrow\Rp$, 
$n\geqslant e$, a $n\times d$ matrix $p$ consisting of polynomials with coefficients in $\K$
, $x_0\in\K^d \cap I$, $y_0\in\K^n$
and $y:I\rightarrow\R^n$ satisfying for all $x\in I$:
\begin{itemize}
\item $y(x_0)=y_0$ and $\jacobian{y}(x)=p(y(x))$ 
  \hfill$\blacktriangleright$ $y$ satisfies a differential
  equation\footnote{$\jacobian{y}$ denotes the Jacobian matrix of $y$.}
\item $f(x)=y_{1..e}(x)$\hfill$\blacktriangleright$ $f$ is a component of $y$
\item $\infnorm{y(x)}\leqslant
  \mtt{sp}(\infnorm{x})$\hfill$\blacktriangleright$ $y$ is
  polynomially bounded
\end{itemize}


\end{definition}

This class can be seen as an extended version of PIVPs. Indeed, when $I$ is an interval,
the Jacobian of $y$ simply becomes the derivative of $y$ and we get the solutions
of $y'=p(y)$ where $p$ is a vector of polynomials.

Note that,  although functions in $\gpval{}_\K$ can be viewed as
solutions of partial differential equations (PDEs) (as we use a
Jacobian), 
we will never have to deal with classical
problems related to PDEs: PDEs have no general theory about the
existence of solutions, etc. This comes  from the way how we
define functions in $\gpval{}_\K$. Namely, in this paper, we will explictly present the functions in
$\gpval{}_\K$ which we will be used and we will show that they satisfy the conditions of Definition \ref{def:gpac_generable_ext}. Note also that it can be shown \cite[Remark 15]{2016arXiv160200546B} that a solution to the PDE defined with the Jacobian is unique, because the condition $\jacobian{y}(x)=p(y(x))$ is not general enough to capture the
class of all PDEs. We also remark that, because a function in $\gpval{}_\K$ must be polynomially bounded,
it is defined everywhere on $I$.

A much more detailed discussion
of this extension (which includes the results stated in this section) can be found in \cite{BournezGP16a}. The key property of this extension
is that it yields a much more stable class of functions than the
original class considered in \cite{Sha41}.  In particular, we can add,
subtract, multiply 
generable functions, and we can even do so while
keeping them polynomially bounded.

\begin{lemma}[Closure properties of $\gpval$
] \label{lem:closure}  Let
$(f:\subseteq\R^d\rightarrow\R^n),(g:\subseteq\R^e\rightarrow\R^m)\in\gpval{}_\K$. Then\footnote{For matching dimensions of course.}
$f+g$, $f-g$, $fg$ are in $\gpval{}_\K$.
\end{lemma}

As we said, we are basically mostly interested by the case $\K=\Q$,  but unfortunately, it turns out that $\gpval{}_\Q$  is not closed by
 composition\footnote{To observe that $\gpval{}_\Q$  is not closed by
 composition, see for example that $\pi$ is not rational and hence the
   constant function $\pi$
   does not belong to $\gpval{}_\Q$. However it can be obtained from 
 $\pi = 4 \arctan 1$.}, while $\gpval{}_\K$ is closed by composition for particular
fields $\K$: An interesting case is when $\K$ is supposed to be a \emph{generable field}
 as introduced in \cite{BournezGP16a}. All the reader needs to know about generable fields
is that they are fields and are stable by generable functions
(introduced in Section~\ref{sec:generable}). More precisely, 

\begin{proposition}[Generable field stability
] \label{coro55deinfoandcomputation}
Let $\K$ be a generable field. 
If $\alpha\in\K$ and $f$
is generable using coefficients in $\K$ (i.e. $f\in\gpval[\K]$) then
$f(\alpha)\in\K$.
\end{proposition}

It is shown in \cite{BournezGP16a} that there exists a smallest
generable field $\Rgen$ 
lying somewhere between $\Q$ and $\Rpoly$.

\begin{lemma}[Closure properties of $\gpval$
] \label{lem:closured}  
Let $\K$  be a generable field. 
Let
$(f:\subseteq\R^d\rightarrow\R^n),(g:\subseteq\R^e\rightarrow\R^m)\in\gpval{}_\K$. Then\footnote{For matching dimensions of course.}
$f\circ g$
in $\gpval{}_\K$.
\end{lemma}

As dealing with a class of functions closed by composition helps in
many constructions, we will first reason assuming that 
$\K$ is a generable field with $\Rgen\subseteq\K\subseteq\Rpoly$: From now on, $\K$ always denotes
such a generable field, and we write  $\gpval$ for $\gpval{}_\K$. We
will later prove  that non-rational coefficients can be eliminated in
order to
come back to the case $\K=\Q$. Up to  Section
\ref{sec:avant:rationalelimination} we allow coefficients in $\K$. Section
\ref{sec:rationalelimination}  is devoted to prove than their can then
be eliminated. 
 
As
$\Rpoly$ is generable, if this helps, the reader can consider that $\K=\Rpoly$ without any
significant loss of generality.

Another crucial property of class $\gpval$ is that it is closed
under solutions of ODE.  In practice, this means that we can write
differential equations of the form $y'=g(y)$ where $g$ is generable,
knowing that this can always be rewritten as a PIVP.

\begin{lemma}[Closure by ODE of $\gpval$
]\label{lem:gpac_ext_ivp_stable}
Let $J\subseteq\R$ be an interval,
$f:\subseteq\R^d\rightarrow\R^d$ in $\gpval$, $t_0\in\Q\cap J$ and $y_0\in\Q^d\cap\dom{f}$.
Assume there exists $y:J\rightarrow\dom{f}$ and a polynomial
$\mtt{sp}:\Rp\rightarrow\Rp$ satisfying for all $t\in J$:
\[y(t_0)=y_0\qquad y'(t)=f(y(t))
\qquad\infnorm{y(t)}\leqslant\mtt{sp}(t)\]
Then $y$ is unique and belongs to $\gpval$.
\end{lemma}

The class $\gpval$  contains many classic polynomially bounded analytic\footnote{Functions from $\gpval$ are necessarily
analytic, as solutions of an analytic ODE are analytic.} functions.
For example, all polynomials belong to $\gpval$, as well as sine and cosine.
Mostly notably, the hyperbolic tangent ($\tanh$) also belongs to $\gpval$. This function
appears very often in our constructions. Lemmas~\ref{lem:closure} and~\ref{lem:gpac_ext_ivp_stable}
are very useful to build new generable functions.


Functions from $\gpval$ are also known to have a polynomial modulus of
continuity. 

\begin{proposition}[Modulus of
  continuity
\label{prop:generable_mod_cont}
]
Let $f\in\gpval$ with corresponding polynomial 
$\mtt{sp}:\Rp\rightarrow\Rp$. There exists $q\in\K[\R]$ such that
for any $x_1,x_2\in\dom{f}$, if $[x_1,x_2]\subseteq\dom{f}$ then
    $\infnorm{f(x_1)-f(x_2)}\leqslant\infnorm{x_1-x_2}q(\mtt{sp}(\max(\infnorm{x_1},\infnorm{x_2})))$.
In particular, if $\dom{f}$ is convex then $f$ has a polynomial modulus of continuity.
\end{proposition}

\subsection{Computable functions}\label{sec:computable}

In \cite{\JOURNALOFCOMPLEXITY}, we introduced several notions of computation based on
polynomial differential equations extending the one introduced
by \cite{BCGH07} by adding a measure of complexity. The idea, illustrated in
Figure~\ref{fig:glc} is to put the input value $x$ as part of the initial condition
of the system and to look at the asymptotic behavior of the system. 

Our key insight to have a proper notion of
complexity is to measure the \emph{length} of the curve, instead of the time.
Alternatively, a proper notion of complexity is achieved by considering both time
\emph{and} space, where space is defined as the maximum value of all components of
the system.

Earlier attempts
at defining a notion of complexity for the GPAC based on other notions
failed because of time-scaling. Indeed,
given a solution $y$ of a PIVP, the function $z=y\circ\exp$ is also
solution of a 
PIVP, but converges exponentially faster. A longer discussion on this topic can be found in \cite{\JOURNALOFCOMPLEXITY}.
In this section, we recall the main complexity classes and restate the main
equivalence theorem. We denote by $\K[\A^n]$ the set of polynomial
functions with $n$ variables, coefficients in $\K$ and domain of
definition $\A^n$.

The following definition is a generalization (to general length bound
$\myOmega$ and field $\K$) of Definition
\ref{def:gplc}: Following class $\gplc$ when $\K=\Q$,
i.e. $\gplc[\Q]$, 
corresponds of course to poly-length-computable functions (Definition
\ref{def:gplc}).

\begin{definition}[Analog Length Computability]\label{def:glc}
Let $f:\subseteq\R^n\rightarrow\R^m$ and $\myOmega:\Rp^2\rightarrow\Rp$.
We say that $f$ is $\myOmega$-length-computable if and only if there exist $d\in\N$,
and $p\in\K^d[\R^d],q\in\K^d[\R^n]$
such that for any $x\in\dom{f}$, there exists (a unique) $y:\Rp\rightarrow\R^d$ satisfying
for all $t\in\Rp$:
\begin{itemize}
\item $y(0)=q(x)$ and $y'(t)=p(y(t))$\hfill$\blacktriangleright$ $y$ satisfies a PIVP
\item for any $\mu\in\Rp$, if $\glen{y}(0,t)\geqslant\myOmega(\infnorm{x},\mu)$
    then $\infnorm{y_{1..m}(t)-f(x)}\leqslant e^{-\mu}$\\\hphantom{a}\hfill $\blacktriangleright$ $y_{1..m}$ converges to $f(x)$
\item $\infnorm{y'(t)}\geqslant1$\hfill$\blacktriangleright$ technical condition: the
  length grows at least linearly with time\footnote{This is a
    technical condition required for the proof. This can be weakened,
    for example to $\infnorm{p(y(t))}\geqslant\frac{1}{\poly(t)}$. The
    technical issue is that if the speed of the system
becomes extremely small, it might take an exponential time to reach a
polynomial length, and we want to avoid such ``unnatural''
cases.}\footnote{This could be
  replaced by only assuming that we have somewhere the additional
  ordinary differential equation $y'_0=1$.}
\end{itemize}
We denote by $\glc{\myOmega}$ the set of $\myOmega$-length-computable
functions, and by $\gplc$ the set of $\myOmega$-length-computable
functions where $\myOmega$ is a polynomial, and more generally by $\cglc$ the
length-computable functions (for some $\myOmega$). If we want to explicitly
mention the set $\K$ of the coefficients, we write
$\glc[\K]{\myOmega}$, $\gplc[\K]$ and $\cglc[\K]$.
\end{definition}

This notion of computation turns out to be
equivalent to various other notions: The following equivalence result
is proved in \cite{\JOURNALOFCOMPLEXITY}. 

\begin{proposition}[Main equivalence,
  \cite{\JOURNALOFCOMPLEXITY}] \label{prop:mainequivalence} \label{th:main_eq}
Let $f:\subseteq\R^n\rightarrow\R^m$. Then the following are equivalent for
any generable field $\K$:
\begin{enumerate}
\item  (illustrated by  Figure~\ref{fig:glc}) $f\in\gplc$ ;
\item (illustrated by  Figure~\ref{fig:gc})  There exist $d\in\N$,
and $p,q\in\K^d[\R^n]$, polynomials $\myOmega:\Rp^2\rightarrow\Rp$
and $\Upsilon:\Rp^2\rightarrow\Rp$
such that for any $x\in\dom{f}$, there exists (a unique) $y:\Rp\rightarrow\R^d$ satisfying
for all $t\in\Rp$:
\begin{itemize}
\item $y(0)=q(x)$ and $y'(t)=p(y(t))$\hfill$\blacktriangleright$ $y$ satisfies a PIVP
\item $\forall \mu\in\Rp$, if $t\geqslant\myOmega(\infnorm{x},\mu)$ then $\infnorm{y_{1..m}(t)-f(x)}\leqslant e^{-\mu}$
\hfill $\blacktriangleright$ $y_{1..m}$ converges to $f(x)$
\item $\infnorm{y(t)}\leqslant\Upsilon(\infnorm{x},t)$\hfill$\blacktriangleright$ $y$ is bounded
\end{itemize}

\item  There exist $d\in\N$,
and $p\in\K^d[\R^d],q\in\K^d[\R^{n+1}]$, and polynomial $\myOmega:\Rp^2\rightarrow\Rp$
and $\Upsilon:\Rp^3\rightarrow\Rp$
such that for any $x\in\dom{f}$ and $\mu\in\Rp$, there exists (a unique) $y:\Rp\rightarrow\R^d$ satisfying
for all $t\in\Rp$:
\begin{itemize}
	\item $y(0)=q(x,\mu)$ and $y'(t)=p(y(t))$\hfill$\blacktriangleright$ $y$ satisfies a PIVP
	\item if $t\geqslant\myOmega(\infnorm{x},\mu)$ then $\infnorm{y_{1..m}(t)-f(x)}\leqslant e^{-\mu}$
	\hfill$\blacktriangleright$ $y_{1..m}$ approximates $f(x)$
	\item $\infnorm{y(t)}\leqslant\Upsilon(\infnorm{x},\mu,t)$
	\hfill$\blacktriangleright$ $y$ is bounded
\end{itemize}

\item (illustrated by Figure \ref{fig:goc}) There exist $\delta\geqslant0$, $d\in\N$ and
$p\in\K^d[\R^d\times\R^{n}]$, $y_0\in\K^d$ and polynomials
$\Upsilon,\myOmega,\Lambda:\Rp^2\rightarrow\Rp$, such that for any $x\in C^0(\Rp,\R^n)$,
there exists (a unique) $y:\Rp\rightarrow\R^d$ satisfying for all $t\in\Rp$:
\begin{itemize}
\item $y(0)=y_0$ and $y'(t)=p(y(t),x(t))$
\hfill$\blacktriangleright$ $y$ satisfies a PIVP (with input)
\item $\infnorm{y(t)}\leqslant\Upsilon\big(\pastsup{\delta}{\infnorm{x}}(t),t\big)$
\hfill$\blacktriangleright$ $y$ is bounded
\item for any $I=[a,b]\subseteq\Rp$, if there exist $\bar{x}\in\dom{f}$ and $\bar{\mu}\geqslant0$ such that
for all $t\in I$, $\infnorm{x(t)-\bar{x}}\leqslant e^{-\Lambda(\infnorm{\bar{x}},\bar{\mu})}$ then
$\infnorm{y_{1..m}(u)-f(\bar{x})}\leqslant e^{-\bar{\mu}}$ whenever
$a+\myOmega(\infnorm{\bar{x}},\bar{\mu})\leqslant u\leqslant b$.
\hfill$\blacktriangleright$ $y$ converges to $f(x)$ when input $x$ is stable
\end{itemize}
\item There exist $\delta\geqslant0$, $d\in\N$ and
$(g:\R^{d}\times\R^{n+1}\rightarrow\R^d)\in\gpval[\K]{}$ and
polynomials $\Upsilon:\Rp^3\rightarrow\Rp$ and $\myOmega,\Lambda,\Theta:\Rp^2\rightarrow\Rp$  such that for any
$x\in C^0(\Rp,\R^n)$, $\mu\in C^0(\Rp,\Rp)$, $y_0\in\R^d$, $e\in C^0(\Rp,\R^d)$
there exists (a unique) $y:\Rp\rightarrow\R^d$ satisfying for all $t\in\Rp$:
\begin{itemize}
\item $y(0)=y_0$ and $y'(t)=g(t,y(t),x(t),\mu(t))+e(t)$
\item $\infnorm{y(t)}\leqslant
\Upsilon\left(\pastsup{\delta}{\infnorm{x}}(t),\pastsup{\delta}{\mu}(t),
    \infnorm{y_0}\indicator{[1,\delta]}(t)+\int_{\max(0,t-\delta)}^t\infnorm{e(u)}du\right)$
\item For any $I=[a,b]$, if there exist $\bar{x}\in\dom{f}$ and $\check{\mu},\hat{\mu}\geqslant0$ such that
for all $t\in I$:
\[\mu(t)\in[\check{\mu},\hat{\mu}]\text{ and }\infnorm{x(t)-\bar{x}}\leqslant e^{-\Lambda(\infnorm{\bar{x}},\hat{\mu})}
\text{ and }\int_{a}^b\infnorm{e(u)}du\leqslant e^{-\Theta(\infnorm{\bar{x}},\hat{\mu})}\]
then
\[\infnorm{y_{1..m}(u)-f(\bar{x})}\leqslant e^{-\check{\mu}}\text{ whenever }
a+\myOmega(\infnorm{\bar{x}},\hat{\mu})\leqslant u\leqslant b.\]
\end{itemize}

\end{enumerate}
\end{proposition}


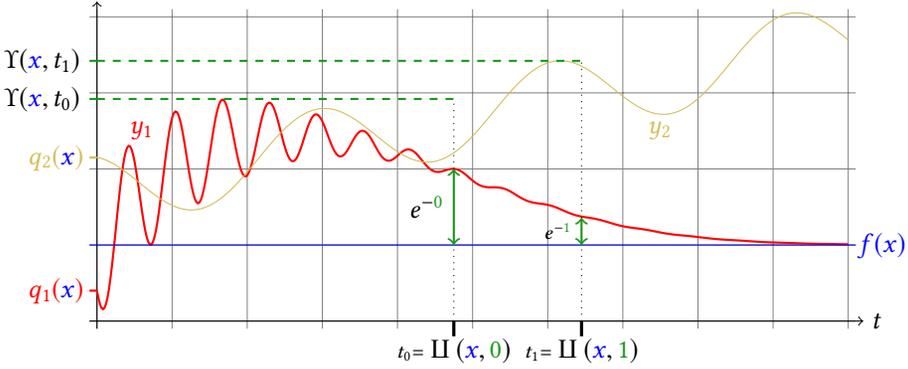
\begin{figure}
\centering
\begin{tikzpicture}[domain=1:11,samples=500,scale=1]
    \draw[very thin,color=gray] (0.9,-0.1) grid (11.1,4.1);
    \draw[->] (1,-0.1) -- (1,4.2);
    \draw[->] (0.9,0) -- (11.2,0) node[right] {$t$};
    \draw[color=red,thick] plot[id=fn_2] function{
        1+exp((x-1)*(7-x)/10)/2+(1+sin(10*x))/(1+exp(x-3))-1.5*exp(-(x-1))
        };
    \draw[blue] (0.9,1) -- (11.1,1);
    \draw[right,blue] (11,1) node {$f(\textcolor{blue}{x})$};
    \draw[color=red,thick] (1.0,0.4) -- (0.9,0.4);
    \draw[red] (0.95,0.4) node[left] {$q_1(\textcolor{blue}{x})$};
    \draw[color=red] (1.6,2.5) node {$y_1$};
    \draw[very thick] (5.75,0) -- (5.75,-0.2);
    \draw[<->,darkgreen,thick] (5.75,1) -- (5.75,2) node[midway,black,left] {$e^{-\textcolor{darkgreen}{0}}$};
    \draw (5.75,-0.1) node[below] {${\scriptstyle t_0=}\myOmega(\textcolor{blue}{x},\textcolor{darkgreen}{0})$};
    \draw[dotted] (5.75,0) -- (5.75,1);
    \draw[<->,darkgreen,thick] (7.45,1) -- (7.45,1.36) node[midway,black,left] {$\scriptstyle e^{-\textcolor{darkgreen}{1}}$};
    \draw[dotted] (7.45,0) -- (7.45,1);
    \draw[very thick] (7.45,0) -- (7.45,-0.2);
    \draw (7.45,-0.1) node[below] {${\scriptstyle t_1=}\myOmega(\textcolor{blue}{x},\textcolor{darkgreen}{1})$};
    \draw[color=myyellow] plot[id=fn_3] function {1.5+x/5+sin(x*2)/2};
    \draw[color=myyellow,thick] (1.0,2.15) -- (0.9,2.15);
    \draw[myyellow] (0.95,2.15) node[left] {$q_2(\textcolor{blue}{x})$};
    \draw[color=myyellow] (8.5,2.5) node {$y_2$};
    \draw[dashed,thick,darkgreen] (1,2.92) -- (5.75,2.92);
    \draw[dotted] (5.75,2) -- (5.75,2.92);
    \draw[darkgreen,thick] (0.9,2.92) -- (1,2.92);
    \draw (.9,2.92) node[left] {$\Upsilon(\textcolor{blue}{x},t_0)$};
    \draw[dashed,thick,darkgreen] (1,3.42) -- (7.45,3.42);
    \draw[dotted] (7.45,1.36) -- (7.45,3.42);
    \draw[darkgreen,thick] (0.9,3.42) -- (1,3.42);
    \draw (.9,3.42) node[left] {$\Upsilon(\textcolor{blue}{x},t_1)$};
\end{tikzpicture}
\caption{$f \in \gc{\Upsilon}{\myOmega}$: On input $x$, starting from initial condition $q(x)$,
the PIVP $y'=p(y)$ ensures that $y_1(t)$ gives $f(x)$ with accuracy better than $e^{-\mu}$
as soon as the time $t$ is greater than $\myOmega(\infnorm{x},\mu)$. At the same time,
all variables $y_j$ are bounded by $\Upsilon(\infnorm{x},t)$. Note that the variables $y_2,\ldots,y_d$
need not converge to anything.\label{fig:gc}}
\end{figure}

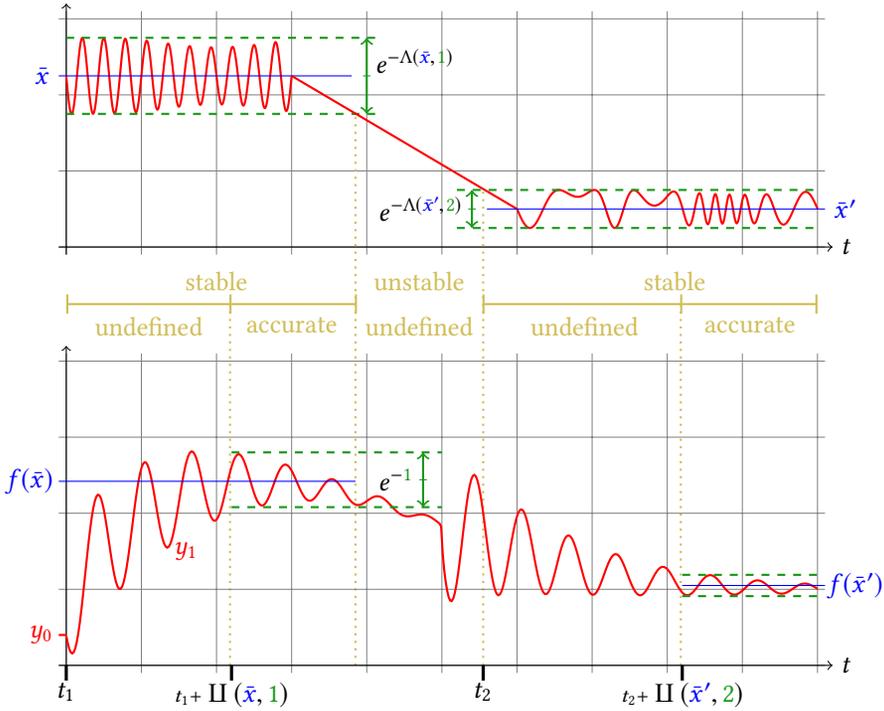
\begin{figure}
\centering
\begin{tikzpicture}[domain=1:11,samples=500,scale=1]
    \begin{scope}[shift={(0,5.5)}]
    \draw[very thin,color=gray] (0.9,-0.1) grid (11.1,3.1);
    \draw[->] (1,-0.1) -- (1,3.2);
    \draw[->] (0.9,0) -- (11.2,0) node[right] {$t$};
    \draw[color=red,thick] plot[domain=1:4,id=fn_4] function{2.25+sin(x*7*pi)/4*(3+tanh((x-3)*(x-3)))/2};
    \draw[blue] (0.9,2.25) -- (4.8,2.25);
    \draw[blue] (0.9,2.25) node[left] {$\bar{x}$};
    \draw[dashed,darkgreen,thick] (1,2.75) -- (5.2,2.75);
    \draw[dashed,darkgreen,thick] (1,1.75) -- (5.2,1.75);
    \draw[<->,darkgreen,thick] (5,1.75) -- (5,2.75);
    \draw[darkgreen] (4.95,2.25) -- (5.05,2.25);
    \draw (5,2.45) node[right] {$e^{-\Lambda(\textcolor{blue}{\bar{x}},\textcolor{darkgreen}{1})}$};

    \draw[red,thick] (4,2.25) -- (7,.5);

    \draw[color=red,thick] plot[domain=7:11,id=fn_5] function{0.5+sin(x*3*(3+tanh((x-9)*(x-9)))/4*pi)/4*(3+tanh((x-10)*(x-10)))/4};
    \draw[dashed,darkgreen,thick] (6.2,0.75) -- (11,0.75);
    \draw[dashed,darkgreen,thick] (6.2,0.25) -- (11,0.25);
    \draw[<->,darkgreen,thick] (6.4,0.25) -- (6.4,0.75);
    \draw[darkgreen] (6.35,.5) -- (6.45,.5);
    \draw (6.4,.5) node[left] {$e^{-\Lambda(\textcolor{blue}{\bar{x}'},\textcolor{darkgreen}{2})}$};
    \draw[blue] (6.6,.5) -- (11.1,.5);
    \draw[blue] (11.1,.5) node[right] {$\bar{x}'$};
    \end{scope}

    \begin{scope}[shift={(0,4.75)}]
    \draw[thick,myyellow,|-|] (1,0) -- (3.2,0);
    \draw[thick,myyellow,-|] (3.2,0) -- (4.87,0);
    \draw[thick,myyellow,|-|] (6.54,0) -- (9.2,0);
    \draw[thick,myyellow,-|] (9.2,0) -- (11,0);
    \end{scope}

    \draw[dotted,myyellow,thick] (4.85,0) -- (4.85,7.25);
    \draw[dotted,myyellow,thick] (3.18,0) -- (3.18,4.75);
    \draw (2.1,4.7) node[below,myyellow] {undefined};
    \draw (4,4.7) node[below,myyellow] {accurate};
    \draw (3,4.8) node[above,myyellow] {stable};

    \draw (5.7,4.8) node[above,myyellow] {unstable};
    \draw (5.7,4.7) node[below,myyellow] {undefined};
    
    \draw[dotted,myyellow,thick] (6.55,0) -- (6.55,6.25);
    \draw[dotted,myyellow,thick] (9.18,0) -- (9.18,4.75);
    \draw (7.9,4.7) node[below,myyellow] {undefined};
    \draw (10.1,4.7) node[below,myyellow] {accurate};
    \draw (9.1,4.8) node[above,myyellow] {stable};

    \draw[very thin,color=gray] (0.9,-0.1) grid (11.1,4.1);
    \draw[->] (1,-0.1) -- (1,4.2);
    \draw[->] (0.9,0) -- (11.2,0) node[right] {$t$};
    \draw[color=red,thick] plot[id=fn_6] function{
        1+exp((x-1)*(7-x)/10)/2+(1+sin(10*x))/(1+exp(x-2.8))-1.5*exp(-(x-1))
        +(1+tanh(100*(x-6)))/2*sin(10*x)*exp((6-x)/1.75)
        };
    \draw[blue] (0.9,2.42) -- (4.85,2.42);
    \draw[right,blue] (0.95,2.42) node[left] {$f(\textcolor{blue}{\bar{x}})$};
    \draw[color=red,thick] (1.0,0.4) -- (0.9,0.4);
    \draw[red] (0.95,0.4) node[left] {$y_0$};
    \draw[color=red] (2.6,1.5) node {$y_1$};
    \draw[darkgreen,dashed,thick] (3.2,2.8) -- (6,2.8);
    \draw[darkgreen,dashed,thick] (3.2,2.08) -- (6,2.08);
    \draw[<->,darkgreen,thick] (5.75,2.08) -- (5.75,2.8) node[midway,black,left] {$e^{-\textcolor{darkgreen}{1}}$};
    \draw[darkgreen] (5.7,2.44) -- (5.8,2.44);

    \draw[very thick] (1,0) -- (1,-0.2);
    \draw (1,-.1) node[below] {$t_1$};
    \draw[very thick] (3.2,0) -- (3.2,-0.2);
    \draw (3.2,-0.1) node[below] {${\scriptstyle t_1+}\myOmega(\textcolor{blue}{\bar{x}},\textcolor{darkgreen}{1})$};

    \draw[blue] (9.2,1.05) -- (11.1,1.05);
    \draw[right,blue] (11,1.05) node[right] {$f(\textcolor{blue}{\bar{x}'})$};
    \draw[very thick] (6.55,0) -- (6.55,-0.2);
    \draw (6.55,-.1) node[below] {$t_2$};
    \draw[very thick] (9.2,0) -- (9.2,-0.2);
    \draw (9.2,-0.1) node[below] {${\scriptstyle t_2+}\myOmega(\textcolor{blue}{\bar{x}'},\textcolor{darkgreen}{2})$};
    \draw[darkgreen,dashed,thick] (9.2,1.19) -- (11,1.19);
    \draw[darkgreen,dashed,thick] (9.2,0.91) -- (11,0.91);
\end{tikzpicture}
\caption{$f \in \goc{\Upsilon}{\myOmega}{\Lambda}$: starting from the (constant) initial condition $y_0$,
the PIVP $y'(t)=p(y(t),x(t))$ has two possible behaviors depending on the input signal $x(t)$.
If $x(t)$ is unstable, the behavior of the PIVP $y'(t)=p(y(t),x(t))$ is undefined.
If $x(t)$ is stable around $\bar{x}$ with error at most $e^{-\Lambda(\infnorm{\bar{x}},\mu)}$
then $y(t)$ is initially undefined, but after a delay of at most $\myOmega(\infnorm{\bar{x}},\mu)$, $y_1(t)$ gives $f(\bar{x})$
with accuracy better than $e^{-\mu}$.
In all cases, all variables $y_j(t)$ are bounded by a function ($\Upsilon$)
of the time $t$ and the supremum of $\infnorm{x(u)}$ during a small time interval $u\in[t-\delta,t]$.
\label{fig:goc}}
\end{figure}

Note that (1) and (2) in the previous proposition are very closely related, and only differ
in how the complexity is measured. In (1), based on length, we measure the length
required to reach precision $e^{-\mu}$. In (2), based on  time+space, we measure
the time $t$ required to reach precision $e^{-\mu}$ and the space (maximum value of
all components) during the time interval $[0,t]$.





Item (3) in the previous proposition gives an apparently weaker form of
computability where the system is no longer required to converge to $f(x)$ on
input $x$. Instead, we give the system an input $x$ and a precision $\mu$, and
ask that the system stabilizes within $e^{-\mu}$ of $f(x)$.

Item  (4) in the previous proposition is a form of
online-computability: the input is no
longer part of the initial condition but rather given by an external input $x(t)$.
The intuition is that if $x(t)$ approaches a value $\bar{x}$ sufficiently close, then by waiting
long enough (and assuming that the external input stays near the value $\bar{x}$ during that time interval), we will get
an approximation of $f(\bar{x})$ with some desired accuracy. 
This will be called online-computability.

Item (5) is a version  robust with respect to perturbations. This
notion will only be used in some proofs, and will be called \unaware{} computability.

\begin{remark}[Effective Limit computability]\label{rem:gpwc_gpc}\label{rem:gdpwc_gdpc}
A careful look at Item (3)  of previous Proposition shows that it
corresponds to a form of effective limit computability. Formally, let $f:I\times\Rps\rightarrow\R^n$, $g:I\rightarrow\R^n$ and $\mho:\Rp^2\rightarrow\Rp$
a polynomial. Assume that $f\in\mygpc$ and that for any $x\in I$ and $\tau\in\Rps$,
if $\tau\geqslant\mho(\infnorm{x},\mu)$ then $\infnorm{f(x,\tau)-g(x)}\leqslant e^{-\mu}$.
Then $g\in\mygpc$ because the analog system for $f$ satisfies all the items of the
definition.
\end{remark}

For notational purpose, we will write  $f \in
\gc{\Upsilon}{\myOmega}$ \label{page:def:gc}  when
$f$ satisfies (2) with corresponding polynomials $\Upsilon$ and
$\myOmega$, $f\in\gwc{\Upsilon}{\myOmega}$ when when $f$ satisfies (3) with corresponding polynomials $\Upsilon$ and 
$\myOmega$,  $f \in \goc{\Upsilon}{\myOmega}{\Lambda}$ \label{page:def:goc}
when $f$ satisfies (4) with corresponding polynomials $\Upsilon$, 
$\myOmega$ and $\Lambda$, and we will write $f \in
\guc{\Upsilon}{\myOmega}{\Lambda}{\Theta}$ \label{page:def:guc}  when $f$ satisfies (5) with corresponding polynomials $\Upsilon$, 
$\myOmega$, $\Lambda$, $\Theta$.

\subsection{Dynamics and encoding can be assumed generable}

Before moving on to some basic properties of computable functions, we
observe that a certain aspect of the definitions does not really
matter: In Item (2) of Proposition~\ref{th:main_eq}, we required
that $p$ and $q$ be polynomials. It turns out, surprisingly, that the
class is the same if we only assume that $p,q\in\gpval$. This remark
also applies to the Item  (3). This turns out to be very useful
when defining computable function. 

Following  proposition follows from
Remark 26 of \cite{\JOURNALOFCOMPLEXITY}. 

\begin{remark} Notice that this also holds
for class $\mygpc$, even if not stated explicitely in
\cite{\JOURNALOFCOMPLEXITY}. Indeed, in Theorem 20 of
\cite{\JOURNALOFCOMPLEXITY} ($\operatorname{ALP}=\operatorname{
  ATSP}$), the inclusion $\operatorname{ATSP}\subseteq\operatorname{
  ALP}$ is trivial. Now,  when proving that  $\operatorname{ALP} \subseteq \operatorname{
  ATSP}$,  the considered $p$ and $g$ could have been assumed
generable without any difficulty. 
\end{remark}

\begin{proposition}[Polynomial versus
  generable
]\label{prop:gp_gpw_gen}
Theorem~\ref{th:main_eq} is still true if we  only assume that
$p,q\in\gpval$ in Item (2) or (3) (instead of $p,q$ polynomials).
\end{proposition}

We will use intensively this remark from now on.  Actually, in several of
the proofs, given a function from $\mygpc$, we will use the fact that it
satisfies item (2) (the stronger notion) to build another function
satifying item (3) with functions $p$ and $q$ in $\gpval$ (the weaker
notion). From Proposition~\ref{th:main_eq}, this proves that the
constructed function is indeed in $\mygpc$.



\section{Some preliminary results}
\label{sec:original}

%
%
%


In this section, we present new and original results the exception being in subsection \ref{subsection:resultsElsewhere}. First we relate
generability to computability. Then, we prove some closure results for
the class of computable functions. Then, we discuss their continuity
and growth. Finaly, we prove that some basic functions such as
$\min,\max$ and absolute value, and rounding functions are in $\mygpc$.

\subsection{Generable implies computable over star domains}

We introduced the notion of GPAC generability and of GPAC
computability. The later can be considered as a generalization of the
first, and as such, it may seem natural that any
generable function must be computable: The intuition tells us that computing the value of $f$, a generable function,
at point $x$ is only a matter of finding a path \emph{in the domain of definition}
from the initial value $x_0$ to $x$, and simulating the differential equation along
this path. 

This however requires some discussions and hypotheses on the the
domain of definition of the function: We recall that a function is
generable if it satisfies a PIVP over an open connected subset.  We
proved in \cite{BournezGP16a} that there is always a path between $x_0$
to $x$ and it can even be assumed to be generable.

\begin{proposition}[Generable path connectedness
]\label{prop:connected_is_generable_connected}
An open, connected subset $U$ of $\R^n$ is always \emph{generable-path-connected}:
for any $a,b\in (U\cap\K^n)$, there exists $(\phi:\R\rightarrow U)\in\gpval[\K]{}$ such that
$\phi(0)=a$ and $\phi(1)=b$.
\end{proposition}

However, the proof is not constructive and we have no easy way of computing such a
path given $x$.  


For this reason, we 
restrict ourselves to the case where
 finding the path is trivial: star domains with a generable vantage point. 

\begin{definition}[Star domain]\label{def:star_domain}
    A set $X\subseteq\R^n$ is called a \emph{star domain}
if there exists $x_0\in X $ such that for all $x\in U$ the line segment from $x_0$ to $x$
is in $X$, i.e $[x_0,x]\subseteq X$. Such an $x_0$ is called a \emph{vantage point}.
\end{definition}

The following result is true, where a generable vantage point means a vantage point
which belongs to a generable field. We will
mostly need this theorem for domains of the form $\R^n\times\Rp^m$, which happen
to be star domains. 

\begin{theorem}[$\gpval\subseteq\mygpc$ over star domains
  ]
  If
$f\in\gpval$ has a star domain with a generable
vantage point then $f\in\mygpc$.
\end{theorem}

\begin{proof}
Let $(f:\subseteq\R^n\rightarrow\R^m)\in\gpval$ 
and $z_0\in\dom{f}\cap\K^n$ a generable vantage point.
Apply Definition~\ref{def:gpac_generable_ext} to get $\mtt{sp}, d,p,x_0,y_0$ and $y$. Since $y$ is generable
and $z_0\in\K^d$, apply Proposition \ref{coro55deinfoandcomputation}
to get that $y(z_0)\in\K^d$.
Let $x\in\dom{f}$ and consider the following system:
\[
\left\{\begin{array}{@{}r@{}l@{}}
x(0)&=x\\\gamma(0)&=x_0\\z(0)&=y(z_0)
\end{array}\right.
\qquad
\left\{\begin{array}{@{}r@{}l@{}}
x'(t)&=0\\\gamma'(t)&=x(t)-\gamma(t)\\z'(t)&=p(z(t))(x(t)-\gamma(t))
\end{array}\right.
\]
First note that $x(t)$ is constant and check that $\gamma(t)=x+(x_0-x)e^{-t}$ and note that $\gamma(\Rp)\subseteq[x_0,x]\subseteq\dom{f}$
because it is a star domain.
Thus $z(t)=y(\gamma(t))$ since $\gamma'(t)=x(t)-\gamma(t)$ and $\jacobian{y}=p$.
It follows that $\infnorm{f(x)-z_{1..m}(t)}=\infnorm{f(x)-f(\gamma(t))}$ since $z_{1..m}=f$.
Apply Proposition \ref{prop:generable_mod_cont} to $f$ to get a polynomial $q$ such that
\[\forall x_1,x_2\in\dom{f}, [x_1,x_2]\subseteq\dom{f}\;\Rightarrow\;
    \infnorm{f(x_1)-f(x_2)}\leqslant\infnorm{x_1-x_2}q(\mtt{sp}(\max(\infnorm{x_1},\infnorm{x_2}))).\]
Since $\infnorm{\gamma(t)}\leqslant\infnorm{x_0,x}$ we have
\[\infnorm{f(x)-z_{1..m}(t)}\leqslant\infnorm{x-x_0}e^{-t}q(\infnorm{x_0,x})\leqslant e^{-t}\poly(\infnorm{x}).\]
Finally, $\infnorm{z(t)}\leqslant\mtt{sp}(\gamma(t))\leqslant\poly(\infnorm{x})$ because $\mtt{sp}$
is a polynomial. Then, by Proposition~\ref{th:main_eq}, $f\in\mygpc$.

\end{proof}

\subsection{Closure by arithmetic operations and composition}

The class of polynomial time computable
function is stable under addition, subtraction and multiplication, and
composition.

\begin{theorem}[Closure by arithmetic operations]
\label{th:gpac_comp_arith}
If $f,g\in\mygpc$ then $f\pm g,fg\in\mygpc$, with the obvious restrictions on the domains
of definition.
\end{theorem}

\begin{proof}
We do the proof for the case of $f+g$ in detail. The other cases are similar. 
Apply Proposition
\ref{prop:mainequivalence} to get polynomials
$\myOmega,\Upsilon,\myOmega^*,\Upsilon^*$
such that $f\in\gc{\Upsilon}{\myOmega}$ and $g\in\gc{\Upsilon^*}{\myOmega^*}$
with corresponding $d,p,q$ and $d^*,p^*,q^*$ respectively.
 Let $x\in\dom{f}\cap\dom{g}$ and consider
the following system:
\[
\left\{\begin{array}{r@{}l}y(0)&=q(x)\\z(0)&=q^*(x)\\w(0)&=q(x)+q^*(x)\end{array}\right.
\qquad
\left\{\begin{array}{r@{}l}y'(t)&=p(y(t))\\z'(t)&=p^*(z(t))\\w'(t)&=p(y(t))+p^*(z(t)))\end{array}\right..
\]
Notice that $w$ was built so that $w(t)=y(t)+z(t)$.
Let
\[\hat{\myOmega}(\alpha,\mu)=\max(\myOmega(\alpha,\mu+\ln2),\myOmega^*(\alpha,\mu+\ln2))\]
and
\[\hat{\Upsilon}(\alpha,t)=\Upsilon(\alpha,t)+\Upsilon^*(\alpha,t).\]
Since, by construction, $w(t)=y(t)+z(t)$, if $t\geqslant\hat{\myOmega}(\alpha,\mu)$
then $\infnorm{y_{1..m}(t)-f(x)}\leqslant e^{-\mu-\ln2}$ and $\infnorm{z_{1..m}(t)-g(x)}\leqslant e^{-\mu-\ln2}$
thus $\infnorm{w_{1..m}(t)-f(x)-g(x)}\leqslant e^{-\mu}$. Furthermore,
$\infnorm{y(t)}\leqslant\Upsilon(\infnorm{x},t)$ and $\infnorm{z(t)}\leqslant\Upsilon^*(\infnorm{x},t)$
thus $\infnorm{w(t)}\leqslant\hat{\Upsilon}(\infnorm{x},t)$.

The case of $f-g$ is exactly the same. The case of $fg$ is slightly more involved:
one needs to take
\[w'(t)=y_1'(t)z_1(t)+y_1(t)z_1'(t)=p_1(y(t))z_1(t)+y_1(t)p_1^*(z(t))\]
so that $w(t)=y_1(t)z_1(t)$. The error analysis is a bit more complicated because the speed
of convergence now depends on the length of the input.

First note that $\infnorm{f(x)}\leqslant1+\Upsilon(\infnorm{x},\myOmega(\infnorm{x},0))$
and $\infnorm{g(x)}\leqslant1+\Upsilon^*(\infnorm{x},\myOmega^*(\infnorm{x},0))$,
and denote by $\ell(\infnorm{x})$ and $\ell^*(\infnorm{x})$ those two bounds respectively.
If $t\geqslant\myOmega(\infnorm{x},\mu+\ln2\ell^*(\infnorm{x}))$ then $\infnorm{y_1(t)-f(x)}\leqslant e^{-\mu-\ln2\infnorm{g(x)}}$
and similarly if $t\geqslant\myOmega^*(\infnorm{x},\mu+\ln2(1+\ell^*(\infnorm{x})))$ then
$\infnorm{z_1(t)-g(x)}\leqslant e^{-\mu-\ln2(1+\infnorm{f(x)})}$. Thus for $t$ greater
than the maximum of both bounds,
\[\infnorm{y_1(t)z_1(t)-f(x)g(x)}\leqslant\infnorm{(y_1(t)-f(x))g(x)}+\infnorm{y_1(t)(z_1(t)-g(x))}
\leqslant e^{-\mu}\]
because $\infnorm{y_1(t)}\leqslant1+\infnorm{f(x)}\leqslant1+\ell(\infnorm{x})$.
\end{proof}




Recall that we assume we are working over a generable $\K$.  

\begin{theorem}[Closure by composition]
\label{th:gpac_comp_composition}
If $f,g\in\mygpc$ and $f(\dom{f})\subseteq\dom{g}$ then $g\circ f\in\mygpc$.
\end{theorem}

\begin{proof}
Let $f:I\subseteq\R^n\rightarrow J\subseteq\R^m$ and $g:J\rightarrow K\subseteq\R^l$.
We will show that $g\circ f$ is computable by using the fact that $g$ is
online-computable. We could show directly that $g\circ f$ is online-computable but this would
only complicate the proof for no apparent gain.

Apply Proposition~\ref{th:main_eq} to get that $g \in 
\goc{\Upsilon}{\myOmega}{\Lambda}$ with corresponding $r,\Delta,z_0$. 
Apply Proposition~\ref{th:main_eq} to get that $f \in
\gc{\Upsilon'}{\myOmega'}$ with corresponding $d,p,q$. 
Let $x\in I$ and consider the following system:
\begin{equation*}
\left\{\begin{array}{@{}r@{}l}y(0)&=q(x)\\y'(t)&=p(y(t))\end{array}\right.
\qquad
\left\{\begin{array}{@{}r@{}l}z(0)&=z_0\\z'(t)&=r(z(t),y_{1..m}(t))\end{array}\right..
\end{equation*}

Define $v(t)=(x(t),y(t),z(t))$. Then it immediately follows that $v$ satisfies a PIVP of the form
$v(0)=\poly(x)$ and $v'(t)=\poly(v(t))$. Furthermore, by definition:
\begin{align*}
\infnorm{v(t)}&=\max(\infnorm{x},\infnorm{y(t)},\infnorm{z(t)})\\
    &\leqslant\max\left(\infnorm{x},\infnorm{y(t)},\Upsilon\left(\sup_{u\in[t,t-\Delta]\cap\Rp}\infnorm{y_{1..m}(t)},t\right)\right)\\
    &\leqslant\poly\left(\infnorm{x},\sup_{u\in[t,t-\Delta]\cap\Rp}\infnorm{y(t)},t\right)\\
    &\leqslant\poly\left(\infnorm{x},\sup_{u\in[t,t-\Delta]\cap\Rp}\Upsilon'\left(\infnorm{x},u\right),t\right)\\
    &\leqslant\poly\left(\infnorm{x},t\right).\\
\end{align*}

Define $\bar{x}=f(x)$, $\Upsilon^*(\alpha)=1+\Upsilon'(\alpha,0)$ and
$\myOmega''(\alpha,\mu)=\myOmega'(\alpha,\Lambda(\Upsilon^*(\alpha),\mu))+\myOmega(\Upsilon^*(\alpha),\mu)$.
By definition of $\Upsilon'$, $\infnorm{\bar{x}}\leqslant1+\Upsilon'(\infnorm{x},0)=\Upsilon^*(\infnorm{x})$.
Let $\mu\geqslant0$ then by definition of $\myOmega'$,
if $t\geqslant\myOmega'(\infnorm{x},\Lambda(\Upsilon^*(\infnorm{x}),\mu))$ then
$\infnorm{y_{1..m}(t)-\bar{x}}\leqslant e^{-\Lambda(\Upsilon^*(\infnorm{x}),\mu)}
\leqslant e^{-\Lambda(\infnorm{\bar{x}},\mu)}$. 
For $a=\myOmega'(\infnorm{x},\Lambda(\Upsilon^*(\infnorm{x}),\mu))$ we get that
$\infnorm{z_{1..l}(t)-g(f(x))}\leqslant e^{-\mu}$ for any $t\geqslant a+\myOmega(\bar{x},\mu)$.
And since $t\geqslant a+\myOmega(\bar{x},\mu)$ whenever $t\geqslant\myOmega''(\infnorm{x},\mu)$,
we get that $g\circ f\in\gc{\poly}{\myOmega''}$. This concludes the proof because $\myOmega''$ is a polynomial.
\end{proof}

\subsection{Continuity and growth}\label{sec:cont_growth}

All computable functions are continuous.
More importantly, they admit a polynomial modulus of continuity,
in a similar spirit as in Computable Analysis. 

\begin{theorem}[Modulus of continuity] 
\label{th:comp_implies_cont}
If $f\in\mygpc$ then $f$ admits a polynomial modulus of continuity: there exists
a polynomial $\mho:\Rp^2\rightarrow\Rp$ such that for all $x,y\in\dom{f}$ and $\mu\in\Rp$,
\[\infnorm{x-y}\leqslant e^{-\mho(\infnorm{x},\mu)}\quad\Rightarrow\quad\infnorm{f(x)-f(y)}\leqslant e^{-\mu}.\]
In particular $f$ is continuous.
\end{theorem}

\begin{proof}
Let $f\in\mygpc$, apply Proposition~\ref{th:main_eq} to get that
$f\in\goc{\Upsilon}{\myOmega}{\Lambda}$ 
with
corresponding $\delta,d,p$ and $y_0$.  Without loss of generality,
we assume polynomial $\myOmega$ to be an increasing function.  Let $u,v\in\dom{f}$
and $\mu\in\Rp$. Assume that $\infnorm{u-v}\leqslant e^{-\Lambda(\infnorm{u}+1,\mu+\ln2)}$
and consider the following system:
\[y(0)=y_0\qquad y'(t)=p(y(t),u).\]
This is simply the online system where we hardwired the input of the system to the constant
input $u$. The idea is that the definition of online computability can be applied to both
$u$ with $0$ error, or $v$ with error $\infnorm{u-v}$.

By definition, $\infnorm{y_{1..m}(t)-f(u)}\leqslant e^{-\mu-\ln2}$ for all $t\geqslant\myOmega(\infnorm{u},\mu+\ln2)$.
For the same reason, $\infnorm{y_{1..m}(t)-f(v)}\leqslant e^{-\mu-\ln2}$ for all $t\geqslant\myOmega(\infnorm{v},\mu+\ln2)$
because $\infnorm{u-v}\leqslant e^{-\Lambda(\infnorm{u}+1,\mu+ln2)}\leqslant e^{-\Lambda(\infnorm{v},\mu+\ln2)}$
and $\infnorm{v}\leqslant\infnorm{u}+1$.
Combine both results at $t=\myOmega(\infnorm{u}+1,\mu+\ln2)$ to get that $\infnorm{f(u)-f(v)}\leqslant e^{-\mu}$.
\end{proof}

It is  is worth observing that all
functions in $\mygpc$ are polynomially bounded (this follows trivially
from condition (2) of  Proposition \ref{prop:mainequivalence}).

\begin{proposition}\label{prop:gp_growth}
Let $f\in\mygpc$, there exists a polynomial $P$ such that $\infnorm{f(x)}\leqslant P(\infnorm{x})$
for all $x\in\dom{f}$.
\end{proposition}

\subsection{Some basic functions proved to be in  $\mygpc$}

\subsubsection{Absolute, minimum, maximum value}

We will now show that basic functions like the absolute value,
the minimum and maximum value are computable. We will also show a powerful
result when limiting a function to a computable range. In essence all
these result follow from the fact that the absolute value belongs to $\mygpc$,
which is a surprisingly non-trivial result (see the example below).

\begin{example}[Broken way of computing the absolute value]\label{ex:non_comp_abs}
Computing the absolute value in polynomial length, or equivalently in
polynomial time with polynomial bounds,  is a surprisingly difficult
operation, for unintuitive reasons. This example illustrates the problem.
A natural idea to compute the absolute value is to notice that
$|x|=x\sgn{x}$, where $\sgn{x}$ denotes the sign function (with
conventionally $\sgn{0}=0$). 
To this end, define $f(x,t)=x\tanh(xt)$
which works because $\tanh(xt)\rightarrow\sgn{x}$ when
$t\rightarrow\infty$. Unfortunately, $\big||x|-f(x,t)\big|\leqslant|x|e^{-|x|t}$
which \textbf{converges very slowly for small $x$}. Indeed, if $x=e^{-\alpha}$
then $\big||x|-f(x,t)\big|\leqslant e^{-\alpha-e^{-\alpha}t}$ so we must take
$t(\mu)=e^\alpha\mu$ to reach a precision $\mu$. This is unacceptable because it grows
with $\frac{1}{|x|}$ instead of $|x|$. In particular, it is unbounded when $x\rightarrow0$
which is clearly wrong.
\end{example}

The sign function is not computable because it not continuous. However, if $f$
is a continuous function that is zero at $0$ then $\sgn{x}f(x)$ is continuous.
We prove an effective version of this remark below. The absolute value will
then follows as a special case of this result.

The proof is not difficult but the idea is not very intuitive. As the example above outlines,
we cannot simply compute $f(x)\tanh(g(x)t)$ and hope that it converges quickly enough
when $t\rightarrow\infty$. What if we could replace $t$ by $e^t$ ? It would work
of course, but we cannot afford to compute the exponential function. Except if we can ?
The crucial point is to notice that
we do not really need to compute $\tanh(g(x)e^t)$ for arbitrary large $t$,
we only need it to ``bootstrap'' so that $g(x)e^t\approx\poly(t)$. In other words,
we need a fast start (exponential) but only a moderate asymptotic growth (polynomial).
This can be done in a clever way by bounding the growth the function when it becomes too large.

\begin{proposition}[Smooth sign is computable]\label{prop:smooth_sign}
For any polynomial $p:\Rp\rightarrow\Rp$, $H_p\in\mygpc$ where
\[
H_p(x,z)=\sgn{x}z,\qquad (x,z)\in
U_p:=\big\{(0,0)\big\}\cup\left\{(x,z)\in\Rs\times\R\thinspace:\thinspace
    \left|\tfrac{z}{x}\right|\leqslant e^{p(\infnorm{x,z})}\right\}.
\]
\end{proposition}

\begin{proof}
Let $(x,z)\in U$ and consider the following system:
\[
\left\{\begin{array}{@{}c@{}l}s(0)&=x\\y(0)&=z\tanh(x)\end{array}\right.
\qquad
\left\{\begin{array}{@{}c@{}l}
s'(t)&=\tanh(s(t))\\
y'(t)&=\left(1-\tanh(s(t))^2\right)y(t)\end{array}\right.
\]
First check that $y(t)=z\tanh(s(t))$. The case of $x=0$ is trivial because $s(t)=0$ and $y(t)=0=H(x,z)$.
If $x<0$ then check that the same system for $-x$ has the opposite value
for $s$ and $y$ so all the convergence result will the exactly the same and will be correct
because $H(x,z)=-H(-x,z)$. Thus we can assume that $x>0$. We will need the following
elementary property of the hyperbolic tangent for all $t\in\R$:
\[1-\sgn{t}\tanh(t)\leqslant e^{-|t|}.\]

Apply the above formula to get that $1-e^{-u}\leqslant\tanh(u)\leqslant1$ for all $u\in\Rp$.
Thus $\tanh(s(t))\geqslant1-e^{-s(t)}$ and by a classical result of differential inequalities,
$s(t)\geqslant w(t)$ where $w(0)=s(0)=x$ and $w'(t)=1-e^{-w(t)}$. Check that
$w(t)=\ln\left(1+(e^x-1)e^t\right)$ and conclude that
\[
|z-y(t)|
    \leqslant|z|(1-\tanh(s(t))|
    \leqslant|z|e^{-s(t)}
    \leqslant\frac{|z|}{1+(e^x-1)e^t}
    \leqslant\frac{|z|e^{-t}}{e^x-1}
    \leqslant\frac{|z|}{x}e^{-t}
    \leqslant e^{p(\infnorm{x,z})-t}.
\]
Thus $|z-y(t)|\leqslant e^{-\mu}$ for all $t\geqslant\mu+p(\infnorm{x,z})$ which
is polynomial in $\infnorm{x,z,\mu}$. Furthermore, $|s(t)|\leqslant|x|+t$ because $|s'(t)|\leqslant1$.
Similarly, $|y(t)|\leqslant|z|$ so the system is polynomially bounded.
Finally, the system is of the form $(s,y)(0)=f(x)$ and $(s,y)'(t)=g((s,y)(t))$
where $f,g\in\gpval$ so $H_p\in\mygpc$ with generable functions.
Apply Proposition~\ref{prop:gp_gpw_gen} to conclude.
\end{proof}

\begin{theorem}[Absolute value is computable]\label{th:gpac_comp_abs}
$(x\mapsto|x|)\in\mygpc$.
\end{theorem}

\begin{proof}
Let $p(x)=0$ which is a polynomial, and $a(x)=H_p(x,x)$ where $H_p\in\mygpc$ comes
from Proposition~\ref{prop:smooth_sign}. It is not hard to see that $a$ is defined over $\R$
because $(0,0)\in U_p$ and for any $x\neq0$, $\left|\tfrac{x}{x}\right|\leqslant1=e^{p(|x|)}$
thus $(x,x)\in U_p$. Consequently $a\in\mygpc$ and for any $x\in\R$, $a(x)=\sgn{x}x=|x|$
which concludes.
\end{proof}

\begin{corollary}[Max, Min are computable]\label{cor:max_min_comp}
$\max,\min\in\mygpc$.
\end{corollary}

\begin{proof}
Use that $\max(a,b)=\frac{a+b}{2}+\left|\frac{a+b}{2}\right|$ and $\min(a,b)=-\max(-a,-b)$.
Conclude with Theorem~\ref{th:gpac_comp_abs} and closure by
arithmetic operations and composition of $\mygpc$. 
\end{proof}

\subsubsection{Rounding}

In \cite{\INFORMATIONANDCOMPUTATION} we showed that it is possible to build a generable
rounding function of very good quality. 

\begin{lemma}[Round
]\label{lem:rnd}
There exists $\functionrnd\in\gpval$ such that for any $n\in\Z$, $\lambda\geqslant2$, $\mu\geqslant0$
and $x\in\R$ we have:
\begin{itemize}
\item $|\functionrnd(x,\mu,\lambda)-n|\leqslant\frac{1}{2}$ if $x\in\left[n-\frac{1}{2},n+\frac{1}{2}\right]$,
\item $|\functionrnd(x,\mu,\lambda)-n|\leqslant e^{-\mu}$ if $x\in\left[n-\frac{1}{2}+\frac{1}{\lambda},n+\frac{1}{2}-\frac{1}{\lambda}\right]$.
\end{itemize}
\end{lemma}

In this section, we will see that we can do even better with computable
functions. More precisely, we will build a computable function that rounds perfectly
everywhere, except on a small, periodic, interval of length $e^{-\mu}$ where $\mu$ is a parameter.
This is the best can do because of the continuity and modulus of continuity
requirements of computable functions, as shown in Theorem~\ref{th:comp_implies_cont}.
We will need a few technical lemmas before getting to the rounding function itself.
We start by a small remark that will be useful later on.

\begin{remark}[Constant function]\label{rem:online_comp_constant}
Let $f\in\mygpc$, $I$ a convex subset of $\dom{f}$ and assume that $f$
is constant over $I$, with value $\alpha$.
From Proposition~\ref{th:main_eq}, we have $f \in 
\goc{\Upsilon}{\myOmega}{\Lambda}$ for some polynomials
$\Upsilon,\myOmega,\Lambda$ with corresponding $d,\delta,p$ and $y_0$. Let $x\in C^0(\Rp,\dom{f})$ and
consider the system:
\[y(0)=y_0\qquad y'(t)=p(y(t),x(t))\]
If there exists $J=[a,b]$ and $M$ such that for all $x(t)\in I$ and $\infnorm{x(t)}\leqslant M$
for all $t\in J$, then $\infnorm{y_{1..m}(t)-\alpha}\leqslant e^{-\mu}$
for all $t\in[a+\myOmega(M,\mu),b]$. This is unlike the usual case where the input must be
nearly constant and it is true because whatever the system can sample from the input $x(t)$,
the resulting output will be the same. Formally, it can shown by building a small system
around the online-system that samples the input, even if it unstable.
\end{remark}

\begin{proposition}[Clamped exponential]\label{prop:clamped_exp}
For any $a,b,c,d\in\K$ and $x\in\R$ such that $a\leqslant b$, define $h$ as follows. Then $h\in\mygpc$:
\[h(a,b,c,d,x)=\max(a,\min(b,ce^x+d)).\]
\end{proposition}

\begin{proof}
First note that we can assume that $d=0$ because $h(a,b,c,d,x)=h(a-d,b-d,c,0,x)+d$.
Similarly, we can assume that $a=-b$ and $b\geqslant|c|$ because
$h(a,b,c,d,x)=\max(a,\min(b,h(-|c|-\max(|a|,|b|),|c|+\max(|a|,|b|),c,d,x)))$
and $\min,\max,|\cdot|\in\mygpc$. So we are left with $H(\ell,c,x)=\max(-\ell,\min(\ell,ce^x))$
where $\ell\geqslant|c|$ and $x\in\R$.
Furthermore, we can assume that $c\geqslant0$ because $H(\ell,c,x)=\sgn{c}H(\ell,|c|,x)$
and it belongs to $\mygpc$ for all $\ell\geqslant|c|$ and $x\in\R$ thanks to Proposition~\ref{prop:smooth_sign}. Indeed,
if $c=0$ then $H(\ell,|c|,x)=0$ and if $c\neq 0$, $\ell\geqslant|c|$ and $x\in\R$,
then $\left|\tfrac{c}{H(\ell,|c|,x)}\right|\geqslant e^{-|x|}$.

We will show that $H\in\mygpc$.
 Let $\ell\geqslant c\geqslant0$, $\mu\in\Rp$,
$x\in\R$ and consider the following system:
\[
\left\{\begin{array}{@{}r@{}l}y(0)&=c\\z(0)&=0
\end{array}\right.
\qquad
\left\{\begin{array}{@{}r@{}l}y'(t)&=z'(t)y(t)\\z'(t)&=(1+\ell-y(t))(x-z(t))
\end{array}\right.
\]
Note that formally, we should add extra variables to hold $x$, $\mu$ and $\ell$ (the inputs).
Also note that to make this a PIVP, we should replace $z'(t)$ by its expression
in the right-hand side, but we kept $z'(t)$ to make things more readable.
By construction $y(t)=ce^{z(t)}$, and since $\ell\geqslant c\geqslant0$, by a classical
differential argument, $z(t)\in[0,x]$ and $y(t)\in[0,\min(ce^x,\ell+1)]$. This shows in particular
that the system is polynomially bounded in $\infnorm{\ell,x,c}$. There are two cases
to consider.
\begin{itemize}
\item If $\ell\geqslant ce^x$ then $\ell-y(t)=\ell-ce^{z(t)}\geqslant c(e^x-e^{z(t)})\geqslant c(x-z(t))\geqslant0$
    thus by a classical differential inequalities reasoning, $z(t)\geqslant w(t)$
    where $w$ satisfies $w(0)=0$ and $w'(t)=(x-w(t))$. This system can be
    solved exactly and $w(t)=x(1-e^{-t})$. Thus
    \[y(t)\geqslant ce^{w(t)}\geqslant ce^xe^{-xe^{-t}}
    \geqslant ce^x(1-xe^{-t})\geqslant ce^x-cxe^{x-t}.\]
    So if $t\geqslant\mu+x+c$ then
    $y(t)\geqslant ce^x-e^{-\mu}$. Since $y(t)\leqslant ce^x$ it shows that
    $|y(t)-ce^x|\leqslant e^{-\mu}$.
\item If $\ell\leqslant ce^x$ then by the above reasoning, $\ell+1\geqslant y(t)\geqslant\ell$
    when $t\geqslant\mu+x+c$.
\end{itemize}
We will modify this sytem to feed $y$ to an online-system computing $\min(-\ell,\max(\ell,\cdot))$.
The idea is that when $y(t)\geqslant\ell$, this online-system is constant so
the input does not need to be stable.

Let $G(x)=\min(\ell,x)$ then $G\in\goc{\Upsilon}{\myOmega}{\Lambda}$
with polynomials $\Lambda,\myOmega,\Upsilon$ are polynomials and corresponding $d,\delta,p$ and $y_0$. Let $x,c,\ell,\mu$ and consider the following system
(where $y$ and $z$ are from the previous system):
\[w(0)=y_0\qquad w'(t)=p(w(t),y(t))\]
Again, there are two cases.
\begin{itemize}
\item If $\ell\geqslant ce^x$ then $|y(t)-ce^x|\leqslant e^{-\Lambda(\ell,\mu)}\leqslant e^{-\Lambda(ce^x,\mu)}$
    when $t\geqslant \Lambda(\ell,\mu)+x+c$, thus $|w_1(t)-G(ce^x)|\leqslant e^{-\mu}$
    when $t\geqslant\Lambda(\ell,\mu)+x+c+\myOmega(\ell,\mu)$
    and this concludes because $G(ce^x)=ce^x$.
\item If $\ell\leqslant ce^x$ then by the above reasoning, $\ell+1\geqslant y(t)\geqslant\ell$
    when $t\geqslant\Lambda(\ell,\mu)+x+c$ and thus $|w_1(t)-\ell|\leqslant e^{-\mu}$
    when $t\geqslant\Lambda(\ell,\mu)+x+c+\myOmega(\ell,\mu)$ by Remark~\ref{rem:online_comp_constant}
    because $G(x)=\ell$ for all $x\geqslant\ell$.
\end{itemize}
To conclude the proof that $H\in\mygpc$,
note that $w$ is also polynomially bounded
.
\end{proof}

\begin{definition}[Round]\label{def:comp:round}
Let $\crnd\in C^0(\R,\R)$ be the unique function such that:
\begin{itemize}
\item $\crnd(x,\mu)=n$ for all $x\in\left[n-\frac{1}{2}+e^{-\mu},n+\frac{1}{2}-e^{-\mu}\right]$
for all $n\in\Z$
\item $\crnd(x,\mu)$ is affine over $\left[n+\frac{1}{2}-e^{-\mu},n+\frac{1}{2}+e^{-\mu}\right]$
for all $n\in\Z$
\end{itemize}
\end{definition}

\begin{theorem}[Round]\label{th:comp:round}
$\crnd\in\mygpc$.
\end{theorem}

\begin{proof}
The idea of the proof is to build a function computing the ``fractional part'' function,
by this we mean a $1$-periodic function that maps $x$ to $x$ over $[-1+e^{-\mu},1-e^{-\mu}]$
and is affine at the border to be continuous. The rounding function immediately
follows by subtracting the fractional of $x$ to $x$. Although the idea behind this construction is simple, the details are not so immediate. The intuition is that $\frac{1}{2\pi}\arccos(\cos(2\pi x))$
works well over $[0,1/2-e^{-\mu}]$ but needs to be fixed at the border (near $1/2$),
and also its parity needs to be fixed based on the sign of $\sin(2\pi x)$.

Formally, define for $c\in[-1,1]$, $x\in\R$ and $\mu\in\Rp$:
\[
g(c,\mu)=\max(0,\min((1-\tfrac{e^{\mu}}{2})(\arccos(c)-\pi),\arccos(c))),\]
\[
f(x,\mu)=\frac{1}{2\pi}\sgn{\sin(2\pi x)}g(\cos(2\pi x),\mu).
\]
Remark that $g\in\mygpc$ because of Theorem~\ref{prop:clamped_exp}
and that $\arccos\in\mygpc$ because $\arccos\in\gpval$.
Then $f\in\mygpc$ by Proposition~\ref{prop:smooth_sign}. Indeed, if $\sin(2\pi x)=0$ then $g(\cos(2\pi x),\mu)=0$
and if $\sin(2\pi x)\neq0$, a tedious computation shows that
$\left|\tfrac{g(\cos(2\pi x),\mu)}{\sin(2\pi x)}\right|=
\min\left((1-\tfrac{e^\mu}{2})\frac{\arccos(\cos(2\pi x))-\pi}{\sin(2\pi x)},
\frac{\arccos(\cos(2\pi x))}{\sin(2\pi x)}\right)\leqslant2\pi e^{\mu}$
because $g(\cos(2\pi x),\mu)$ is piecewise affine with slope $e^\mu$ at most
(see below for more details).

Note that $f$ is $1$-periodic because of the sine and cosine so we only need to
analyze if over $[-\tfrac{1}{2},\tfrac{1}{2}]$, and since $f$ is an odd function,
we only need to analyze it over $[0,\tfrac{1}{2}]$. Let $x\in[0,\tfrac{1}{2}]$
and $\mu\in\Rp$ then $2\pi x\in[0,\pi]$ thus $\arccos(\cos(2\pi x))=2\pi x$
and $f(x,\mu)=\min((1-\tfrac{e^{\mu}}{2})(x-\tfrac{1}{2}),\tfrac{x}{2\pi})$. There are two cases.
\begin{itemize}
\item If $x\in[0,\tfrac{1}{2}-e^{-\mu}]$ then $x-\tfrac{1}{2}\leqslant-e^{-\mu}$ thus
    $(1-\tfrac{e^{\mu}}{2})(x-\tfrac{1}{2})\geqslant\tfrac{1}{2}-e^{-\mu}\geqslant\tfrac{x}{2\pi}$
    so $f(x,\mu)=x$.
\item If $x\in[\tfrac{1}{2}-e^{-\mu},\tfrac{1}{2}]$ then $0\geqslant x-\tfrac{1}{2}\geqslant-e^{-\mu}$ thus
    $(1-\tfrac{e^{\mu}}{2})(x-\tfrac{1}{2})\leqslant\tfrac{1}{2}-e^{-\mu}\leqslant\tfrac{x}{2\pi}$
    so $f(x,\mu)=(1-\tfrac{e^{\mu}}{2})(x-\tfrac{1}{2})$ which is affine.
\end{itemize}
Finally define $\crnd(x,\mu)=x-f(x,\mu)$ to get the desired function.
\end{proof}

\subsubsection{Some functions considered elsewhere: Norm, and
  Bump functions}\label{subsection:resultsElsewhere}

The following functions have already been considered in some other
articles, and proved to be in $\gpval$ (and hence in $\mygpc$). 

A useful function when dealing with error bound is the norm function.
Although it would be possible to build a very good infinity norm, in practice
we will only need a constant \emph{overapproximation} of it. The following
results can be found in \cite[Lemma 44 and 46]{\INFORMATIONANDCOMPUTATION}.

\begin{lemma}[Norm function]
\label{lem:norm}
For every $\delta\in]0,1]$, there exists $\norm_{\infty,\delta}\in\gpval$
such that for any $x\in\R^n$ we have
\[\infnorm{x}\leqslant\norm_{\infty,\delta}(x)\leqslant\infnorm{x}+\delta.\]
\end{lemma}

A crucial function when simulating computation is a ``step'' or ``bump'' function.
Unfortunately, for continuity reasons, it is again impossible to build a perfect one
but we can achieve a good accuracy except on a small transition interval.

\begin{lemma}[``low-X-high'' and ``high-X-low''
]\label{lem:lxh_hxl}
For every $I=[a,b]$, $a,b \in \K$,  there exists $\lxh_I,\hxl_I\in\gpval$ such that for every $\mu\in\Rp$ and $t,x\in\R$
we have:
\begin{itemize}
\item $\lxh_I$ is of the form $\lxh_I(t,\mu,x)=\phi_1(t,\mu,x)x$ where $\phi_1\in\gpval$,
\item $\hxl_I$ is of the form $\lxh_I(t,\mu,x)=\phi_2(t,\mu,x)x$ where $\phi_2\in\gpval$,
\item if $t\leqslant a, |\lxh_I(t,\mu,x)|\leqslant e^{-\mu}$ and $|x-\hxl_I(t,\mu,x)|\leqslant e^{-\mu}$,
\item if $t\geqslant b, |x-\lxh_I(t,\mu,x)|\leqslant e^{-\mu}$ and $|\hxl_I(t,\mu,x)|\leqslant e^{-\mu}$,
\item in all cases, $|\lxh_I(t,\mu,x)|\leqslant|x|$ and $|\hxl_I(t,\mu,x)|\leqslant|x|$.
\end{itemize}
\end{lemma}

\section{Encoding The Step Function of a Turing machine}\label{sec:pivp_turing:mt}

In this section, we will show how to encode and simulate one step of a Turing machine with a computable
function in a robust way. The empty word will be denoted by $\lambda$. We define the integer part function
$\intp(x)$ by $\max(0,\lfloor x\rfloor)$ and the fractional part function $\fracp(x)$
by $x-\intp{x}$. We also denote by $\card{S}$ the cardinal of a finite set $S$.

\subsection{Turing Machine}

There are many possible definitions of Turing machines. The exact kind we pick is
usually not important but since we are going to simulate one with differential equations,
it is important to specify all the details of the model. We will simulate deterministic,
one-tape Turing machines, with complete transition functions.

\begin{definition}[Turing Machine]\label{def:tm}
A \emph{Turing Machine} is a tuple $\mathcal{M}=(Q,\Sigma,b,\delta,q_0,q_\infty)$
where $Q=\intinterv{0}{m-1}$ are the states of the machines, $\Sigma=\intinterv{0}{k-2}$ is the alphabet
and $b=0$ is the blank symbol, $q_0\in Q$ is the initial state, $q_\infty\in Q$ is the
halting state and $\delta:Q\times\Sigma\rightarrow Q\times\Sigma\times\{L,S,R\}$
is the transition function with $L=-1$, $S=0$ and $R=1$.
We write $\delta_1,\delta_2,\delta_3$ as the components of $\delta$.
That is $\delta(q,\sigma)=(\delta_1(q,\sigma),\delta_2(q,\sigma),\delta_3(q,\sigma))$
where $\delta_1$ is the new state, $\delta_2$ the new symbol and $\delta_3$ the head move direction.
We require that $\delta(q_\infty,\sigma)=(q_\infty,\sigma,S)$.
\end{definition}

\begin{remark}[Choice of $k$]
The choice of $\Sigma=\intinterv{0}{k-2}$ will be crucial for the simulation, to
ensure that the transition function is continuous. See Lemma~\ref{lem:tm_config_range}.
\end{remark}

For completeness, and also to make the statements of the next theorems easier,
we introduce the notion of configuration of a machine, and define one step
of a machine on configurations. This allows us to define the result of a computation.
Since we will characterize $\FP$, our machines not only accept or reject a word,
but compute an output word.

\begin{definition}[Configuration]\label{def:tm_config}
A \emph{configuration} of $\mathcal{M}$ is a tuple $c=(x,\sigma,y,q)$
where $x\in\Sigma^*$ is the part of the tape at left of the head,
$y\in\Sigma^*$ is the part at the right, $\sigma\in\Sigma$ is the symbol under the head and $q\in Q$ the current state. 
More precisely $x_1$ is the symbol immediately at the left of the head and $y_1$ the symbol immediately at the right.
See Figure~\ref{fig:tm_config} for a graphical representation. The set of configurations
of $\mathcal{M}$ is denoted by $\machcfg{\mathcal{M}}$. The \emph{initial configuration}
is defined by $c_0(w)=(\emptyword,b,w,q_0)$ and the \emph{final configuration}
by $c_\infty(w)=(\emptyword,b,w,q_\infty)$ where $\emptyword$ is the empty word.
\end{definition}

\begin{definition}[Step]\label{def:tm_config_step}
The \emph{step} function of a Turing machine $\mathcal{M}$ is the function, acting
on configurations, denoted by $\machstep{\mathcal{M}}$ and defined by:
\[\mathcal{M}(x,\sigma,y,q)=\begin{cases}
(\emptyword,b,\sigma'y,q') & \text{if }d=L\text{ and }x=\emptyword\\
(x_{2..|x|},x_1,\sigma'y,q') & \text{if }d=L\text{ and }x\neq\emptyword\\
(x,\sigma',y,q') & \text{if }d=S\\
(\sigma'x,b,\emptyword,q') & \text{if }d=R\text{ and }y=\emptyword\\
(\sigma'x,y_1,y_{2..|y|},q') & \text{if }d=R\text{ and }y\neq\emptyword\\
\end{cases}\quad\text{where }
\left\{\begin{array}{@{}l@{}}q'=\delta_1(q,\sigma)\\
\sigma'=\delta_2(q,\sigma)\\
d=\delta_3(q,\sigma)
\end{array}\right..
\]
\end{definition}

\begin{definition}[Result of a computation]\label{def:tm_compute}
The \emph{result of a computation} of $\mathcal{M}$ on a word $w\in\Sigma^*$
is defined by:
\[\mathcal{M}(w)=\begin{cases}
x&\text{if }\exists n\in\N, \fiter{\mathcal{M}}{n}(c_0(w))=c_\infty(x)\\
\bot&\text{otherwise}
\end{cases}\]
\end{definition}

\begin{remark}
The result of a computation is well-defined because we imposed that when a machine
reaches a halting state, it does not move, change state or change the symbol
under the head.
\end{remark}

\begin{figure}
\begin{center}
\begin{tikzpicture}[scale=1.3]
\draw (0.05,0) rectangle (0.45,0.5);
\draw (0.25,0.25) node {$\sigma$};

\draw (0.5,0) rectangle (1,0.5);
\draw (0.75,0.25) node {$y_1$};
\draw (1,0) rectangle (1.5,0.5);
\draw (1.25,0.25) node {$y_2$};
\draw (1.5,0) rectangle (2.0,0.5);
\draw (1.75,0.25) node {$y_3$};
\draw (2.0,0) rectangle (2.5,0.5);
\draw (2.25,0.25) node {$\cdots$};
\draw (2.5,0) rectangle (3.0,0.5);
\draw (2.75,0.25) node {$y_k$};

\draw (0,0) rectangle (-0.5,0.5);
\draw (-0.25,0.25) node {$x_1$};
\draw (-0.5,0) rectangle (-1,0.5);
\draw (-0.75,0.25) node {$x_2$};
\draw (-1,0) rectangle (-1.5,0.5);
\draw (-1.25,0.25) node {$x_3$};
\draw (-1.5,0) rectangle (-2,0.5);
\draw (-1.75,0.25) node {$\cdots$};
\draw (-2.0,0) rectangle (-2.5,0.5);
\draw (-2.25,0.25) node {$x_l$};

\draw[->,very thick] (0.25,0.8) -- (0.25,0.5);
\draw (0.25,0.8) node[above] {$q$};
\end{tikzpicture}
\end{center}
\caption{Example of generic configuration $c=(x,\sigma,y,q)$\label{fig:tm_config}}
\end{figure}
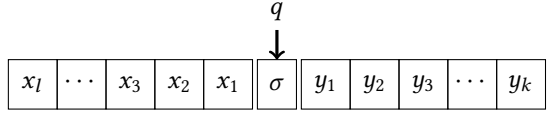

\subsection{Finite set interpolation}

In order to implement the transition function of the Turing Machine,
we will use an interpolation scheme.

\begin{lemma}[Finite set interpolation]\label{lem:lagrange_interp}
For any finite $G\subseteq\K^d$ and $f:G\rightarrow\K$,  there exists $\lagrange{f}\in\mygpc$
with $\frestrict{\lagrange{f}}{G}=f$, where
$\frestrict{\lagrange{f}}{G}$ denotes restriction of $f$ to $G$. 
\end{lemma}

\begin{proof}
For $d=1$, consider for example Lagrange
polynomial \[\lagrange{f}(x)=\sum_{\bar{x}\in
    G}f(\bar{x})\prod_{\substack{y\in
      G\\y\neq\bar{x}}}\prod_{i=1}^d\frac{x_i-y_i}{\bar{x}_i-y_i}.\]
The fact that $\lagrange{f}$ matches $f$ on $G$ is a classical calculation. Also
$\lagrange{f}$ is a polynomial with coefficients in $\K$ so clearly it
belongs to $\mygpc$. The generalization to $d>1$ is clear, but
tedious to be fully detailed so we leave it to the reader. 
\end{proof}

It is customary to prove robustness of the interpolation, which means that on
the neighborhood of $G$, $\lagrange{f}$ is nearly constant. However this result
is a byproduct of the effective continuity of $\lagrange{f}$, thanks to Theorem~\ref{th:comp_implies_cont}.

We will often need to interpolate characteristic functions,
that is polynomials that value $1$ when $f(x)=a$ and $0$ otherwise.
For convenience we define a special notation for it.

\begin{definition}[Characteristic interpolation]
Let $f:G\rightarrow\R$ where $G$ is a finite subset of $\R^d$, $\alpha\in\R$, and define:
\[\lagreq{f}{\alpha}(x)=\lagrange{f_\alpha}(x)\quad\text{and}\quad
\lagrneq{f}{\alpha}(x)=\lagrange{1-f_\alpha}(x)\]
where
\[f_\alpha(x)=\begin{cases}1&\text{if }f(x)=\alpha\\0&\text{otherwise}\end{cases}.\]
\end{definition}

\begin{lemma}[Characteristic interpolation]\label{lem:char_interp}
For any finite set $G\subseteq\K^d$, $f:G\rightarrow\K$ and $\alpha\in\K$, $\lagreq{f}{\alpha},\lagrneq{f}{\alpha}\in\mygpc$.
\end{lemma}

\begin{proof}
Observe that $f_\alpha:G\rightarrow\{0,1\}$ and $\{0,1\}\subseteq\K$. Apply Lemma~\ref{lem:lagrange_interp}.
\end{proof}

\subsection{Encoding}\label{sec:encoding}

In order to simulate a machine, we will need to encode configurations with real numbers.
There are several ways of doing so but not all of them are suitable for use when
proving complexity results. This particular issue is discussed in Remark~\ref{rem:encoding_size}.
For our purpose, it is sufficient to say that we will encode a configuration as a tuple,
we store the state and current letter as integers and the left and right parts of
the tape as real numbers between $0$ and $1$. Intuitively, the tape is represented
as two numbers whose digits in a particular basis are the letters of the tape.
Recall that the alphabet is $\Sigma=\intinterv{0}{k-2}$.

\begin{definition}[Real encoding]\label{def:tm_config_real_enc}
Let $c=(x,\sigma,y,q)$ be a configuration of $\mathcal{M}$, the \emph{real encoding} of $c$ is
$\realenc{c}=(0.x,\sigma,0.y,q)\in\Q\times\Sigma\times\Q\times Q$ where $0.x=x_1k^{-1}+x_2k^{-2}+\cdots+x_{|w|}k^{-|w|}\in\Q$.
\end{definition}

\begin{lemma}[Encoding range]\label{lem:tm_config_range}
For any word $x\in\intinterv{0}{k-2}^*$, $0.x\in\left[0,\frac{k-1}{k}\right]$.
\end{lemma}

\begin{proof}
$0\leqslant0.x=\sum_{i=1}^{|x|}x_ik^{-i}
\leqslant\sum_{i=1}^{\infty}(k-2)k^{-i}\leqslant\frac{k-2}{k-1}\leqslant\frac{k-1}{k}$.
\end{proof}

The same way we defined the step function for Turing machines on configurations,
we have to define a step function that works directly the encoding of configuration.
This function is ideal in the sense that it is only defined over real numbers that
are encoding of configurations.

\begin{definition}[Ideal real step]\label{def:tm_config_ideal_real_step}
The \emph{ideal real step} function of a Turing machine $\mathcal{M}$ is the function
defined over $\realenc{\machcfg{\mathcal{M}}}$ by:
\[\idealrealstep{\mathcal{M}}(\tilde{x},\sigma,\tilde{y},q)=\begin{cases}
\left(\fracp(k\tilde{x}),\intp(k\tilde{x}),\frac{\sigma'+\tilde{y}}{k},q'\right) & \text{if }d=L\\
\left(\tilde{x},\sigma',\tilde{y},q'\right) & \text{if }d=S\\
\left(\frac{\sigma'+\tilde{x}}{k},\intp(k\tilde{y}),\fracp(k\tilde{y}),q'\right) & \text{if }d=R\\
\end{cases}\quad\text{where }
\left\{\begin{array}{@{}c@{}l@{}}q'&=\delta_1(q,\sigma)\\
\sigma'&=\delta_2(q,\sigma)\\
d&=\delta_3(q,\sigma)
\end{array}\right..
\]
\end{definition}

\begin{lemma}[$\idealrealstep{\mathcal{M}}$ is correct]\label{lem:real_step_correct}
For any machine $\mathcal{M}$ and configuration $c$,
$\idealrealstep{\mathcal{M}}(\realenc{c})=\realenc{\machstep{\mathcal{M}}(c)}$.
\end{lemma}

\begin{proof}
Let $c=(x,\sigma,y,q)$ and $\tilde{x}=0.x$.
The proof boils down to a case analysis (the analysis is the same for $x$ and $y$):
\begin{itemize}
\item If $x=\emptyword$ then $\tilde{x}=0$ so $\intp(k\tilde{x})=b$ and $\fracp(k\tilde{x})=0=0.\emptyword$ because $b=0$.
\item If $x\neq\emptyword$, $\intp(k\tilde{x})=x_1$ and $\fracp(k\tilde{x})=0.x_{2..|x|}$ because $k\tilde{x}=x_1+0.x_{2..|x|}$
and Lemma~\ref{lem:tm_config_range}.
\end{itemize}
\end{proof}

The previous function was ideal but this is not enough to simulate a
machine: We need 
a step function robust to small perturbations and computable.
For this reason, we define a new step function with both features and that relates
closely to the ideal function.

\newcommand\chose{
\[\intp^*(x)=\crnd\left(x-\tfrac{1}{2}+\tfrac{1}{2k},\mu+\ln k\right)\qquad\fracp^*(x)=x-\intp^*(x),\]
%
}

\begin{definition}[Real step]\label{def:tm_config_real_step}
For any $\bar{x},\bar{\sigma},\bar{y},\bar{q}\in\R$ and $\mu\in\Rp$, define the \emph{real step} function
of a Turing machine $\mathcal{M}$ by:
\[\realstep{\mathcal{M}}(\bar{x},\bar{\sigma},\bar{y},\bar{q},\mu)=
\realstep{\mathcal{M}}^*(\bar{x},\crnd(\bar{\sigma},\mu),\bar{y},\crnd(\bar{q},\mu),\mu)\]
%
where
\[\realstep{\mathcal{M}}^*(\bar{x},\bar{\sigma},\bar{y},\bar{q},\param{\mu})=
\realstep{\mathcal{M}}^{\star}\big(\bar{x},\bar{y},\lagrange{\delta_1}(\bar{q},\bar{\sigma}),
\lagrange{\delta_2}(\bar{q},\bar{\sigma}),\lagrange{\delta_2}(\bar{q},\bar{\sigma}),\param{\mu}\big)\]
where
\[\realstep{\mathcal{M}}^{\star}\big(\bar{x},\bar{y},\bar{q},\bar{\sigma},\bar{d},\param{\mu}\big)=
\begin{pmatrix}
\myop{choose}\left[\fracp^*(k\bar{x}),\bar{x},\frac{\bar{\sigma}+\bar{x}}{k}\right]\\
\myop{choose}\left[\intp^*(k\bar{x}),\bar{\sigma},\intp^*(k\bar{y})\right]\\
\myop{choose}\left[\frac{\bar{\sigma}+\bar{y}}{k},\bar{y},\fracp^*(k\bar{y})\right]\\
\bar{q}
\end{pmatrix}\]
where

\[\myop{choose}[l,s,r]=\lagreq{\idfun}{L}(\bar{d})l+\lagreq{\idfun}{S}(\bar{d})s+\lagreq{\idfun}{R}(\bar{d})r,\]
%
\chose
\[\crnd\text{ is defined in Definition~\ref{def:comp:round}}.\]
\end{definition}

\begin{theorem}[Real step is robust]\label{th:real_setp_robust}
For any machine $\mathcal{M}$, $c\in\machcfg{\mathcal{M}}$, $\mu\in\Rp$
and $\bar{c}\in\R^4$, if $\infnorm{\realenc{c}-\bar{c}}\leqslant\frac{1}{2\cramped{k^2}}-e^{-\mu}$ then
$\infnorm{\realstep{\mathcal{M}}(\bar{c},\mu)-\realenc{\machstep{\mathcal{M}}(c)}}
\leqslant k\infnorm{\realenc{c}-\bar{c}}$. Furthermore $\realstep{\mathcal{M}}\in\mygpc$.
\end{theorem}

\begin{proof}
We begin by a small result about $\intp^*$ and $\fracp^*$:
if $\infnorm{\bar{x}-0.x}\leqslant\frac{1}{2\cramped{k^2}}-e^{-\mu}$
then $\intp^*(k\bar{x})=\intp(k0.x)$ and $\infnorm{\fracp^*(k\bar{x})-\fracp(k0.x)}\leqslant k\infnorm{\bar{x}-0.x}$.
Indeed, by Lemma~\ref{lem:tm_config_range}, $k0.x=n+\alpha$ where $n\in\N$ and $\alpha\in\left[0,\frac{k-1}{k}\right]$.
Thus $\intp^*(k\bar{x})=\crnd\left(k\bar{x}-\frac{1}{2}+\frac{1}{2k},\mu\right)=n$
because $\alpha+k\infnorm{\bar{x}-0.x}-\frac{1}{2}+\frac{1}{2k}\in
\left[-\frac{1}{2}+ke^{-\mu},\frac{1}{2}-ke^{-\mu}\right]$. Also,
$\fracp^*(k\bar{x})=k\bar{x}-\intp^*(k\bar{x})=k\infnorm{\bar{x}-0.x}+kx-\intp(kx)=\fracp(kx)+k\infnorm{\bar{x}-0.x}$.

Write $\realenc{c}=(x,\sigma,y,q)$ and $\bar{c}=(\bar{x},\bar{\sigma},\bar{y},\bar{q})$.
Apply Definition~\ref{def:comp:round} to get that $\crnd(\bar{\sigma},\mu)=\sigma$ and $\crnd(\bar{q},\mu)=q$
because $\infnorm{(\bar{\sigma},\bar{q})-(\sigma,q)}\leqslant\frac{1}{2}-e^{-\mu}$.
Consequently, $\lagrange{\delta_i}(\bar{q},\bar{\sigma})=\delta_i(q,\sigma)$
and $\realstep{\mathcal{M}}(\bar{c},\mu)=\realstep{\mathcal{M}}^{\star}(\bar{x},\bar{y},q',\sigma',d')$
where $q'=\delta_1(q,\sigma)$, $\sigma'=\delta_2(q,\sigma)$ and $d'=\delta_3(q,\sigma)$.
In particular $d'\in\{L,S,R\}$ so there are three cases to analyze.
\begin{itemize}
\item \textbf{If $d'=L$ then} $\myop{choose}[l,s,r]=l$, $\intp^*(k\bar{x})=\intp(kx)$,
$\infnorm{\fracp^*(k\bar{x})-\fracp(kx)}\leqslant k\infnorm{\bar{x}-x}$ and
$\infnorm{\frac{\sigma'+\bar{y}}{k}-\frac{\sigma'+y}{k}}\leqslant\infnorm{\bar{x}-x}$.
Thus $\infnorm{\realstep{\mathcal{M}}(\bar{c},\mu)-\idealrealstep{\mathcal{M}}(\realenc{c})}\leqslant
k\infnorm{\bar{c}-\realenc{c}}$. Conclude using Lemma~\ref{lem:real_step_correct}.
\item \textbf{If $d'=S$ then} $\myop{choose}[l,s,r]=s$ so we immediately have that
$\infnorm{\realstep{\mathcal{M}}(\bar{c},\mu)-\idealrealstep{\mathcal{M}}(\realenc{c})}\leqslant
\infnorm{\bar{c}-\realenc{c}}$. Conclude using Lemma~\ref{lem:real_step_correct}.
\item \textbf{If $d'=R$ then} $\myop{choose}[l,s,r]=r$ and everything else is similar to
the case of $d'=L$.
\end{itemize}
Finally apply Lemma~\ref{lem:lagrange_interp}, Theorem~\ref{th:comp:round},
Theorem~\ref{th:gpac_comp_arith} and Theorem~\ref{th:gpac_comp_composition}
to get that $\realstep{\mathcal{M}}\in\mygpc$.
\end{proof}

\section{A Characterization of $FP$}\label{sec:pivp_turing:equiv}

We will now provide a characterization of $\FP$ by introducing a
notion of function \emph{emulation}. This characterization builds on our
notion of computability introduced previously.

In this section, we fix an alphabet $\Gamma$ and all languages are considered over
$\Gamma$
. It is common to take $\Gamma=\{0,1\}$
but the proofs work for any finite alphabet. We will assume that
$\Gamma$ comes with an 
injective mapping $\gamma:\Gamma\rightarrow\N\setminus \{0\}$, in other words every letter has an
uniquely assigned positive number. By extension, $\gamma$ applies letterwise over words.

\subsection{Main statement}



\begin{definition}[Discrete emulation]\label{def:fp_gpac_embed}
$f:\Gamma^*\rightarrow\Gamma^*$ is called $\K$-\emph{emulable} if 
there exists $g\in\mygpc[\K]$
and $k\geqslant1+\max(\gamma(\Gamma))$ such that for any word $w\in\Gamma^*$:
\[g(\psi_k(w))=\psi_k(f(w))\qquad\text{where}\quad
\psi_k(w)=\left(\sum_{i=1}^{|w|}\gamma(w_i)k^{-i},|w|\right).\]
We say that $g$ $\K$-\emph{emulates} $f$ with $k$. When the field $\K$ is
unambiguous, we will simply say that $f$ is emulable.
\end{definition}


\begin{remark}[Encoding length]\label{rem:encoding_size}
The exact details of the encoding $\psi$ chosen in the definition above are not extremely important, however the length
of the encoding is crucial. More precisely, the proof heavily relies on the fact
that $\infnorm{\psi(w)}\approx|w|$. Note that this works both ways:
\begin{itemize}
\item $\infnorm{\psi(w)}$ must be polynomially bounded in $|w|$ so that
a simulation of the system 
runs in polynomial time in $|w|$.
\item $\infnorm{\psi(w)}$ must be polynomially lower bounded in $|w|$ so
that we can recover the output length from the length of its encoding.
\end{itemize}
\end{remark}


The sef $\FP$ of polynomial-time
computable functions can then be characterized as follows. 

\begin{theorem}[$\FP$ equivalence]\label{th:fp_gpac}
For any generable field $\K$ such that $\Rgen\subseteq\K\subseteq\Rpoly$
and $f:\Gamma^*\rightarrow\Gamma^*$,
$f\in\FP$ if and only if $f$ is $\K$-emulable (with $k=2+\max(\gamma(\Gamma))$).
\end{theorem}

The rest of this section is devoted to the proof of Theorem \ref{th:fp_gpac}

\subsection{Reverse direction of Theorem \ref{th:fp_gpac}}

The reverse direction of the equivalence between Turing machines and analog systems will
involve polynomial initial value problems such as \eqref{eq:gpac}. 

\subsubsection{Complexity of solving polynomial differential equations}
\label{sec:solving_pivp}

The complexity
of solving this kind of differential equation has been heavily studied over compact
domains but there are few results over unbounded domains. In \cite{\PAPIERODETCS} we
studied the complexity of this problem over unbounded domains and obtained a
bound that involved the length of the solution curve. In \cite{\PAPIERODE},
we extended this result to work with any real inputs (and not just rationals)
in the framework of Computable Analysis.

We need a few notations to state the result.
For any multivariate polynomial $p(x)=\sum_{|\alpha|\leqslant k} a_\alpha x^\alpha$,
we call $k$ the degree if $k$ is the minimal integer $k$ for which the condition   $p(x)=\sum_{|\alpha|\leqslant k} a_\alpha x^\alpha$ holds and we denote the sum of the norm of the coefficients by
$\sigmap{p}=\sum_{|\alpha|\leqslant k}|a_\alpha|$ (also known as the length of $p$).
For a vector of polynomials,
we define the degree and $\sigmap{p}$ as the maximum over all components.
For any continuous function $y$ and polynomial $p$ define
the \emph{pseudo-length}
\[\LenI_{y,p}(a,b)=\int_a^b\sigmap{p}\max(1,\infnorm{y(u)})^{\degp{p}}du.\]

\begin{theorem}[
  \cite{\PAPIERODETCS}, \cite{\PAPIERODE}]\label{th:pivp_comp_analysis}
Let $I=[a,b]$ be an interval, $p\in\R^n[\R^{n}]$ and $k$ its degree and $y_0\in\R^n$.
Assume that $y:I\rightarrow\R^n$ satisfies for all $t\in I$ that
\begin{equation}\label{eq:ode}
y(a)=y_0\qquad y'(t)=p(y(t)),
\end{equation}
then $y(b)$ can be computed with precision $2^{-\mu}$ in time bounded by
\begin{equation}\label{eq:pivp_comp_analysis_bound}
\poly(k,\LenI_{y,p}(a,b),\log\infnorm{y_0},\log\sigmap{p},\mu)^n.
\end{equation}
More precisely, there exists a Turing machine $\mathcal{M}$ such that for any oracle
$\mathcal{O}$ representing\footnote{See \cite{Ko91} for more details. In short, the machine can ask arbitrary approximations
of $a, y_0, p$ and $b$ to the oracle. The polynomial is represented by the finite list of coefficients.} $(a,y_0,p,b)$ and any $\mu\in\N$,
$\infnorm{\mathcal{M}^\mathcal{O}(\mu)-y(b)}\leqslant2^{-\mu}$
where $y$ satisfies \eqref{eq:ode}, and the number of steps of the machine is bounded by \eqref{eq:pivp_comp_analysis_bound}
for all such oracles.
\end{theorem}

Finally, we would like to remind the reader that the existence of a solution $y$
of a PIVP up to a given time is undecidable, see \cite{GBC07} for more details. This explains
why, in the previous theorem, we have to assume the existence of the solution if
we want to have any hope of computing it.


\subsubsection{Proof of Reverse direction of Theorem \ref{th:fp_gpac}}

Assume that $f$ is $\Rpoly$-emulable and apply Definition~\ref{def:fp_gpac_embed} to get
$g\in\gc{\Upsilon}{\myOmega}$ where $\Upsilon,\myOmega$ are polynomials,
with respective $d,p,q$.
 Let $w\in\Gamma^*$: we will describe an $\FP$ algorithm
to compute $f(w)$. Consider the following system:
\[y(0)=q(\psi_k(w))\qquad y'(t)=p(y(t)).\]
Note that, by construction, $y$ is defined over $\Rp$. Also note, 
that the coefficients of $p,q$ belong to $\Rpoly$ which means
that they are polynomial time computable. And since $\psi_k(w)$ is a pair of rational numbers
with polynomial length (with respect to $|w|$), then $q(\psi_k(w))\in\Rpoly^d$.

The algorithm works in two steps: first we compute a rough approximation of the output
to guess the length of the output. Then we rerun the system with enough precision to
get the full output.

Let $t_w=\myOmega(|w|,2)$ for any $w\in\Sigma^*$. Note that $t_w\in\Rpoly$ and that
it is polynomially bounded in $|w|$ because $\myOmega$ is a polynomial. Apply Theorem~\ref{th:pivp_comp_analysis}
to compute $\tilde{y}$ such that $\infnorm{\tilde{y}-y(t_w)}\leqslant e^{-2}$: this takes a time polynomial
in $|w|$ because $t_w$ is polynomially bounded
and because\footnote{See Section~\ref{sec:solving_pivp} for the expression $\LenI$.} $\LenI_{y,p}(0,t_w)\leqslant\poly(t_w,\sup_{[0,t_w]}\infnorm{y})$
and by construction, $\infnorm{y(t)}\leqslant\Upsilon(\infnorm{\psi_k(w)},t_w)$ for $t\in[0,t_w]$
where $\Upsilon$ is a polynomial. Furthermore, by definition of $t_w$, $\infnorm{y(t_w)-g(\psi_k(w))}\leqslant e^{-2}$
thus $\infnorm{\tilde{y}-\psi_k(f(w))}\leqslant2e^{-2}\leqslant\frac{1}{3}$. But
since $\psi_k(f(w))=(0.\gamma(f(w)),|f(w)|)$, from $\tilde{y}_2$ we can find $|f(w)|$ by
rounding to the closest integer (which is unique because it is within distance at most $\frac{1}{3}$).
In other words, we can compute $|f(w)|$ in polynomial time in $|w|$. Note that this
implies that $|f(w)|$ is at most polynomial in $|w|$.

Let $t_w'=\myOmega(|w|,2+|f(w)|\ln k)$ which is polynomial in $|w|$ because $\myOmega$
is a polynomial and $|f(w)|$ is at most polynomial in $|w|$. We can use the
same reasoning and apply Theorem~\ref{th:pivp_comp_analysis} to get $\tilde{y}$ such
that $\infnorm{\tilde{y}-y(t_w')}\leqslant e^{-2-|f(w)|\ln k}$. Again this
takes a time polynomial in $|w|$. Furthermore,
$\infnorm{\tilde{y}_1-0.\gamma(f(w))}\leqslant2e^{-2-|f(w)|\ln k}\leqslant \frac{1}{3}k^{-|f(w)|}$.
We claim that this allows to recover $f(w)$ unambiguously in polynomial time in $|f(w)|$.
Indeed, it implies that $\infnorm{k^{|f(w)|}\tilde{y}_1-k^{|f(w)|}0.\gamma(f(w))}\leqslant\frac{1}{3}$.
Unfolding the definition shows that $k^{|f(w)|}0.\gamma(f(w))=\sum_{i=1}^{|f(w)|}\gamma(f(w)_i)k^{|f(w)|-i}\in\N$
thus by rounding $k^{|f(w)|}\tilde{y}_1$ to the nearest integer, we recover $\gamma(f(w))$,
and then $f(w)$. This is all done in polynomial time in $|f(w)|$, which proves that $f$
is polynomial time computable.

\subsection{Direct direction of Theorem \ref{th:fp_gpac}}


\subsubsection{Iterating a function}

The direct direction of  of the equivalence between Turing machines and analog systems will
involve iterations of the robust real step associated to a Turing
machine of previous section.

We now state that iterating a function is computable
under reasonable assumptions. Iteration is a powerful operation, which is why
reasonable complexity classes are never closed under unrestricted iteration.
If we want to keep to polynomial-time computability for Computable Analysis, there are at least two immediate necessary conditions:
the iterates cannot grow faster than a polynomial and the
iterates must keep a polynomial modulus of continuity. The optimality of the
conditions of next theorem is discussed in Remark~\ref{rem:gpac_iter_opt_growth} and Remark~\ref{rem:gpac_iter_opt_mod}.
However there is the subtler issue of the domain of definition that comes into play
and is discussed in Remark~\ref{rem:gpac_iter_opt_domain}.

In short, the conditions to iterate a function can be summarized as follows:
\begin{itemize}
\item $f$ has domain of definition $I$;
\item there are subsets $I_n$ of $I$ such that points of $I_n$ can be iterated
up to $n$ times;
\item the iterates of $f$ on $x$ over $I_n$ grow at most polynomially in $\infnorm{x}$ and $n$;
\item each point $x$ in $I_n$ has an open neighborhood in $I$ of size at least $e^{-\poly(\infnorm{x})}$
and $f$ has modulus of continuity of the form $\poly(\infnorm{x})+\mu$ over this set.
\end{itemize}

Formally: 



\begin{theorem}[Simulating Discrete by Continuous Time]\label{th:gpac_comp_iter}
Let $I\subseteq\R^m$, $(f:I\rightarrow \R^m)\in\mygpc$, $\eta\in\left[0,1/2\right[$ and
assume that there exists a family of subsets $I_n\subseteq I$, for all $n\in\N$ and polynomials
$\mho:\Rp\rightarrow\Rp$ and $\Pi:\Rp^2\rightarrow\Rp$ such that for all $n\in\N$:
\begin{itemize}
\item $I_{n+1}\subseteq I_n$ and $f(I_{n+1})\subseteq I_n$
\item $\forall x\in\ I_n$, $\infnorm{f^{[n]}(x)}\leqslant\Pi(\infnorm{x},n)$
\item $\forall x\in I_n$, $y\in\R^m, \mu\in\Rp$, if
    $\infnorm{x-y}\leqslant e^{-\mho(\infnorm{x})-\mu}$
    then $y\in I$ and $\infnorm{f(x)-f(y)}\leqslant e^{-\mu}$
\end{itemize}
Define $f_\eta^*(x,u)=f^{[n]}(x)$ for $x\in I_n$, $u\in[n-\eta,n+\eta]$ and $n\in\N$. Then
$f_\eta^*\in\mygpc$.
\end{theorem}

This result is far from beeing trivial, and the whole  Section
\ref{sec: gpac_comp_iter} is devoted to its proof.

\begin{remark}[Optimality of growth constraint]\label{rem:gpac_iter_opt_growth}
It is easy to see that without any restriction, the iterates can produce an exponential
function. Pick $f(x)=2x$ then $f\in\mygpc$ and $\fiter{f}{n}(x)=2^nx$ which is clearly not polynomial
in $x$ and $n$. More generally, 
it is necessary
that $f^*$ be polynomially bounded so clearly $f^{[n]}(x)$ must be polynomially
bounded in $\infnorm{x}$ and $n$.
\end{remark}

\begin{remark}[Optimality of modulus constraint]\label{rem:gpac_iter_opt_mod}
Without any constraint, it is easy to build an iterated function with exponential
modulus of continuity. Define $f(x)=\sqrt{x}$ then
$f$ can be shown to be in $\mygpc$  and $\fiter{f}{n}(x)=x^{\frac{1}{2^n}}$.
For any $\mu\in\R$, $\fiter{f}{n}(e^{-2^n\mu})-\fiter{f}{n}(0)=(e^{-2^n\mu})^{\frac{1}{2^n}}=e^{-\mu}$.
Thus $f^*$ has exponential modulus of continuity in $n$.
\end{remark}

\begin{remark}[Domain of definition]\label{rem:gpac_iter_opt_domain}
Intuitively we would have written the theorem differently, only requesting that $f(I)\subseteq I$,
however this has some problems. First if $I$ is discrete, the iterated
modulus of continuity becomes useless and the theorem is false. Indeed, define $f(x,k)=(\sqrt{x},k+1)$
and $I=\left\{(\sqrt[2^n]{e},n),n\in\N\right\}$: $\frestrict{f}{I}$ has polynomial
modulus of continuity $\mho$ because $I$ is discrete, yet $\frestrict{f^*}{I}\notin\mygpc$
as we saw in Remark~\ref{rem:gpac_iter_opt_mod}. But in reality, the problem is more subtle than
that because if $I$ is open but the neighborhood of each point is too small, a
polynomial system cannot take advantage of it. To illustrate this issue, define
$I_n=\left]0,\sqrt[2^n]{e}\right[\times\left]n-\tfrac{1}{4},n+\tfrac{1}{4}\right[$
and $I=\cup_{n\in\N}I_n$. Clearly $f(I_n)=I_{n+1}$ so $I$ is $f$-stable but
$\frestrict{f^*}{I}\notin\mygpc$ for the same reason as before.
\end{remark}

\begin{remark}[Classical error bound]\label{rem:gpac_iter_classic_err}
The third condition in Theorem~\ref{th:gpac_comp_iter} is usually far more subtle than necessary.
In practice, is it useful to note this condition is satisfied if $f$ verifies for some
constants $\varepsilon,K>0$ that
\[\text{for all }x\in I_n\text{ and }y\in\R^m,\text{ if }
    \infnorm{x-y}\leqslant \varepsilon\text{ then }
    y\in I\text{ and }\infnorm{f(x)-f(y)}\leqslant K\infnorm{x-y}.\]
\end{remark}

\begin{remark}[Dependency of $\mho$ in $n$]\label{rem:gpac_iter_modulus_n}
In the statement of the theorem, $\mho$ is only allowed to depend on $\infnorm{x}$ whereas
it might be useful to also make it depend on $n$. In fact the theorem is still true
if the last condition is modified to be $\infnorm{x-y}\leqslant e^{-\mho(\infnorm{x},n)-\mu}$.
One way of showing this is to explicitly add $n$ to the domain of definition by taking
$h(x,k)=(f(x),k-1)$ and to take $I_n'=I_n\times[n,+\infty[$ for example.
\end{remark}

\subsubsection{Proof of Direct direction of Theorem \ref{th:fp_gpac}}

Let $f\in\FP$, then there exists a Turing machine
$\mathcal{M}=(Q,\Sigma,b,\delta,q_0,F)$ where $\Sigma=\intinterv{0}{k-2}$ and
$\gamma(\Gamma)\subset\Sigma\setminus\{b\}$,
and a polynomial $p_\mathcal{M}$ such that for any word $w\in\Gamma^*$,
$\mathcal{M}$ halts in at most $p_\mathcal{M}(|w|)$ steps, that is
$\fiter{\mathcal{M}}{p_\mathcal{M}(|w|)}(c_0(\gamma(w)))=c_\infty(\gamma(f(w)))$.
Note that we assume that $p_\mathcal{M}(\N)\subseteq\N$. Also note that $\psi_k(w)=(0.\gamma(w),|w|)$
for any word $w\in\Gamma^*$.

{
Define $\mu=\ln(4\cramped{k^2})$} and $h(c)=\realstep{\mathcal{M}}(c,\mu)$ for all $c\in\R^4$.
Define $I_\infty=\realenc{\machcfg{\mathcal{M}}}$ and
$I_n=I_\infty+\left[-\varepsilon_n,\varepsilon_n\right]^4$
{where $\varepsilon_n=\tfrac{1}{4\cramped{k^{2+n}}}$ for all $n\in\N$. Note that
$\varepsilon_{n+1}\leqslant\tfrac{\varepsilon_n}{k}$ and
that $\varepsilon_0\leqslant\tfrac{1}{2\cramped{k^2}}-e^{-\mu}$.}
By Theorem~\ref{th:real_setp_robust} we have $h\in \mygpc$ and $h(I_{n+1})\subseteq I_n$.
In particular
{
$\infnorm{\fiter{h}{n}(\bar{c})-\fiter{h}{n}(c)}\leqslant
k^n\infnorm{c-\bar{c}}$
}
for all $c\in I_\infty$ and $\bar{c}\in I_n$, for all $n\in\N$.
{Let $\delta\in\left[0,\tfrac{1}{2}\right[$ and define
$J=\cup_{n\in\N}I_n\times[n-\delta,n+\delta]$}. Apply Theorem~\ref{th:gpac_comp_iter} to
get $(h^*_\delta:J\rightarrow I_0)\in\mygpc$ such that for all $c\in I_\infty$ and $n\in\N$
and $h^*_\delta(c,n)=\fiter{h}{n}(c)$.

Let $\pi_i$
denote the $i^{th}$ projection, that is $\pi_i(x)=x_i$, then $\pi_i\in\mygpc$.
Define
\[g(y,\ell)=\pi_3(h^*_\delta(0,b,\pi_1(y),q_0,p_\mathcal{M}(\ell)))\]
for $y\in \psi_k(\Gamma^*)$ and $\ell\in\N$.
Note that $g\in\mygpc$ and is well-defined. 
Indeed, if $\ell\in\N$ then $p_\mathcal{M}(\ell)\in\N$
and if $y=\psi_k(w)$ then $\pi_1(y)=0.\gamma(w)$ then $(0,b,\pi_1(y),q_0)=\realenc{(\emptyword,b,w,q_0)}=\realenc{c_0(w)}\in I_\infty$.
Furthermore, by construction, for any word $w\in\Gamma^*$ we have:
\begin{align*}
g(\psi_k(w),|w|)&=\pi_3\left(h^*_\delta(\realenc{c_0(w)},p_\mathcal{M}(|w|))\right)\\
    &=\pi_3\left(\fiter{h}{p_\mathcal{M}(|w|)}(c_0(w))\right)\\
    &=\pi_3\left(\realenc{\fiter{\machcfg{\mathcal{M}}}{p_\mathcal{M}(|w|)}(c_0(w))}\right)\\
    &=\pi_3\left(\realenc{c_\infty(\gamma(f(w)))}\right)\\
    &=0.\gamma(f(w))=\pi_1(\psi_k(f(w))).
\end{align*}
Recall that to show emulation, we need to compute $\psi_k(f(w))$ and so far we only
have the first component: the output tape encoding, but we miss the second component: its length.
Since the length of the tape cannot be greater than the initial length plus
the number of steps, we have that
$|f(w)|\leqslant|w|+p_\mathcal{M}(|w|)$. Apply
Corollary~\ref{cor:size_recovery} (this corollary will appear only on the next section. But its proof does not
depend on this result and therefore this does not pose a problem)
to get that tape length 
$\myop{tlength}_\mathcal{M}(g(\psi_k(w),|w|),|w|+p_\mathcal{M}(|w|))=|f(w)|$
since $f(w)$ does not contain any blank character (this is true because $\gamma(\Gamma)\subset\Sigma\setminus\{b\})$.
This proves that $f$ is emulable because $g\in\mygpc$ and $\myop{tlength}_\mathcal{M}\in\mygpc$.

\subsection{On the robustness of previous characterization}

An interesting question arises when looking at this theorem: does the choice
of $k$ in Definition~\ref{def:fp_gpac_embed} matters, especially for the equivalence
with $\FP$ ? Fortunately not, as long as $k$ is large enough, as shown in the next lemma.

Actually in several cases, we will need to either decode words from
noisy encodings, or re-encode a word in a different basis. This is not
a trivial operation because small changes in the input can result in
big changes in the output. Furthermore, continuity forbids us from
being able to decode all inputs. The following theorem is a very
general tool. Its proof is detailed page \pageref{proof_th:decoding}.
 The following
Corollary~\ref{cor:reencoding} is a simpler version when one only
needs to re-encode a word.

\begin{theorem}[Word decoding]\label{th:decoding}
Let $k_1,k_2\in\N^*$ and $\kappa:\intinterv{0}{k_1-1}\rightarrow\intinterv{0}{k_2-1}$.
There exists a function $\left(\myop{decode}_\kappa:\subseteq\R\times\N\times\R\rightarrow\R\right)\in\mygpc$
such that for any word $w\in\intinterv{0}{k_1-1}^*$ and $\mu,\varepsilon\geqslant0$:
\[\text{if }\varepsilon\leqslant k_1^{-|w|}(1-e^{-\mu})
\text{ then }\myop{decode}_\kappa\left(\sum_{i=1}^{|w|}w_ik_1^{-i}+\varepsilon,|w|,\mu\right)=
\left(\sum_{i=1}^{|w|}\kappa(w_i)k_2^{-i},\#\{i|w_i\neq0\}\right)\]
\end{theorem}


\begin{corollary}[Re-encoding]\label{cor:reencoding} 
Let $k_1,k_2\in\N^*$ and $\kappa:\intinterv{1}{k_1-2}\rightarrow\intinterv{0}{k_2-1}$.
There exists a function $\left(\myop{reenc}_\kappa:\subseteq\R\times\N\rightarrow\R\times\N\right)\in\mygpc$
such that for any word $w\in\intinterv{1}{k_1-2}^*$
and $n\geqslant|w|$ we have:
\[\myop{reenc}_\kappa\left(\sum_{i=1}^{|w|}w_ik_1^{-i},n\right)=\left(\sum_{i=1}^{|w|}\kappa(w_i)k_2^{-i},|w|\right)\]
\end{corollary}

\begin{proof}
The proof is immediate: extend $\kappa$ with $\kappa(0)=0$ and define
\[\myop{reenc}_\kappa(x,n)=\myop{decode}_\kappa(x,n,0).\]
Since $n\geqslant|w|$, we can apply Theorem~\ref{th:decoding} with $\varepsilon=0$ to get the result.
Note that strictly speaking, we are not applying the theorem to $w$ but rather to $w$ padded with as
many $0$ symbols as necessary, ie $w0^{n-|w|}$. Since $w$ does not contain the symbol $0$,
its length is the same as the number of non-blank symbols it contains.
\end{proof}

\begin{remark}[Nonreversible re-encoding]
Note that the previous theorem and corollary allows from nonreversible re-encoding when $\kappa(\alpha)=0$
or $\kappa(\alpha)=k_2-1$ for some $\alpha\neq0$. For example, it allows one to re-encode a word over $\{0,1,2\}$ with $k_1=4$
to a word over $\{0,1\}$ with $k_2=2$ with $\kappa(1)=0$ and $\kappa(2)=1$
but the resulting number cannot be decoded in general (for continuity reasons).
In some cases, only the more general Theorem~\ref{th:decoding} provides a way to recover the encoding.
\end{remark}

A typically application of this function is to recover the length of the tape after
a computation.
Indeed way to do this is to keep track of the tape length during
the computation, but this usually requires a modified machine and some delimiters on the tape.
Instead, we will use the previous theorem to recover the length from the encoding,
assuming it does not contain any blank character. The only limitation is that to
recover the lenth of $w$ from its encoding $0.w$, we need to have an upper bound
on the length of $w$.

\begin{corollary}[Length recovery]\label{cor:size_recovery}
For any machine $\mathcal{M}$, there exists a function
$(\myop{tlength}_{\mathcal{M}}:\realenc{\machcfg{\mathcal{M}}}\times\N\rightarrow\N)\in\mygpc$ such
that for any word $w\in\left(\Sigma\setminus\{b\}\right)^*$ and any
$n\geqslant|w|$, $\myop{tlength}_{\mathcal{M}}(0.w, n)=|w|$.
\end{corollary}

\begin{proof}
It is an immediate consequence of Corollary~\ref{cor:reencoding} with $k_1=k_2=k$ and $\kappa=\idfun$
where we throw away the re-encoding.
\end{proof}

The previous tools are also precisely what is needed to prove that our notion of
emulation is independant of $k$.

\begin{lemma}[Emulation re-encoding]\label{lem:fp_gpac_embed_reenc}
Assume that $g\in\mygpc$ emulates $f$ with $k\in\N$. Then for
any $k'\geqslant k$, there exists $h\in\mygpc$ that emulates $f$ with $k'$.
\end{lemma}

\begin{proof}
  The proof follows from Corollary~\ref{cor:reencoding} by a
  standard game playing with encoding/reencoding.

More precisely, let $k'\geqslant k$
and define $\kappa:\intinterv{1}{k'}\rightarrow\intinterv{1}{k}$ and $\kappa^{-1}:\intinterv{1}{k}\rightarrow\intinterv{1}{k'}$
as follows:
\[\kappa(w)=\begin{cases}w&\text{if }w\in\gamma(\Gamma)\\1&\text{otherwise}\end{cases}\qquad\kappa^{-1}(w)=w.\]
In the following,
$0.w$ (resp. $0'.w$) denotes the rational encoding in basis $k$ (resp. $k'$).
Apply Corollary~\ref{cor:reencoding} twice to get that $\myop{reenc}_\kappa,\myop{reenc}_{\kappa^{-1}}\in\mygpc$.
Define:
\[h=\myop{reenc}_{\kappa^{-1}}\circ g\circ\myop{reenc}_{\kappa}.\]
Note that $\gamma(\Gamma)\subseteq\intinterv{1}{k-1}^*\subseteq\intinterv{1}{k'-1}^*$ since $\gamma$ never
maps letters to $0$ and $k\geqslant1+\max(\gamma(\Gamma))$ by
definition. 
{
Consequently for $w\in\Gamma^*$:
\begin{align*}
h(\psi_{k'}(w))&=h(0'.\gamma(w),|w|)\tag*{By definition of $\psi_{k'}$}\\
    &=\myop{reenc}_{\kappa^{-1}}(g(\myop{reenc}_\kappa(0'.\gamma(w),|w|)))\\
    &=\myop{reenc}_{\kappa^{-1}}(g(0.\kappa(\gamma(w)),|w|))\tag*{Because $\gamma(w)\in\intinterv{1}{k'}^*$}\\
    &=\myop{reenc}_{\kappa^{-1}}(g(0.\gamma(w),|w|))\tag*{Because $\gamma(w)\in\gamma(\Gamma)^*$}\\
    &=\myop{reenc}_{\kappa^{-1}}(g(\psi_k(w)))\tag*{By definition of $\psi_k$}\\
    &=\myop{reenc}_{\kappa^{-1}}(\psi_k(f(w)))\tag*{Because $g$ emulates $f$}\\
    &=\myop{reenc}_{\kappa^{-1}}(0.\gamma(f(w)),|f(w)|)\tag*{By definition of $\psi_k$}\\
    &=(0'.\kappa^{-1}(\gamma(f(w))),|f(w)|)\tag*{Because $\gamma(f(w))\in\gamma(\Gamma)^*$}\\
    &=(0'.\gamma(f(w)),|f(w)|)\tag*{By definition of $\kappa^{-1}$}\\
    &=\psi_{k'}(f(w))\tag*{By definition of $\psi_{k'}$}.
\end{align*}
}
\end{proof}

The previous notion of emulation was for single input functions, which is sufficient in theory
because we can always encode tuples of words using a single word or give Turing machines
several input/output tapes. But for the next results of this section,
it will be useful to have functions with multiple inputs/outputs without going through an encoding.
We extend the notion of discrete encoding in the natural way to handle this case.

\begin{definition}[
  emulation]\label{def:fp_gpac_embed_multi}
$f:\left(\Gamma^*\right)^n\rightarrow\left(\Gamma^*\right)^m$ is
called  \emph{
  emulable}
if there exists $g\in\mygpc$ and $k\in\N$ such that for any word $\vec{w}\in\left(\Gamma^*\right)^n$:
\[g(\psi_k(\vec{w}))=\psi_k(f(\vec{w}))\qquad\text{where}\quad
\psi_k(x_1,\ldots,x_\ell)=\left(\psi(x_1),\ldots,\psi(x_\ell)\right)\]
and $\psi_k$ is defined as in Definition~\ref{def:fp_gpac_embed}.
\end{definition}

It is trivial that Definition~\ref{def:fp_gpac_embed_multi} matches Definition~\ref{def:fp_gpac_embed}
in the case of unidimensional functions, thus the two definitions are consistent
with each other.

Theorem \ref{th:fp_gpac} then generalizes to the multidimensional
case naturally as follows. Proof is page \pageref{proof:th:fp_gpac_multi}.

\begin{theorem}[Multidimensional $\FP$ equivalence]\label{th:fp_gpac_multi}
For any 
$f:\left(\Gamma^*\right)^n\rightarrow\left(\Gamma^*\right)^m$,
$f\in\FP$ if and only if $f$ is 
emulable.
\end{theorem}

\section{A Characterization of $\PTIME$}\label{sec:ptime}

We will now use this characterization of $\FP$ to give a
characterization of $\PTIME$:  Our purpose is now to prove that a decision problem (language)
$\mathcal{L}$ belongs to the class $\PTIME$ if and only if it is
poly-length-analog-recognizable.

The following definition is a generalization (to general field $\K$) of Definition
\ref{def:discrete_rec_q}:

\begin{definition}[Discrete recognizability]\label{def:discrete_rec}
A language $\mathcal{L}\subseteq\Gamma^*$  is called $\K$-\emph{poly-length-analog-recognizable} if there
exists a vector $q$ of bivariate polynomials and a vector $p$ of polynomials with $d$ variables,
both with coefficients in $\K$, and a polynomial $\myOmega: \Rp \to \Rp$,
such that for all $w\in\Gamma^*$,
there is a (unique) $y:\Rp\rightarrow\R^d$ such that for all $t\in\Rp$:
\begin{itemize}
\item $y(0)=q(\psi_k(w))$ and $y'(t)=p(y(t))$
\hfill$\blacktriangleright$ $y$ satisfies a differential equation
\item if $|y_1(t)|\geqslant1$ then 
$|y_1(u)|\geqslant1$ for all $u\geqslant t$
\hfill$\blacktriangleright$ decision is stable
\item if $w\in\mathcal{L}$ (resp. $\notin\mathcal{L}$) and $\glen{y}(0,t)\geqslant\myOmega(|w|)$
then $y_1(t)\geqslant1$ (resp. $\leqslant-1$)
\hfill$\blacktriangleright$ decision
\item $\glen{y}(0,t)\geqslant t$
\hfill$\blacktriangleright$ technical condition\footnote{This could be
  replaced by only assuming that we have somewhere the additional
  ordinary differential equation $y'_0=1$.}
\end{itemize}
\end{definition}

\begin{theorem}[$\PTIME$ equivalence]\label{th:p_gpac}
Let $\K$ be a generable field such that
$\Rgen\subseteq\K\subseteq\Rpoly$.
For any 
language $\mathcal{L}\subseteq\Gamma^*$,
$\mathcal{L}\in\PTIME$ if and only if $\mathcal{L}$ is $\K$-poly-length-analog-recognizable.
\end{theorem}

\begin{proof}
The direct direction will build on the equivalence with $\FP$,
except that a technical point is to make sure that the decision of
the system is irreversible. 

Let $\mathcal{L}\in\PTIME$. Then there exist $f\in\FP$ and two distinct symbols
$\bar{0},\bar{1}\in\Gamma$ such that for any $w\in\Gamma^*$,
$f(w)=\bar{1}$ if $w\in\mathcal{M}$ and $f(w)=\bar{0}$ otherwise.
Let $\myop{dec}$ be defined by $\myop{dec}(k^{-1}\gamma(\bar{0}))=-2$
and $\myop{dec}(k^{-1}\gamma(\bar{1}))=2$. Recall that $\lagrange{\myop{dec}}\in\mygpc$ by Lemma~\ref{lem:lagrange_interp}.
Apply Theorem~\ref{th:fp_gpac} to get $g$ and $k$ that emulate $f$. Note in particular that for any $w\in\Gamma^*$,
$f(w)\in\{\bar{0},\bar{1}\}$ so $\psi(f(w))=(\gamma(\bar{0})k^{-1},1)$ or $(\gamma(\bar{1})k^{-1},1)$.
Define $g^*(x)=\lagrange{\myop{dec}}(g_1(x))$ and check that $g^*\in\mygpc$.
Furthermore, $g^*(\psi_k(w))=2$ if $w\in\mathcal{L}$ and $g^*(\psi_k(w))=-2$ otherwise,
by definition of the emulation and the interpolation.

We have $g^*\in\gc{\Upsilon}{\myOmega}$ for some polynomials $\myOmega$
and $\Upsilon$ be polynomials with corresponding $d,p,q$. 
Assume, without loss of generality, that $\myOmega$ and $\Upsilon$ are increasing functions.
Let $w\in\Gamma^*$ and consider the following system:
\[
\left\{\begin{array}{@{}r@{}l}
y(0)&=q(\psi_k(w))\\
v(0)&=\psi_k(w)\\
z(0)&=0\\
\tau(0)&=0\\
\end{array}\right.
\qquad
\left\{\begin{array}{@{}r@{}l}
y'(t)&=p(y(t))\\
v'(t)&=0\\
z'(t)&=\lxh_{[0,1]}(\tau(t)-\tau^*,1,y_1(t)-z(t))\\
\tau'(t)&=1
\end{array}\right.
\]
\[\tau^*=\myOmega(v_2(t),\ln2)
\]
In this system, 
$v$ is a constant variable used to store
the input and in particular the input length ($v_2(t)=|w|$), $\tau(t)=t$
is used to keep the time and $z$ is the decision variable.
Let $t\in[0,\tau^*]$, then by Lemma~\ref{lem:lxh_hxl}, $\infnorm{z'(t)}\leqslant e^{-1-t}$
thus $\infnorm{z(t)}\leqslant e^{-1}<1$. In other words, at time $\tau^*$ the system
has still not decided if $w\in\mathcal{L}$ or not. Let $t\geqslant\tau^*$,
then by definition of $\myOmega$ and since $v_2(t)=\psi_{k,2}(w)=|w|=\infnorm{\psi_k(w)}$,
$\infnorm{y_1(t)-g^*(\psi_k(w))}\leqslant e^{-\ln2}$.
Recall that $g^*(\psi_k(w))\in\{-2,2\}$ and let $\varepsilon\in\{-1,1\}$ such that $g^*(\psi_k(w))=\varepsilon2$.
Then $\infnorm{y_1(t)-\varepsilon2}\leqslant\frac{1}{2}$ which means that $y_1(t)=\varepsilon \lambda(t)$
where $\lambda(t)\geqslant\frac{3}{2}$.
Apply Lemma~\ref{lem:lxh_hxl} to conclude that $z$ satisfies for $t\geqslant\tau^*$:
\[z(\tau^*)\in[-e^{-1},e^{-1}]\qquad z'(t)=\phi(t)(\varepsilon\lambda(t)-z(t))\]
where $\phi(t)\geqslant0$ and $\phi(t)\geqslant1-e^{-1}$ for $t\geqslant\tau^*+1$.
Let $z_\varepsilon(t)=\varepsilon z(t)$ and check that $z_\varepsilon$ satisfies:
\[z_\varepsilon(\tau^*)\in[-e^{-1},e^{-1}]\qquad z_\varepsilon'(t)\geqslant\phi(t)(\tfrac{3}{2}-z_\varepsilon(t))\]
It follows that $z_\varepsilon$ is an increasing function and from a classical argument about differential inequalities that:
\[
z_\varepsilon(t)\geqslant\frac{3}{2}-\left(\frac{3}{2}-z_\varepsilon(\tau^*)\right)e^{-\int_{\tau^*}^t\phi(u)du}\]
In particular for $t^*=\tau^*+1+2\ln 4$
we have:
\[z_\varepsilon(t)\geqslant\frac{3}{2}-(\tfrac{3}{2}-z_\varepsilon(\tau^*))e^{-2\ln4(1-e^{-1})}
\geqslant\frac{3}{2}-2e^{-\ln4}\geqslant1.\]
This proves that $|z(t)|=z_\varepsilon(t)$ is an increasing function, so in particular once it has
reached $1$, it stays greater than $1$. Furthermore, if $w\in\mathcal{L}$ then $z(t^*)\geqslant1$
and if $w\notin\mathcal{L}$ then $z(t^*)\leqslant1$. Note that $\infnorm{(y,v,z,w)'(t)}\geqslant1$
for all $t\geqslant1$ so the technical condition is satisfied. Also note that $z$ is bounded by a constant, by a very similar reasoning.
This shows that if $Y=(y,v,z,\tau)$, then $\infnorm{Y(t)}\leqslant\poly(\infnorm{\psi_k(w)},t)$
because $\infnorm{y(t)}\leqslant\Upsilon(\infnorm{\psi_k(w)},t)$. Consequently, there is
a polynomial $\Upsilon^*$ such that $\infnorm{Y'(t)}\leqslant\Upsilon^*$ (this is immediate
from the expression of the system), and without loss of generality, we can assume that
$\Upsilon^*$ is an increasing function. And since $\infnorm{Y'(t)}\geqslant1$, we have that
$t\leqslant\glen{Y}(0,t)\leqslant t\sup_{u\in[0,t]}\infnorm{Y'(u)}\leqslant
t\Upsilon^*(\infnorm{\psi_k(w)},t)$. Define $\myOmega^*(\alpha)=t^*\Upsilon^*(\alpha,t^*)$
which is a polynomial because $t^*$ is polynomially bounded in $\infnorm{\psi_k(w)}=|w|$.
Let $t$ such that $\glen{Y}(0,t)\geqslant\myOmega^*(|w|)$, then by the above reasoning,
$t\Upsilon^*(|w|,t)\geqslant\myOmega^*(|w|)$ and thus $t\geqslant t^*$ so $|z(t)|\geqslant1$,
i.e. the system has decided.

The reverse direction of the proof is the following: assume that
$\mathcal{L}$ is $\K$- poly-length-analog-recognizable. Apply Definition~\ref{def:discrete_rec} to get $d,q,p$ and $\myOmega$.
Let $w\in\Gamma^*$ and consider the following system:
\[y(0)=q(\psi_k(w))\qquad y'(t)=p(y(t))\]
We will show that we can decide in time polynomial in $|w|$ whether $w\in\mathcal{L}$ or not.
Note that $q$ is a polynomial with coefficients in $\Rpoly$ (since we
consider $\K
\subset \Rpoly$) 
and $\psi_k(w)$ is a rational number so $q(\psi_k(w))\in\Rpoly^d$. Similarly, $p$
has coefficients in $\Rpoly$. Finally, note that\footnote{See Section~\ref{sec:solving_pivp} for the expression $\LenI$.}:
\begin{align*}
\LenI_{y,p}(0,t)&=\int_0^t\sigmap{p}\max(1,\infnorm{y(u)})^kdu\\
    &\leqslant t\sigmap{p}\max\left(1,\sup_{u\in[0,t]}\infnorm{y(u)}^k\right)\\
    &\leqslant t\sigmap{p}\max\left(1,\sup_{u\in[0,t]}\left(\infnorm{y(0)}+\glen{y}(0,t)\right)^k\right)\\
    &\leqslant t\poly(\glen{y}(0,t))\\
    &\leqslant \poly(\glen{y}(0,t))
\end{align*}
where the last inequality holds because $\glen{y}(0,t)\geqslant t$ thanks to the technical condition.
We can now apply Theorem~\ref{th:pivp_comp_analysis} to conclude that we are able to compute $y(t)\pm e^{-\mu}$
in time polynomial in $t,\mu$ and $\glen{y}(0,t)$. 

At this point, there is a slight subtlety:
intuitively we would like to evaluate $y$ at time $\myOmega(|w|)$ but it could be that
the length of the curve is exponential at this time.

Fortunately,
the algorithm that solves the PIVP works by making small time steps, and at each step
the length cannot increase by more than a constant\footnote{For the unconvinced reader,
it is still possible to write this argument formally by running the algorithm for
increasing values of $t$, starting from a very small value and making sure that at each
step the increase in the length of the curve is at most constant. This is very similar
to how Theorem~\ref{th:pivp_comp_analysis} is proved.}.
This means that we can stop the algorithm
as soon as the length is greater than $\myOmega(|w|)$. Let $t^*$ be the time at which
the algorithm stops. Then the running time of the algorithm will be polynomial in
$t^*,\mu$ and $\glen{y}(0,t^*)\leqslant\myOmega(|w|)+\bigO{1}$. Finally, thanks to the
technical condition, $t^*\leqslant\glen{y}(0,t^*)$ so this algorithm has
running time polynomial in $|w|$ and $\mu$. Take $\mu=\ln2$ then we get $\tilde{y}$
such that $\infnorm{y(t^*)-\tilde{y}}\leqslant\frac{1}{2}$. By definition of $\myOmega$,
$y_1(t)\geqslant1$ or $y_1(t)\leqslant-1$ so we can decide from $\tilde{y}_1$ if
$w\in\mathcal{L}$ or not.
\end{proof}

\section{A Characterization of Computable Analysis}\label{sec:ComputableAnalysis}

\subsection{Computable Analysis}

There exist many equivalent definitions of polynomial-time computability in the
framework of Computable Analysis. In this paper, we will use a particular
characterization by \cite{Ko91} in terms of computable rational approximation
and modulus of continuity. In the next theorem (which can be found e.g.~in \cite{Wei00}), $\D$ denotes the set of dyadic
rationals:
\[\D=\{m2^{-n},m\in\Z,n\in\N\}.\]

\begin{theorem}[Alternative definition of computable
  functions] \label{def:alt_comp_analysis}
A real function $f:[a,b]\rightarrow\R$ is computable (resp. polynomial time computable) if and
only if there exists a computable (resp. polynomial time computable\footnote{The second argument of $g$ must be in unary.}) function
$\psi:(\D\cap[a,b])\times\N\rightarrow\D$ and a computable (resp. polynomial) function $m:\N\rightarrow\N$ such that:
\begin{itemize}
\item $m$ is a modulus of continuity for $f$
\item for any $n\in\N$ and $d\in[a,b]\cap\D$, $|\psi(d,n)-f(d)|\leqslant2^{-n}$
\end{itemize}
\end{theorem}

This characterization is very useful for us because it does not involved the notion
of oracle, that would be difficult to formalize with differential equation. However,
in one direction of the proofs, it will be useful to have the following unusual variation of
the previous theorem:

\begin{theorem}[Unusual characterization of computable functions] \label{th:alt_comp_analysis_ex}
A real function $f:[a,b]\rightarrow\R$ is polynomial time computable if and
only if there exists a polynomial $q:\N\to\N$, a polynomial time computable\footnote{The second argument of $g$ must be in unary.} function
$\psi:X_q\rightarrow\D$ such that
\[\text{for all }x\in[a,b]\text{ and }(r,n)\in X_q(x), |\psi(r,n)-f(x)|\leqslant2^{-n}\]
where
\begin{align*}
X_q&=\bigcup_{x\in[a,b]}X_q(x),\\
X_q(x)&=\{(r,n)\in\D\times\N:|r-x|\leqslant 2^{-q(n)}\}.
\end{align*}
\end{theorem}

\begin{proof}
To show this characterization, we will directly use the original definition of
computability by \cite{Ko91} using oracles.

Assume $f$ is polynomial time computable is Ko's sense. Then there exists a
polynomial time Turing machine $\mathcal{M}$ such that for any $x\in[a,b]$
and any oracle\footnote{$\mathcal{O}$ is an oracle for $x$ if for any $n\in\N$,
$|\mathcal{O}(0^n)-x|\leqslant 2^{-n}$.}
$\mathcal{O}$ of $x$, $\mathcal{M}^\mathcal{O}$ is an oracle for $f(x)$. Since $\mathcal{M}$
runs in polynomial time, there exists a polynomial $q$ such that $\mathcal{M}^\mathcal{O}(0^n)$
finishes in less than $q(n)$ steps for all oracles $\mathcal{O}$. In particular,
all the calls to the oracles are of the form $0^{m}$ where $m\leqslant q(n)$.
Now define for any $(r,n)\in X_q$,
\[\psi(r,n)=\mathcal{M}^{\mathcal{N}_r}(0^n)
\quad\text{where }\mathcal{N}_r(0^n)=r.\]
In other words, on input $(r,n)$, $\psi$ runs $\mathcal{M}$ on input $0^n$ with
an oracle that returns $r$ unconditionally. Clearly $\psi$ runs in polynomial
time because $\mathcal{M}$ runs in time $q(n)$ and each call to the oracle takes
a constant time (the time to write down $r$). Now let $x\in[a,b]$ and $(r,n)\in X_q(x)$.
Then $|r-x|\leqslant 2^{-q(n)}$ thus $\mathcal{N}_r$ is an oracle for $x$ for all
calls of the form $0^m$ with $m\leqslant q(n)$. But since $\psi(r,n)$ will run
$\mathcal{M}$ on input $0^n$, it will only make calls with $m\leqslant q(n)$
(we chose $q$ so that it is the case). Thus $|\psi(r,n)-f(x)|\leqslant 2^{-n}$.

In the other direction we can use Theorem~\ref{def:alt_comp_analysis}. Let $r\in[a,b]\cap\Q$
and $n\in\N$. Since $r$ is rational, $(r,n)\in X_q(r)$ (because $|r-r|=0\leqslant 2^{-q(n)}$)
and thus $(r,n)\in X_q$ so we can apply $\psi$ on it and then $|\psi(r,n)-f(r)|\leqslant 2^{-n}$.
Furthermore, $m(n)=1+q(n+1)$ is a modulous of continuity for $f$. Indeed let $x,y\in[a,b]$
such that $|x-y|\leqslant 2^{-m(n)}$. Then let
\[r=2^{-m(n)}\left\lfloor\tfrac{x+y}{2}2^{m(n)}\right\rfloor.\]
Clearly $r\in\Q$ and
\begin{align*}
\tfrac{x+y}{2}2^{m(n)}-1
    &\leqslant \left\lfloor\tfrac{x+y}{2}2^{m(n)}\right\rfloor
    \leqslant \tfrac{x+y}{2}2^{m(n)}\\
\tfrac{x+y}{2}-2^{-m(n)}
    &\leqslant r
    \leqslant \tfrac{x+y}{2}\\
\tfrac{x-y}{2}-2^{-m(n)}
    &\leqslant r-y
    \leqslant \tfrac{x-y}{2}\\
|r-y|&\leqslant\tfrac{|x-y|}{2}+2^{-m(n)}\\
|r-y|&\leqslant\tfrac{2^{-m(n)}}{2}+2^{-m(n)}\\
|r-y|&\leqslant\tfrac{3}{2}2^{-1-q(n+1)}\\
|r-y|&\leqslant 2^{-q(n+1)}
\end{align*}
and similarly, $|r-x|\leqslant 2^{-q(n+1)}$. It follows that $(r,n+1)\in X_q(x)$
and $(r,n+1)\in X_q(y)$. So in particular $|\psi(r,n+1)-f(x)|\leqslant 2^{-n-1}$
and $|\psi(r,n+1)-f(y)|\leqslant 2^{-n-1}$ and thus
\[|f(x)-f(y)|\leqslant|f(x)-\psi(r,n+1)|+|\psi(r,n+1)-f(y)|\leqslant 2^{-n-1}+2^{-n-1}\leqslant 2^{-n}.\]
\end{proof}

\subsection{Mixing functions}\label{sec:mixing}

%
%
%

Suppose that we have two continuous functions $f_0$ and $f_1$ that partially cover $\R$ but
such that $\dom{f_0}\cup\dom{f_1}=\R$. We would like to build a new continuous function defined
over $\R$ out of them. One way of doing this is to build a function $f$ that equals $f_0$
over $\dom{f_0}\setminus\dom{f_1}$, $f_1$ over $\dom{f_1}\setminus\dom{f_0}$ and a linear
combination of both in between. For example consider $f_0(x)=x^2$
defined over $]-\infty,1]$ and $f_1(x)=x$ over $[0,\infty[$.
This approach may work from a mathematical point of view,
but it raises severe computational issues: how do we describe the two domains ?
How do we compute a linear interpolation between arbitrary sets ? What is the complexity
of this operation ? This would require to discuss the complexity of
real sets, which is a whole subject by itself.

A more elementary solution to this problem is what we call \emph{mixing}. We assume that
we are given an indicator function $i$ that covers the domain of both functions.
Such an example would be $i(x)=x$ in the previous example. The intuition is that $i$ describes both
the domains and the interpolation. Precisely, the resulting function should be $f_0(x)$ if $i(x)\leqslant0$, $f_1(x)$ if $i(x)\geqslant1$
and a \emph{mix} of $f_0(x)$ and $f_1(x)$ inbetween. The consequence of this choice is that
the domain of $f_0$ and $f_1$ must overlap on the region $\{x: 0<i(x)<1\}$.
In the previous example, we need to define $f_0$ over $]-\infty,1[=\{x: i(x)<1\}$ and $f_1$
over $]0,\infty]=\{x: i(x)>0\}$. Several types of mixing are possible,
the simplest being linear interpolation: $(1-i(x))f_0(x)+i(x)f_1(x)$. Formally, we would build the following
continuous function:



\begin{definition}[Mixing function]
Let $f_0:\subseteq\R^n\rightarrow\R^d$, $f_1:\subseteq\R^n\rightarrow\R^d$ and $i:\subseteq\R^n\rightarrow\R$.
Assume that $\{x: i(x)<1\}\subseteq\dom{f_0}$ and
$\{x: i(x)>0\}\subseteq\dom{f_1}$, and define for $x\in\dom{i}$:
\[\mix{i}{f_0}{f_1}(x)=\begin{cases}
f_0(x)&\text{if }i(x)\leqslant0\\
(1-i(x))f_0(x)+i(x)f_1(x)&\text{if }0<i(x)<1\\
f_1(x)&\text{if }i(x)\geqslant1\end{cases}.\]
\end{definition}

From closure properties, we get immediately:

\begin{theorem}[Closure by mixing]\label{th:comp:mix}
Let $f_0:\subseteq\R^n\rightarrow\R^d$, $f_1:\subseteq\R^n\rightarrow\R^d$ and $i:\subseteq\R^n\rightarrow\R$.
Assume that $f_0,f_1,i\in\mygpc{}$, that $\{x: i(x)<1\}\subseteq\dom{f_0}$ and that
$\{x: i(x)>0\}\subseteq\dom{f_1}$.  Then
$\mix{i}{f_0}{f_1}\in\mygpc{}$.
\end{theorem}

\begin{proof}
By taking $\min(\max(0,i(x)),1)$, which belongs to $\gplc$, we can assume that $i(x)\in[0,1]$.
Furthermore, it is not hard to see that
\[\mix{i}{f_0}{f_1}(x)=\mix{i}{0}{f_1}(x)+\mix{1-i}{0}{f_0}(x).\]
Thus we only need prove the result for the case where $f_0\equiv 0$, that is
\[g(x)=\begin{cases}0&\text{if }\alpha(x)=0\\\alpha(x)f(x)&\text{if }\alpha(x)>0\end{cases}.\]
Recall that by assumption, $f(x)$ is defined for $\alpha(x)>0$ but may not be defined
for $\alpha(x)=0$. The idea is use Item (4) of Proposition~\ref{th:main_eq}
(online-computability): let $\delta,d,p,y_0$ and $d',q,z_0$ that correspond to $f$ and $\alpha$ respectively.
Consider the following system for all $x\in\dom{\alpha}$:
\[y(0)=y_0,\qquad y'(t)=p(y(t),x),\]
\[z(0)=z_0,\qquad z'(t)=q(y(t),x),\]
\[w(t)=y(t)z(t).\]
There are two cases:
\begin{itemize}
\item If $\alpha(x)>0$ then $x\in\dom{f}$ thus $y(t)\to f(x)$ and $z(t)\to\alpha(x)$
as $t\to\infty$. It follows that $w(t)\to\alpha(x)f(x)=g(x)$ as $t\to\infty$. We
leave the convergence speed analysis to the reader since it's standard.
\item If $\alpha(x)=0$ then we have no guarantee on the convergence of $y$. However we know
that
\[\infnorm{y(t)}\leqslant\Upsilon(\infnorm{x},t)\]
where and $\Upsilon$ is a polynomial, and
\[|z(t)-\alpha(x)|\leqslant e^{-\mu}\qquad\text{for all }t\geqslant\myOmega(\infnorm{x},\mu).\]
Thus for all $\mu\in\Rp$,
\begin{align*}
\infnorm{w(\myOmega(\infnorm{x},\mu))}
    &=\infnorm{z(t)y(t)}\\
    &=\Upsilon(\infnorm{x},\myOmega(\infnorm{x},\mu))e^{-\mu}.
\end{align*}
But since $\Upsilon$ and $\myOmega$ are polynomials, the right-hand side converges exponentially
fast (in $\mu$) to $0$ whereas the time $\myOmega(\infnorm{x},\mu)$ only grows polynomially.
\end{itemize}
This shows that $g\in\gplc$.
\end{proof}

\subsection{Computing effective limits}

Intuitively, our notion of computation already contains the notion of effective limit.
More precisely, if $f$ is computable and is such that $f(x,t)\rightarrow g(x)$ when
$t\rightarrow\infty$ effectively then $g$ is computable. 
The result below extends this result to the case where the limit is
restricted to $t\in\N$. 

\begin{theorem}[Closure by effective limit]\label{th:gpac_comp_limit}
Let $I\subseteq\R^n$, $f:\subseteq I\times\N\rightarrow\R^m$, $g:I\rightarrow\R^m$ and
$\mho:\Rp^2\rightarrow\Rp$ be a nondecreasing polynomial.
Assume that $f\in\mygpc$ and that
\[\{(x,n)\in I\times\N:n\geqslant\mho(\infnorm{x},0)\}\subseteq\dom{f}.\]
Further assume that for all $(x,n)\in\dom{f}$ and $\mu\geqslant 0$,
\[\text{if }n\geqslant\mho(\infnorm{x},\mu)\text{ then }\infnorm{f(x,n)-g(x)}\leqslant e^{-\mu}.\]
Then $g\in \mygpc$.
\end{theorem}

\begin{proof}
First note that $\frac{1}{2}-e^{-2}\geqslant\frac{1}{3}$ and define
for $x\in I$ and $n\geqslant\mho(\infnorm{x},0)$:
\[\begin{array}{ll}
f_0(x,\tau)=f(x,\functionrnd(\tau,2))&\qquad \tau\in\left[n-\frac{1}{3},n+\frac{1}{3}\right],\\
f_1(x,\tau)=f(x,\functionrnd(\tau+\tfrac{1}{2},2))&\qquad \tau\in\left[n+\frac{1}{6},n+\frac{5}{6}\right].
\end{array}\]
By Definition~\ref{def:comp:round} and hypothesis on $f$, both are
well-defined because for all $n\geqslant\mho(\infnorm{x},0)$
and $\tau\in\left[n-\frac{1}{3},n+\frac{1}{3}\right]$,
\[(x,\crnd\left(\tau,2\right))=(x,n)\in\dom{f}\]
and similarly for $f_1$.
Also note that their domain of definition overlap on
$[n+\tfrac{1}{6},n+\tfrac{1}{3}]$ and
$[n+\tfrac{2}{3},n+\tfrac{5}{6}]$.  Apply Theorem~\ref{th:comp:round}
and Theorem~\ref{th:gpac_comp_composition} to get that $f_0,f_1\in\mygpc$.
We also need to build the indicator function: this is where the choice of above values will prove convenient.
Define for any $x\in I$ and $\tau\geqslant\mho(\infnorm{x},0)$:
\[i(x,\tau)=\tfrac{1}{2}-\cos(2 \pi \tau).\]
It is now easy to check that:
\begin{align*}
\{(x,\tau): i(x)<1\}&=I\times\bigcup_{n\geqslant \mho(\infnorm{x},0)}\left]n-\tfrac{1}{3},n+\tfrac{1}{3}\right[\subseteq\dom{f_0}.\\
\{(x,\tau): i(x)>0\}&=I\times\bigcup_{n\geqslant \mho(\infnorm{x},0)}\left]n+\tfrac{1}{6},n+\tfrac{5}{3}\right[\subseteq\dom{f_1}.\\
\end{align*}
Define for any $x\in I$ and $\mu\in\Rp$:
\[f^*(x,\mu)=\mix{i}{f_0}{f_1}(x,\mho(\norm_{\infty,1}(x),\mu)).\]
Recall that $\norm_{\infty,1}$, defined in Lemma~\ref{lem:norm}, belongs to \mygpc and satisfies
$\norm_{\infty,1}(x)\geqslant\infnorm{x}$.
We can thus apply Theorem~\ref{th:comp:mix} to get that $f^*\in\mygpc$. Note that $f^*$ is defined over $I\times\Rp$
since for all $x\in I$ and $\mu\geqslant 0$, $\mho(\norm_{\infty,1}(x),\mu)\geqslant\mho(\infnorm{x},0)$
since $\mho$ is nondecreasing.
We now claim that for any $x\in I$ and $\mu\in\Rp$, if
$\tau\geqslant1+\mho(\infnorm{x},\mu)$ then $\infnorm{f^*(x,\tau)-g(x)}\leqslant 2e^{-\mu}$.
There are three cases to consider:
\begin{itemize}
\item If $\tau\in[n-\tfrac{1}{6},n+\tfrac{1}{6}]$ for some $n\in\N$
  then $i(x)\leqslant0$ so $\mix{i}{f_0}{f_1}(x,\tau)=f_0(x,\tau)=f(x,n)$
and since $n\geqslant\tau-\frac{1}{6}$ then $n\geqslant\mho(\infnorm{x},\mu)$ thus $\infnorm{f^*(x,\tau)-g(x)}\leqslant e^{-\mu}$.
\item If $\tau\in[n+\tfrac{1}{3},n+\tfrac{2}{3}]$ for some $n\in\N$
  then $i(x)\geqslant1$ so $\mix{i}{f_0}{f_1}(x,\tau)=f_1(x,\tau)=f(x,n+1)$
and since $n\geqslant\tau-\frac{2}{3}$ then $n+1\geqslant\mho(\infnorm{x},\mu)$ thus $\infnorm{f^*(x,\tau)-g(x)}\leqslant e^{-\mu}$.
\item If
  $\tau\in[n+\tfrac{1}{6},n+\tfrac{1}{3}]\cup[n+\tfrac{2}{3},n+\tfrac{5}{6}]$
  for some $n\in\N$ then $\infnorm{f^*(x,\tau)-g(x)}\leqslant
  e^{-\mu}$ from Theorem \ref{th:comp:mix} since 
$i(x,\tau)\in[0,1]$ so
$f^*(x,\tau)=(1-i(x,\tau))f_0(x,\tau)+i(x,\tau)f_1(x,\tau)=(1-i(x,\tau))f(x,\round{\tau})+i(x,\tau)f(x,\round{\tau+\tfrac{1}{2}})$. Since
$\round{\tau},\round{\tau+\tfrac{1}{2}}\geqslant\mho(\infnorm{x},\mu)$ yields $\infnorm{f(x,\round{\tau})-g(x)}\leqslant e^{-\mu}$
and $\infnorm{f(x,\round{\tau+\tfrac{1}{2}})-g(x)}\leqslant e^{-\mu}$ thus $\infnorm{f^*(x,\tau)-g(x)}\leqslant2e^{-\mu}$
because $|i(x,\tau)|\leqslant1$.
\end{itemize}

It follows that $g$ is the effective limit of $f^*$ and thus
$g\in\mygpc$ 
(see Remark~\ref{rem:gpwc_gpc}).
\end{proof}
\begin{remark}[Optimality]\label{rem:gpac_comp_limit}
The condition that $\mho$ is a polynomial is essentially optimal. Intuitively,
if $f\in\mygpc$ and satisfies $\infnorm{f(x,\tau)-g(x)}\leqslant e^{-\mu}$
whenever $\tau\geqslant\mho(\infnorm{x},\mu)$ then $\mho$ is a modulus of continuity for $g$.
By Theorem~\ref{th:comp_implies_cont}, if $g\in\mygpc$ then it admits a polynomial modulus of continuity
so $\mho$ must be a polynomial. For a formal proof of this intuition, see
examples \ref{ex:gpac_limit_poly1} and \ref{ex:gpac_limit_poly2}.
\end{remark}

\begin{example}[$\mho$ must be polynomial in $x$]\label{ex:gpac_limit_poly1}
Let $f(x,\tau)=\min(e^x,\tau)$ and $g(x)=e^x$. Trivially $f(x,\cdot)$ converges to $g$
because $f(x,\tau)=g(x)$ for $\tau\geqslant e^x$. But $g\notin\mygpc$ because it is not
polynomially bounded. In this case $\mho(x,\mu)= e^x$ which is exponential and $f\in\mygpc$
by Proposition~\ref{prop:clamped_exp}.
\end{example}

\begin{example}[$\mho$ must be polynomial in $\mu$]\label{ex:gpac_limit_poly2}
Let $g(x)=\frac{-1}{\ln x}$ for $x\in\left[0,e\right]$ which is defined in $0$
by continuity. Observe that $g\notin\mygpc$, since its modulus of continuity is
exponential around $0$ because $g(\cramped{e^{-e^\mu}})=e^{-\mu}$ for all $\mu\geqslant0$.
However note that $g^*\in\mygpc$ where
$g^*(x)=g(\cramped{e^{-x}})=\frac{1}{x}$ for $x\in[1,+\infty[$.
Let $f(x,\tau)=g^*(\min(-\ln x,\tau))$ and check, using that $g$ is increasing
and non-negative, that:
$|f(x,\tau)-g(x)|=\left|g(\max(x,\cramped{e^{-\tau}}))-g(x)\right|
    \leqslant g(\max(x,\cramped{e^{-\tau}}))\leqslant \frac{1}{\tau}$.
Thus $\mho(\infnorm{x},\mu)=e^\mu$ which is exponential and $f\in\mygpc$ because
$(x,\tau)\mapsto\min(-\ln x,\tau)\in\mygpc$ by a proof similar to Proposition~\ref{prop:clamped_exp}.
\end{example}

\subsection{Cauchy completion and complexity}

We want to approach a function $f$ defined over some domain $\mathcal{D}$
by some function $g$, where $g$ is defined over
\[\left\{\left(\frac{p}{2^n},n\right),p\in\Z^d,n\in\N:\frac{p}{2^n}\in\mathcal{D}\right\},\]
the set of dyadic numbers in $\mathcal{D}$ (we need to include the precision $n$
as argument for complexity reasons).

The problem is that the shape of the domain $\mathcal{D}$ 
matters: if we want to
compute $f(x)$
, we will need to ``approach'' $x$ from within the domain, since above domain only allows
$k$-adic numbers in $\mathcal{D}$.
For example if $f$ is defined over $[a,b]$ then to compute $f(a)$ we need to approach $a$ \emph{by above}, but
for $f(b)$, we need to approach $b$ \emph{by below}. For more general domains, finding the
right direction of approach might be (computationally) hard, if even possible,
and depends on the shape of the domain.

To avoid this problem, we requires that $g$ be defined on a slightly larger domain
so that this problem disappears. This notion is motivated by Theorem~\ref{th:alt_comp_analysis_ex}.

\begin{theorem}\label{th:cont_redux}
Let $d,e,\ell\in\N$, $\mathcal{D}\subseteq\R^{d+e}$, $k\geqslant 2$ and $f:\mathcal{D}\to\R^\ell$.
Assume that there exists a polynomial $\mho:\Rp^2\to\Rp$ and $(g:\subseteq\D^d\times\N\times\R^e\to\R^\ell)\in\mygpc$
such that
for all $(x,y)\in\mathcal{D}$ and $n,m\in\N,p\in\Z^d$,
\[\text{if }\infnorm{\tfrac{p}{2^m}-x}\leqslant 2^{-m}\text{ and }m\geqslant\mho(\infnorm{(x,y)},n)
\text{ then\footnotemark}\infnorm{g(\tfrac{p}{2^m},m,y)-f(x,y)}\leqslant 2^{-n}.\]
Then $f\in\mygpc$.
\end{theorem}
\footnotetext{The domain of definition of $g$ is exactly those points $(\tfrac{p}{2^m},m,y)$ that
satisfy the previous ``if''.}

Section \ref{sec:th:cont_redux} is devoted to the proof of this theorem.
We now show that this is sufficient to characterize Computable Analysis
using continuous time systems.

\subsection{From Computable Analysis to $\mygpc$}

\begin{theorem}[From Computable Analysis to $\mygpc$]\label{th:eq_comp_analysis_one_sense}
For any $a,b\in\R$, any generable field $\K$ such that $\Rgen\subseteq\K\subseteq\Rpoly$,
if $f\in C^0([a,b],\R)$ is polynomial-time computable
then $f\in\mygpc$.
\end{theorem}

Note that $a$ and $b$
need not be computable so we must take care not to use them in any
computation!

\begin{proof}
Let $f\in C^0([a,b],\R)$ and assume that $f$ is polynomial-time computable. We will
first reduce the general situation to a simpler case. Let $m,M\in\Q$ such that $m<f(x)<M$ for all $x\in[a,b]$. Let $l,r\in\Q$ such
that $l\leqslant a<b\leqslant r$. Define
\[g(\alpha)=\frac{1}{4}+\frac{f(l+(r-l)(2\alpha-\tfrac{1}{2}))-m}{2(M-m)}\]
for all $\alpha\in[a',b']=\left[\tfrac{1}{4}+\tfrac{a-l}{2(r-l)},\tfrac{1}{4}+\tfrac{b-l}{2(r-l)}\right]\subseteq[\tfrac{1}{4},\tfrac{3}{4}]$.
It follows that $g\in C^0([a',b'],[\tfrac{1}{4},\tfrac{3}{4}])$ with
$[a',b']\subseteq[\tfrac{1}{4},\tfrac{3}{4}]$.
Furthermore, by construction, for every $x\in[a,b]$ we have that
\[f(x)=2(M-m)\left(g\left(\frac{1}{4}+\frac{x-l}{2(r-l)}\right)-\frac{1}{4}\right)+m.\]
Thus if $g\in\mygpc$ then $f\in\mygpc$ because 
of
closure properties of $\mygpc$.
Hence, in the remaining of the proof, we can thus assume that $f\in C^0([a,b],\tfrac{1}{4},\tfrac{3}{4})$ with $[a,b]\subseteq[\tfrac{1}{4},\tfrac{3}{4}]$.
This restriction is useful to simplify the encoding used later in the
proof.

Let $f\in C^0([a,b],\left[\tfrac{1}{4},\tfrac{3}{4}\right])$ with $[a,b]\subseteq\left[\tfrac{1}{4},\tfrac{3}{4}\right]$ be a polynomial time computable function.
Apply Theorem~\ref{th:alt_comp_analysis_ex}
to get $g$ and $\mho$ (we renamed $\psi$ to $g$ and $q$ to $\mho$ to avoid a name clash). Note that $g:X_\mho\to\D$
has its second argument written in unary. In order to apply the $\FP$ characterization,
we need to discuss the encoding of rational numbers and unary integers.
Let us choose a binary alphabet $\Gamma=\{0,1\}$ and its encoding function $\gamma(0)=1$ and $\gamma(1)=2$,
and define for any $w,w'\in\Gamma^*$:
\[\psi_\N(w)=|w|,\qquad\psi_\D(w)=\sum_{i=1}^{|w|}w_i2^{-i}.\]
Note that $\psi_\D$ is a surjection from $\Gamma^*$ to $\D\cap[0,1[$, the dyadic
part of $[0,1[$.
Define for any relevant\footnote{We will discuss the domain of definition below.} $w,w'\in\Gamma^*$:
\[g_\Gamma(w,w')=\psi_\D^{-1}(g(\psi_\D(w),\psi_\N(w'))\]
where $\psi_\D^{-1}(x)$ is the smallest $w$ such $\psi_\D(w)=x$ (it is unique).
For $g_\Gamma(w,w')$ to be defined, we need that
\begin{itemize}
\item $(\psi_\D(w),\psi_\N(w'))\in\dom{g}=X_\mho$: in the case of interest, this is true if
    \[\psi_\D(w)\in\left[a'-2^{-\mho(|a'|,|w'|)},b'+2^{-\mho(|b'|,|w'|)}\right],\]
\item $g(\psi_\D(w),\psi_\N(w'))\in\dom{\psi_\D^{-1}}=\D\cap[0,1[$: since
    $|g(\psi_\D(w),\psi_\N(w'))-f(\psi_\D(w))|\leqslant 2^{-\psi_\N(w')}$
    and $f(\psi_\D(w)\in[\tfrac{1}{4},\tfrac{3}{4}]$, then it is true when
    $\psi(w')=|w'|\geqslant 3$ because
    \[g(\psi_\D(w),\psi_\N(w'))\in f(\psi_\D(w)+[-2^{-3},2^{-3}]\subseteq[\tfrac{1}{4},\tfrac{3}{4}]+[-\tfrac{1}{8},\tfrac{1}{8}]\subset[0,1].\]
\end{itemize}

Since $\psi_\D$ is a polytime computable encoding, then $g_\Gamma\in\FP$ because it has
running time polynomial in the length of $\psi_\D(w)$ and the (unary) value of $\psi_\N(w')$,
which are the length of $w$ and $w'$ respectively, by definition of $\psi_\D$ and $\psi_\N$.
Apply Theorem~\ref{th:fp_gpac_multi} to get that $g_\Gamma$ is emulable. Thus there exist $h\in\mygpc$ and $k\in\N$ such that
for all $w,w'\in\dom{g_\Gamma}$:
\[h(\psi_k(w,w'))=\psi_k(g_\Gamma(w,w')).\]
where $\psi_k$ is defined as in Definition~\ref{def:fp_gpac_embed_multi}.
At this point, everything is encoded: the input and the output of $h$. Our next step
is to get rid of the encoding by building a function that works the dyadic part of $[a,b]$
and returns a real number.

Define $\kappa:\intinterv{0}{k-2}\rightarrow\{0,1\}$ by $\kappa(\gamma(0))=0$ and $\kappa(\gamma(1))=1$
and $\kappa(\alpha)=0$ otherwise. Define $\iota:\{0,1\}\rightarrow\intinterv{0}{k-2}$
by $\iota(0)=\gamma(0)$ and $\iota(1)=\gamma(1)$. For any relevant $q\in\D$ and $n,m\in\N$
define:
\[g^*(q,n,p)=\myop{reenc}_{\kappa,1}(h(\myop{reenc}_{\iota}(q,n),0,p)).\]
We will see that this definition makes sense for some values. Let $n\in\N$, $p\geqslant 3$ and
$m\in\Z$, write $q=m2^{-n}$ and assume that $m2^{-n}\in\left[a'-2^{-\mho(|a'|,p)},b'+2^{-\mho(|b'|,p)}\right]\subseteq[0,1[$.
Then there exists $w^q\in\{0,1\}^n$ such that $m2^{-n}=\sum_{i=1}^{n}w^q_i2^{-i}$. Consequently,
{
\begin{align}
\myop{reenc}_{\iota}(q,n)
    &=\myop{reenc}_{\iota}\left(\sum_{i=1}^{n}w^q_i2^{-i},n\right)&&\text{By Corollary~\ref{cor:reencoding}}\\
    &=\left(\sum_{i=1}^{n}\iota(w^q_i)k^{-i},n\right)&&\text{By definition of }\myop{reenc}_{\iota}\\
    &=\left(\sum_{i=1}^{n}\gamma(w^q_i)k^{-i},n\right)&&\text{Because }\iota=\gamma\\
    &=\psi_k(w^q).\label{eq:reenc_iota_q_n}
\end{align}
}
Furthermore, note that by definition of $w^q$:
\begin{equation}\label{eq:psi_D_w_q}
\psi_\D(w^q)=\sum_{i=1}^{|w^q|}w^q_i2^{-i}=q.
\end{equation}
Similarly, note that
\begin{equation}\label{eq:0_p_psi}
(0,p)=\left(\sum_{i=1}^{p}0k^{-i},p\right)=\psi_k(0^p)
\end{equation}
and
\begin{equation}\label{eq:psi_N_0_p}
\psi_\N(0^p)=|0^p|=p.
\end{equation}
Additionally, for any $w\in\Gamma^*$ we have that
\begin{align}
\myop{reenc}_{\kappa,1}(\psi_k(w))
    &=\myop{reenc}_{\kappa,1}\left(\sum_{i=1}^{|w|}\gamma(w_i)k^{-i},|w|\right)\tag*{By definition of $\psi_k$}\\
    &=\sum_{i=1}^{|w|}\kappa(\gamma(w_i))2^{-i}\tag*{By Corollary~\ref{cor:reencoding}}\\
    &=\sum_{i=1}^{|w|}w_i2^{-i}\tag*{Because $\kappa\circ\gamma=\idfun$}\\
    &=\psi_\D(w).\label{eq:reenc_kappa_psi_k}
\end{align}
Putting everything together, we get that
\begin{align}
g^*(q,n,p)
    &=\myop{reenc}_{\kappa,1}(h(\myop{reenc}_{\iota}(q,n),0,p))\tag*{}\\
    &=\myop{reenc}_{\kappa,1}(h(\psi_k(w^q,0^p)))\tag*{By \eqref{eq:reenc_iota_q_n} and \eqref{eq:0_p_psi}}\\
    &=\myop{reenc}_{\kappa,1}(\psi_k(g_\Gamma(w^q,0^p)))\tag*{By definition of $h$}\\
    &=\myop{reenc}_{\kappa,1}(\psi_k(\psi_\D^{-1}(g(\psi_\D(w^q),\psi_\N(0^p)))))\tag*{By definition of $g_\Gamma$}\\
    &=\myop{reenc}_{\kappa,1}(\psi_k(\psi_\D^{-1}(g(q,p))))\tag*{By \eqref{eq:psi_D_w_q} and \eqref{eq:psi_N_0_p}}\\
    &=\psi_\D(\psi_\D^{-1}(g(q,p)))\tag*{By \eqref{eq:reenc_kappa_psi_k}}\\
    &=g(q,p).\label{eq:g_star_g}
\end{align}
Finally, $g^*\in\mygpc$ because $\myop{reenc}_{\kappa},\myop{reenc}_{\iota}\in\mygpc$
by Corollary~\ref{cor:reencoding}.
Finally for any relevant $n\geqslant 3$ and $q\in\D$, let
\[\tilde{g}(q,n)=g^*(q,n,n).\]
Clearly $\tilde{g}\in\mygpc$. We will show that $\tilde{g}$ satisfies the assumption of Theorem~\ref{th:cont_redux}.
Let $x\in[a,b]$, $m,n\in\N$ and $p\in\Z$ such that
\[\left|x-\tfrac{p}{2^m}\right|\leqslant 2^{-m}\text{ and }m\geqslant\mho(|x|,n+2)+n+2.\]
Then\footnote{The proof is a bit involved because we naturally have $g(\tfrac{p}{2^m},m)$ with
$m\geqslant\mho(|x|,n)$ but we want $g(\tfrac{p}{2^m},n)$ to apply Theorem~\ref{th:cont_redux}.}
\PASCOMMENTESAUFSILATEXDIFF{
\begin{align*}
|\tilde{g}(\tfrac{p}{2^m},m)-f(x)|
    &=|g^*(\tfrac{p}{2^m},m,m)-f(x)|\\
    &=|g(\tfrac{p}{2^m},m)-f(x)|&&\text{by \eqref{eq:g_star_g}}\\
    &\leqslant|g(\tfrac{p}{2^m},m)-f(\tfrac{p}{2^m})|+|f(\tfrac{p}{2^m})-f(x)|\\
\intertext{But for any rational $q$, $|g(q,n)-f(q)|\leqslant 2^{-n}$ for all $n\in\N$,}
    &\leqslant 2^{-m}+|f(\tfrac{p}{2^m})-f(x)|\\
    &\leqslant 2^{-m}+|f(\tfrac{p}{2^m})-g(\tfrac{p}{2^m},n+2)|+|g(\tfrac{p}{2^m},n+2)-f(x)|\\
\intertext{But for any rational $q$, $|g(q,n)-f(q)|\leqslant 2^{-n}$ for all $n\in\N$,}
    &\leqslant 2^{-m}+2^{-n-2}+|g(\tfrac{p}{2^m},n+2)-f(x)|\\
\intertext{But $|x-\tfrac{p}{2^m}|\leqslant 2^{-m}\leqslant 2^{-\mho(|x|,n+2)}$
so we can apply Theorem~\ref{th:alt_comp_analysis_ex},}
    &\leqslant 2^{-m}+2^{-n-2}+2^{-n-2}\\
    &\leqslant 3\cdot2^{-n-2}&&\text{since }m\geqslant n+2\\
    &\leqslant 2^{-n}.
\end{align*}
}
Thus we can apply Theorem~\ref{th:cont_redux} to $\tilde{g}$ and get that $f\in\mygpc$.
\end{proof}





\subsection{Equivalence with Computable Analysis}

Note that the characterization works over $[a,b]$ where $a$ and $b$ can be arbitrary real
numbers.

\begin{theorem}[Equivalence with Computable Analysis]\label{th:eq_comp_analysis}
For any 
$f\in C^0([a,b],\R)$, $f$ is polynomial-time computable
if and only if $f\in\mygpc$.
\end{theorem}

\begin{proof}
The proof of the missing direction of the theorem is the following:
Let $f \in \mygpc$. Then $f\in\gc{\Upsilon}{\myOmega}$ where $\Upsilon,\myOmega$
are polynomials which we can assume to be increasing functions, and corresponding  $d,p$ and $q$.
Apply Theorem~\ref{th:comp_implies_cont} to $f$ to get $\mho$ and define
\[m(n)=\tfrac{1}{\ln2}\mho(\max(|a|,|b|),n\ln2).\]
It follows from the definition that $m$ is a modulus of continuity of $f$ since for any $n\in\N$ and $x,y\in[a,b]$
such that $|x-y|\leqslant2^{-m(n)}$ we have:
\[|x-y|\leqslant2^{-\frac{1}{\ln2}\mho(\max(|a|,|b|),n\ln2)}\\
    =e^{-\mho(\max(|a|,|b|),n\ln2)}\\
    \leqslant e^{-\mho(|x|,n\ln2)}.\]
Thus $|f(x)-f(y)|\leqslant e^{-n\ln2}=2^{-n}$. We will now see how to approximate $f$ in
polynomial time. Let $r\in\Q$ and $n\in\N$. We would like to compute $f(r)\pm2^{-n}$.
By definition of $f$, there exists a unique $y:\Rp\rightarrow\R^d$ such that for all $t\in\Rp$:
\[y(0)=q(r)\qquad y'(t)=p(y(t).\]
Furthermore, $|y_1(\myOmega(|r|,\mu))-f(r)|\leqslant e^{-\mu}$ for any $\mu\in\Rp$ and
$\infnorm{y(t)}\leqslant\Upsilon(|r|,t)$ for all $t\in\Rp$. Note that since the coefficients
of $p$ and $q$ belongs to $\Rpoly$,
it follows that we can apply Theorem~\ref{th:pivp_comp_analysis} to compute $y$. More concretely,
one can compute a rational $r'$ such that $|y(t)-r'|\leqslant2^{-n}$ in time bounded by
\[\poly(\degp{p},\LenI(0,t),\log\infnorm{y(0)},\log\sigmap{p},-\log2^{-n})^d.\]
Recall that in this case, all the parameters $d,\sigmap{p},\degp{p}$ only depend on $f$ and are thus fixed
and that $|r|$ is bounded by a constant. Thus these are all considered constants.
So in particular, we can compute $r'$ such that $|y(\myOmega(|r|,(n+1)\ln2)-r'|\leqslant2^{-n-1}$ in time:
\[\poly(\LenI(0,\myOmega(|r|,(n+1)\ln2)),\log\infnorm{q(r)},(n+1)\ln2).\]
Note that $|r|\leqslant\max(|a|,|b|)$ and since $a$ and $b$ are constants and $q$
is a polynomial, $\infnorm{q(r)}$ is bounded by a constant. Furthermore,
\begin{align*}
\LenI(0,\myOmega(|r|,(n+1)\ln2))&=\int_0^{\myOmega(|r|,(n+1)\ln2)}\max(1,\infnorm{y(t)})^{\degp{p}}dt\\
    &\leqslant \int_0^{\myOmega(|r|,(n+1)\ln2)}\poly(\Upsilon(\infnorm{r},t))dt\\
    &\leqslant \myOmega(|r|,(n+1)\ln2)\poly(\Upsilon(|r|,\myOmega(|r|,(n+1)\ln2)))dt\\
    &\leqslant\poly(|r|,n)\leqslant\poly(n).
\end{align*}
Thus $r'$ can be computed in time:
\[\poly(n).\]
Which is indeed polynomial time since $n$ is written in unary. Finally:
\begin{align*}
|f(r)-r'|&\leqslant|f(r)-y(\myOmega(|r|,(n+1)\ln2))|+|y(\myOmega(|r|,(n+1)\ln2))-r'|\\
    &\leqslant e{-(n+1)\ln2}+2^{-n-1}\\
    &\leqslant 2^{-n}.
\end{align*}
This shows that $f$ is polytime computable.
\end{proof}

\begin{remark}[Domain of definition]\label{rem:eq_comp_analysis_domain}
The equivalence holds over any interval $[a,b]$ but it can be extended in several ways.
First it is possible to state an equivalence over $\R$ .
Indeed, classical real computability defines the complexity of $f(x)$ over $\R$
as polynomial in $n$ and $p$ where $n$ is the precision and $k$ the length of input,
defined by $x\in[-2^k,2^k]$. Secondly, the equivalence also holds for multidimensional domains
of the form $I_1\times I_2\times\cdots\times I_n$ where $I_k=[a_k,b_k]$ or $I_k=\R$.
However, extending this equivalence to partial functions requires some caution. Indeed,
our definition does not specify the behavior of functions outside of the domain, whereas classical
discrete computability and some authors in Computable Analysis mandate that the machine
never terminates on such inputs. More work is needed in this direction to understand how to state the equivalence
in this case, in particular how to translate the ``never terminates'' part. Of course, the equivalence
holds for partial functions where the behavior outside of the domain is not defined.
\end{remark}

\section{Missing Proofs}\label{sec:proofs}\label{sec:avant:rationalelimination}

\subsection{Proof of Theorem \ref{th:gpac_comp_iter}: Simulating Discrete by Continuous Time}
\label{sec: gpac_comp_iter}

\subsubsection{A construction used elsewhere}

Another very common pattern that we will use is known as ``sample and hold''.
Typically, we have a variable signal and we would like to apply some process to it. Unfortunately,
the device that processes the signal assumes (almost) constant input and does not work in real time (analog-to-digital
converters would be a typical example).
In this case, we cannot feed the signal directly to the processor so we need some black
box that samples the signal to capture its value, and holds this value long enough
for the processor to compute its output. This process is usually used in a $\tau$-periodic
fashion: the box samples for time $\delta$ and holds for time $\tau-\delta$.
This is precisely what the $\sample$ function achieves. In fact, we show that it achieves
much more: it is robust to noise and has good convergence properties when the input
signal converges. The following result is from \cite[Lemma 35]{\JOURNALOFCOMPLEXITY}

\begin{lemma}[Sample and hold]\label{lem:sample}
Let $\tau\in\Rp$ and $I=[a,b]\subsetneq[0,\tau]$. Then there exists $\sample_{I,\tau}\in\gpval$
with the following properties.
Let $y:\Rp\rightarrow\R$, $y_0\in\R$, $x,e\in C^0(\Rp,\R)$ and
$\mu:\Rp\rightarrow\Rp$ be an increasing function. Suppose that for all $t\in\Rp$ we have
\[y(0)=y_0,\qquad y'(t)=\sample_{I,\tau}(t,\mu(t),y(t),x(t))+e(t).\]
Then:
\[|y(t)|\leqslant2+\smashoperator{\int_{\max(0,t-\tau-|I|)}^t}|e(u)|du+\max\left(|y(0)|\indicator{[0,b]}(t),\pastsup{\tau+|I|}|x|(t)\right)\]
Furthermore:
\begin{itemize}
\item If $t\notin I\pmod{\tau}$ then $|y'(t)|\leqslant e^{-\mu(t)}+|e(t)|$.
\item for $n\in\N$, if there exist $\bar{x}\in\R$ and $\nu,\nu'\in\Rp$ such that
$|\bar{x}-x(t)|\leqslant e^{-\nu}$ and $\mu(t)\geqslant\nu'$ for
all $t\in n\tau+I$ then
\[|y(n\tau+b)-\bar{x}|\leqslant\int_{n\tau+I}|e(u)|du+e^{-\nu}+e^{-\nu'}.\]
\item For $n\in\N$, if there exist $\check{x},\hat{x}\in\R$ and $\nu\in\Rp$ such that $x(t)\in[\check{x},\hat{x}]$ and $\mu(t)\geqslant\nu$ for
all $t\in n\tau+I$ then
\[y(n\tau+b)\in[\check{x}-\varepsilon,\hat{x}+\varepsilon]\]
where $\varepsilon=2e^{-\nu}+\int_{n\tau+I}|e(u)|du$.
\item For any $J=[c,d]\subseteq\Rp$, if there exist $\nu,\nu'\in\Rp$ and $\bar{x}\in\R$ such that
$\mu(t)\geqslant\nu'$ for all $t\in J$ and $|x(t)-\bar{x}|\leqslant e^{-\nu}$ for all $t\in J\cap(n\tau+I)$ for some $n\in\N$,
then
\[|y(t)-\bar{x}|\leqslant e^{-\nu}+e^{-\nu'}+\int_{t-\tau-|I|}^t|e(u)|du\] for all $t\in[c+\tau+|I|,d]$.
\item If there exists $\myOmega:\Rp\rightarrow\Rp$ such that for any $J=[c,d]$ and $\bar{x}\in\R$ such that for all
$\nu\in\Rp$, $n\in\N$ and $t\in(n\tau+I)\cap[c+\myOmega(\nu),d]$, $|\bar{x}-x(t)|\leqslant e^{-\nu}$; then
\[|y(t)-\bar{x}|\leqslant e^{-\nu}+\int_{t-\tau-|I|}^t|e(u)|du\]
for all $t\in[c+\myOmega^*(\nu),d]$ where
\[\myOmega^*(\nu)=\max(\myOmega(\nu+\ln(2+\tau)),\mu^{-1}(\nu+\ln(2+\tau)))+\tau+|I|.\]
\end{itemize}
\end{lemma}

Another tool is that of ``digit extraction''. In Theorem~\ref{th:decoding} we saw that
we can decode a value, as long as we are close enough to a word. In essence,
this theorem works around the continuity problem by creating gaps in the domain
of the definition. This approach does not help on the rare occasions when we really
want to extract some information about the encoding. How is it possible to achieve this
without breaking the continuity requirement ? The compromise is to ask for \emph{less
information}. More precisely, write $x=\sum_{n=0}^\infty d_i2^{-i}$, we call $d_i$
is the $i^{th}$ digit. The function that maps $x$ to $d_i$ is not continuous. Instead,
we compute $\cos(\sum_{n\geqslant i}d_i2^{-i})$. Intuitively, this is the next best thing
we can hope for if we want a continuous map: it does not give us $d_i$ but still gives
us enough information.
For simplicity, we will only state this result for the binary encoding.

\begin{lemma}[Extraction]\label{th:extract}
For any $k\geqslant2$, there exists $\myop{extract}_k\in\mygpc$ such that for any $x\in\R$ and $n\in\N$:
\[\myop{extract}_k(x,n)=\cos(2\pi k^nx).\]
\end{lemma}

\begin{proof}
Let $T_k$ be the $k^{th}$ Tchebychev polynomial. It is a well-known fact that for every $\theta\in\R$,
\[\cos(k\theta)=T_k(\cos\theta).\]
For any $x\in[-1,1]$, let
\[f(x)=T_k(x).\]
Then $f([-1,1])=[-1,1]$ and $f\in\mygpc$ because $T_k$ is a polynomial with integer
coefficients. We can thus iterate $f$ and get that for any $x\in\R$ and $n\in\N$,
\[\cos(2\pi k^nx)=\fiter{f}{n}(\cos(2\pi x)).\]
In order to apply Theorem~\ref{th:gpac_comp_iter}, we need to check some hypothesis.
Since $f$ is bounded by $1$, clearly for all $x\in[-1,1]$,
\[\infnorm{\fiter{f}{n}(x)}\leqslant1.\]
Furthermore, $f$ is $C^1$ on $[-1,1]$ which is a compact set, thus $f$ is a Lipschitz
function. We hence conclude that  Theorem~\ref{th:gpac_comp_iter} can
be applied using Remark~\ref{rem:gpac_iter_classic_err} and $f_0^*\in\mygpc$.
For any $x\in\R$ and $n\in\N$, let
\[\myop{extract}_k(x,n)=f_0^*(\cos(2\pi x),n).\]
Since $f_0^*,\cos\in\mygpc$ then $\myop{extract}_k\in\mygpc$. And by construction,
\[\myop{extract}_k(x,n)=\fiter{f}{n}(\cos(2\pi x))=\cos(2\pi x k^n).\]
\end{proof}

\subsubsection{Proof of Theorem \ref{th:gpac_comp_iter}}

\begin{proof}
We use three variables $y$, $z$ and $w$ and build a cycle to be repeated $n$ times. At all time, $y$ is an online system
computing $f(w)$. During the first stage of the cycle, $w$ stays still and $y$ converges to $f(w)$.
During the second stage of the cycle, $z$ copies $y$ while $w$ stays still. During
the last stage, $w$ copies $z$ thus effectively computing one iterate.

A crucial point is in the error estimation, which we informally develop here.
Denote the $k^{th}$ iterate of $x$ by $x^{[k]}$ and by $x^{(k)}$ the point computed
after $k$ cycles in the system. Because we are doing an approximation of $f$ at each
step step, the relationship between the two is that $x_0=x^{[0]}$ and
$\infnorm{x^{(k+1)}-f(x_k)}\leqslant e^{-\nu_{k+1}}$ where $\nu_{k+1}$ is the precision
of the approximation, that we control. Define $\mu_k$ the precision we need to achieve at step $k$:
$\infnorm{x^{(k)}-x^{[k]}}\leqslant e^{-\mu_k}$ and $\mu_n=\mu$. The triangle inequality
ensures that the following choice of parameters is safe:
\[\nu_k\geqslant\mu_k+\ln2\qquad
\mu_{k-1}\geqslant\mho\left(\infnorm{x^{[k-1]}}\right)+\mu_k+\ln2\]
This is ensured by taking $\mu_k\geqslant\sum_{i=k}^{n-1}\mho(\Pi(\infnorm{x},i))+\mu+(n-k)\ln2$
which is indeed polynomial in $k$, $\mu$ and $\infnorm{x}$. Finally a point
worth mentioning is that the entire reasoning makes sense because the assumption
ensures that $x^{(k)}\in I$ at each step.

Formally, apply Theorem~\ref{th:main_eq} to get that
$f\in\guc{\Upsilon}{\myOmega}{\Lambda}{\Theta}$
where $\Upsilon,\Lambda,\Theta,\myOmega$ are polynomials.
Without loss of generability we assume that $\Upsilon,\Lambda,\Theta,\mho$ and $\Pi$
are increasing functions.
Apply Lemma~38 (AXP time rescaling) of \cite{\JOURNALOFCOMPLEXITY} to get that $\myOmega$
can be assumed constant. Thus there exists $\omega\in[1,+\infty[$ such that for all $\alpha\in\R,\mu\in\Rp$
\[\myOmega(\alpha,\mu)=\omega\geqslant1.\]
Hence $f\in\guc{\Upsilon}{\myOmega}{\Lambda}{\Theta}$ with corresponding $\delta,d$ and $g$. Define:
\[\tau=\omega+2.\]
We will show that $f_0^*\in\mygpc$.
Let $n\in\N$, $x\in I_n$, $\mu\in\Rp$ and consider the following system:
\[
\left\{\begin{array}{@{}r@{}l}
\ell(0)&=\norm_{\infty,1}(x)\\
\mu(0)&=\mu\\
n(0)&=n\\
\end{array}\right.
\qquad
\left\{\begin{array}{@{}r@{}l}
\ell'(t)&=0\\
\mu'(t)&=0\\
n'(t)&=0
\end{array}\right.
\qquad
\left\{\begin{array}{@{}r@{}l}
y(0)&=0\\
z(0)&=x\\
w(0)&=x
\end{array}\right.
\]
\[
\left\{\begin{array}{@{}r@{}l}
y'(t)&=g(t,y(t),w(t),\nu(t))\\
z'(t)&=\sample_{[\omega,\omega+1],\tau}(t,\nu(t),z(t),y_{1..n}(t))\\
w'(t)&=\hxl_{[0,1]}(t-n\tau,\nu(t)+t,\sample_{[\omega+1,\omega+2],\tau}(t,\nu^*(t)+\ln(1+\omega),w(t),z(t)))\\
\end{array}\right.
\]
\[\ell^*=1+\Pi(\ell,n)\qquad
\nu=n\mho(\ell^*)+n\ln6+\mu+\ln3\qquad
\nu^*=\nu+\Lambda(\ell^*,\nu)\]

First notice that $\ell,\mu$ and $n$ are constant functions and we identify $\mu(t)$ with $\mu$
and $n(t)$ with $n$. Apply Lemma~\ref{lem:norm} to get that $\infnorm{x}\leqslant\ell\leqslant\infnorm{x}+1$,
so in particular $\ell^*,\nu$ and $\nu^*$ are polynomially bounded in $\infnorm{x}$ and $n$.
We will need a few notations:
for $i\in\intinterv{0}{n}$, define $x^{[i]}=\fiter{f}{i}(x)$ and $x^{(i)}=w(i\tau)$. Note
that $x^{[0]}=x^{(0)}=x$. We will show by induction for $i\in\intinterv{0}{n}$ that
\[\infnorm{x^{(i)}-x^{[i]}}\leqslant e^{-(n-i)\mho(\ell^*)-(n-i)\ln6-\mu-\ln3}.\]
Note that this is trivially true for $i=0$. Let $i\in\intinterv{0}{n-1}$
and assume that the result is true for $i$. We will show that it holds for $i+1$
by analyzing the behavior of the various variables in the system during period $[i\tau,(i+1)\tau]$.
\begin{itemize}
\item\textbf{For $y$ and $w$, if $t\in[i\tau,i\tau+\omega+1]$ then} apply Lemma~\ref{lem:lxh_hxl} to
get that $\hxl\in[0,1]$ and Lemma~\ref{lem:sample} to get that $\infnorm{w'(t)}\leqslant e^{-\nu^*-\ln(1+\omega)}$.
Conclude that $\infnorm{w(i)-w(t)}\leqslant e^{-\nu^*}$, in other words
$\infnorm{w(t)-x^{(i)}}\leqslant e^{-\Lambda(\infnorm{x^{(i)}},\nu)}$ since
$\infnorm{x^{(i)}}\leqslant\infnorm{x^{[i]}}+1\leqslant1+\Pi(\infnorm{x},i)\leqslant\ell^*$
and $\nu^*\geqslant\Lambda(\ell^*,\nu)$.
Thus, by definition of \unaware{} computability, $\infnorm{f(x^{(i)})-y_{1..n}(u)}\leqslant e^{-\nu}$
if $u\in[i\tau+\omega,i\tau+\omega+1]$ because $\myOmega\left(\infnorm{x^{(i)}},\nu\right)=\omega$.
\item\textbf{For $z$, if $t\in[i\tau+\omega,i\tau+\omega+1]$ then} apply Lemma~\ref{lem:sample}
to get that
\[\infnorm{f(x^{(i)})-z(i\tau+\omega+1)}\leqslant2e^{-\nu}.\]
Notice that
we ignore the behavior of $z$ during $[i\tau,i\tau+\omega]$ in this part of the proof.
\item\textbf{For $z$ and $w$, if $t\in[i\tau+\omega+1,i\tau+\omega+2]$ then} apply Lemma~\ref{lem:sample}
to get that $\infnorm{z'(t)}\leqslant e^{-\nu}$ and thus $\infnorm{f(x^{(i)})-z(t)}\leqslant3e^{-\nu}$.
Apply Lemma~\ref{lem:lxh_hxl} to get that
\[\infnorm{y'(t)-\sample_{[\omega+1,\omega+2],\tau}(t,\nu^*+\ln(1+\omega),w(t),z(t))}\leqslant e^{-\nu-t}.\]
Apply Lemma~\ref{lem:sample} again to get that $\infnorm{f(x^{(i)})-w(i\tau+\omega+2)}\leqslant4e^{-\nu}+e^{-\nu^*}\leqslant5e^{-\nu}$.
\end{itemize}
Our analysis concluded that $\infnorm{f(x^{(i)})-w((i+1)\tau)}\leqslant5e^{-\nu}$.
Also, by hypothesis, $\infnorm{x^{(i)}-x^{[i]}}\leqslant e^{-(n-i)\mho(\ell^*)-(n-i)\ln6-\mu-\ln3}\leqslant
e^{-\mho\left(\infnorm{x^{[i]}}\right)-\mu^*}$
where $\mu^*=(n-i-1)\mho(\ell^*)+(n-i)\ln6+\mu+\ln3$ because $\infnorm{x^{[i]}}\leqslant\ell^*$.
Consequently, $\infnorm{f(x^{(i)})-x^{[i+1]}}\leqslant e^{-\mu^*}$ and thus:
\[\infnorm{x^{(i+1)}-x^{[i+1]}}\leqslant 5e^{-\nu}+e^{-\mu^*}\leqslant6e^{-\mu^*}
\leqslant e^{-(n-1-i)\mho(\ell^*)-(n-1-i)\ln6-\mu-\ln3}.\]

From this induction we get that $\infnorm{x^{(n)}-x^{[n]}}\leqslant e^{-\mu-\ln3}$.
We still have to analyze the behavior after time $n\tau$.
\begin{itemize}
\item \textbf{If $t\in[n\tau,n\tau+1]$ then} apply Lemma~\ref{lem:sample} and Lemma~\ref{lem:lxh_hxl}
to get that $\infnorm{w'(t)}\leqslant e^{-\nu^*-\ln(1+\omega)}$ thus
$\infnorm{w(t)-x^{(n)}}\leqslant e^{-\nu^*-\ln(1+\omega)}$.
\item \textbf{If $t\geqslant n\tau+1$ then} apply Lemma~\ref{lem:lxh_hxl}
to get that $\infnorm{w'(t)}\leqslant e^{-\nu-t}$ thus $\infnorm{w(t)-w(n\tau+1)}\leqslant e^{-\nu}$.
\end{itemize}
Putting everything together we get for $t\geqslant n\tau+1$ that:
\begin{align*}
\infnorm{w(t)-x^{[n]}}&\leqslant e^{-\mu-\ln3}+e^{-\nu^*-\ln(1+\omega)}+e^{-\nu}\\
    &\leqslant3e^{-\mu-\ln3}\leqslant e^{-\mu}.
\end{align*}
We also have to show that the system does not grow too fast. The analysis during
the time interval $[0,n\tau+1]$ has already been done (although we did not write
all the details, it is an implicit consequence). For $t\geqslant n\tau+1$,
have $\infnorm{w(t)}\leqslant\infnorm{x^{[n]}}+1\leqslant\Pi(\infnorm{x},n)+1$
which is polynomially bounded. The bound on $y$ comes from \unaware{} computability:
\[
\infnorm{y(t)}\leqslant\Upsilon\left(\pastsup{\delta}{\infnorm{w}}(t),\nu,0\right)
    \leqslant\Upsilon(\Pi(\infnorm{x},n),\nu,0)
    \leqslant\poly(\infnorm{x},n,\mu)
\]
And finally, apply Lemma~\ref{lem:sample} to get that:
\[\infnorm{z(t)}\leqslant2+\pastsup{\tau+1}{\infnorm{y_{1..n}}}(t)\leqslant\poly(\infnorm{x},n,\mu)
\]
This conclude the proof that $f_0^*\in\mygpc$.

{
We can now tackle the case of $\eta>0$. Let $\eta\in]0,\tfrac{1}{2}[$ and $\mu_\eta\in\Q$ such that
$\tfrac{1}{2}-e^{-\mu_\eta}<\eta$. Let $f_\eta^*(x,u)=f_0^*(x,\crnd(u,\mu_\eta))$.
Apply 
Theorem~\ref{th:gpac_comp_composition}
to conclude that $f_\eta^*\in\mygpc$. By definition of $\crnd$, if $u\in\left]n-\eta,n+\eta\right[$
for some $n\in\Z$ then $\crnd(x,\mu)=n$ and thus
$f_\eta^*(x,u)=f_0^*(x,n)=x^{[n]}$.
}
\end{proof}

\subsection{Cauchy completion and complexity}
\label{sec:th:cont_redux}

The purpose of this section is to prove Theorem \ref{th:cont_redux}.

Given $x\in\mathcal{D}$ and $n\in\N$, we want to use $g$ to compute
an approximation of $f(x)$ within $2^{-n}$.
To do so, we use the ``modulus of continuity'' $\mho$ to find a dyadic rational $(q,n)$ such that
$\infnorm{x-q}\leqslant 2^{-\mho(\infnorm{x},n)}$.
We then compute $g(q,n)$ and get that $\infnorm{g(q,n)-f(x)}\leqslant 2^{-n}$.

There are two problems with this approach. First, finding such a dyadic rational is not possible
because it is not a continuous operation. Indeed, consider the mapping $(x,n)\mapsto(q,n)$
that satisfies the above condition: if it is computable, it must be continuous.
But it cannot be continuous because its image is completed disconnected.
This is where mixing comes into play: given $x$ and $n$,
we will compute two dyadic rationals $(q,n)$ and $(q',n')$ such that at least one of them satisfies
the above criteria. We will then apply $g$ on both of them and mix the result.
The idea is that if both are valid, the outputs will be very close (because of the modulus of
continuity) and thus the mixing will give the correct result. See Section~\ref{sec:mixing} for
more details on mixing. The case of multidimensional domains is similar
except that we need to mix on all dimensions simultaneously, thus we need
roughly $2^d$ mixes to ensure that at least one is correct, where $d$ is the dimension.

\begin{proof}[(of Theorem~\ref{th:cont_redux})]
We will show the result by induction on $d$. If $d=0$ then $\infnorm{g(n,y)-f(y)}\leqslant 2^{-n}$ for all
$n\in\N,y\in\mathcal{D}$. We can thus apply Theorem~\ref{th:gpac_comp_limit} to get that $f\in\mygpc$.

Assume that $d>0$. Let $\kappa:\{0,1\}\to\{0,1\},x\mapsto x$ and $\pi_i$ denote the $i^{th}$
projection. For any relevant\footnote{Domain of definition is discussed below.} $u\in\R$
and $n\in\N$ and $\delta\in\{0,1\}$, let
\begin{align*}
    v_\delta(u,n)&=v(u-\tfrac{\delta}{2}2^{-n},n)\\
    v(u,n)&=r(u,n)+v^*(u-r(u,n),n)\\
    v^*(u,n)&=\pi_1(\myop{decode}_{\kappa}(u,n))\\
    r(u,n)&=\crnd(u-\tfrac{1}{2}-e^{-\nu},\nu)\text{ where }\nu=\ln6+n\ln 2.
\end{align*}
We now discuss the domain of definition and properties of these functions.
First $r\in\mygpc$ since $\crnd\in\mygpc$ by Theorem~\ref{th:comp:round}. Furthermore,
by definition of $\crnd$ we have that
\[\text{if }u\in m+\left[0,1-\tfrac{1}{3}2^{-n}\right]\text{ for some }m\in\Z\text{ then }r(u,n)=m.\]
Indeed since $2e^{-\nu}=\tfrac{1}{3}2^{-n}$,
\[\begin{array}{c}
m\leqslant u\leqslant m+1-2e^{-\nu}\\
m-\tfrac{1}{2}+e^{-\nu}\leqslant u-\tfrac{1}{2}+e^{-\nu}\leqslant m+\tfrac{1}{2}-e^{-\nu}
\end{array}\]
thus $r(u,n)=\crnd(u-\tfrac{1}{2}+e^{-\nu},\nu)=m$.
We now claim that we have that
\[\text{if }u=\frac{p}{2^n}+\varepsilon\text{ for some }p\in\Z\text{ and }
\varepsilon\in\left[0,2^{-n}\tfrac{2}{3}\right]\text{ then }
v(u,n)=\frac{p}{2^n}.\]
Indeed, write $p=m2^n+p'$ where $m\in\Z$ and $p'\in\intinterv{0}{2^n-1}$. Then $u=m+\frac{p'}{2^n}+\varepsilon$
and
\[\frac{p'}{2^n}+\varepsilon\leqslant\frac{2^n-1}{2^n}+\varepsilon\leqslant 1-2^{-n}+\frac{2}{3}2^{-n}\leqslant1-\frac{1}{3}2^{-n}.\]
Thus $r(u,n)=m$ and $u-r(u,n)=\frac{p'}{2^n}+\varepsilon$.
Since $p'\in\intinterv{0}{2^n-1}^d$, there exist $w_1,\ldots,w_d\in\{0,1\}$ such that
\[\frac{p'}{2^n}=\sum_{j=1}^nw_{j}2^{-j}.\]
It follows from Theorem~\ref{th:decoding} and the fact that $1-e^{-2}\geqslant\tfrac{2}{3}$
that\footnote{The $*$ denotes ``anything'' because we do not care about the actual value.}
\[\myop{decode}_{\kappa}\left(\tfrac{p'}{2^n}+\varepsilon,n,2\right)
    =\myop{decode}_{\kappa}\left(\sum_{j=1}^nw_{j}2^{-j}+\varepsilon,n,2\right)
    =\left(\sum_{j=1}^nw_{j}2^{-j},*\right)=\left(\tfrac{p'}{2^n},*\right).\]
Consequently,
\begin{align*}
v(u,n)
    &=r(u,n)+v^*\left(u-r(u,n),n\right)\\
    &=m+v^*\left(\tfrac{p'}{2^n}+\varepsilon,n\right)\\
    &=m+\pi_1\left(\myop{decode}_{\kappa}(\tfrac{p'}{2^n}+\varepsilon,n)\right)\\
    &=m+\pi_1\left(\tfrac{p'}{2^n},*\right)\\
    &=m+\tfrac{p'}{2^n}\\
    &=\tfrac{p}{2^n}.
\end{align*}
To summarize, we have shown that
\begin{equation*}
\text{if }u=\frac{p}{2^n}+\varepsilon\text{ for some }p\in\Z\text{ and }
\varepsilon\in\left[0,\tfrac{2}{3}2^{-n}\right]\text{ then }v(u,n)=\frac{p}{2^n}.
\end{equation*}
and thus that for all $\delta\in\{0,1\}$,
\begin{equation}\label{eq:prop_v_delta}
\text{if }u=\frac{p}{2^n}+\frac{\delta}{2}2^{-n}+\varepsilon\text{ for some }p\in\Z\text{ and }
\varepsilon\in\left[0,\tfrac{2}{3}2^{-n}\right]\text{ then }v_\delta(u,n)=\frac{p}{2^n}.
\end{equation}
Before we proceed to mixing, we need an auxiliary function.
For all $u\in\R$ and $n\in\N$, define
\[\myop{sel}(u,n)=\tfrac{1}{2}+\functionextract_2\left(u+\tfrac{1}{6}2^{-n},n\right)\]
where $\functionextract_2$ is given by Lemma~\ref{th:extract}.
We claim that for all $n\in\N$,
\begin{equation}\label{eq:dom_i_lt_1_cup_p}
\left\{u\in\R: \myop{sel}(u,n)<1\right\}
    \subseteq \left(2^{-n}\Z+\left[0,\tfrac{2}{3}2^{-n}\right]\right)\times\{n\}
\end{equation}
and
\begin{equation}\label{eq:dom_i_gt_0_cup_p}
\left\{u\in\R: \myop{sel}(u,n)>0\right\}
    \subseteq \left(2^{-n}\Z+\left[-\tfrac{1}{2}2^{-n},\tfrac{1}{6}2^{-n}\right]\right)\times\{n\}.
\end{equation}
Indeed, by definition of $\functionextract_2$, if $u=\frac{p}{2^n}+\varepsilon$ with $\varepsilon\in\left[0,2^{-n}\right[$,
then
\begin{align*}
\myop{sel}(u,n)
    &=\tfrac{1}{2}+\functionextract_2\left(\tfrac{p}{2^n}+\varepsilon+\tfrac{1}{6}2^{-n},n\right)\\
    &=\tfrac{1}{2}+\cos(2\pi 2^n(\tfrac{p}{2^n}+\varepsilon+\tfrac{1}{6}2^{-n}))\\
    &=\tfrac{1}{2}+\cos(2\pi p+2\pi 2^n\varepsilon+\tfrac{\pi}{3})\\
    &=\tfrac{1}{2}+\cos(2\pi 2^n\varepsilon+\tfrac{\pi}{3})
\end{align*}
where $2\pi 2^n\varepsilon\in[0,2\pi[$ and thus $2\pi 2^n\varepsilon+\tfrac{\pi}{3}\in[\tfrac{\pi}{3},\tfrac{7\pi}{3}]$.
Consequently,
\begin{align*}
\myop{sel}(u,n)<1
    &\Leftrightarrow \varepsilon\in\left]0,\tfrac{2}{3}2^{-n}\right[.
\end{align*}
And similarly,
\begin{align*}
\myop{sel}(u,n)>0
    &\Leftrightarrow \varepsilon\in\left[0,\tfrac{1}{6}2^{-n}\right[\cup\left]\tfrac{1}{2}2^{-n},2^{-n}\right].
\end{align*}
Now define for all relevant\footnote{We will discuss the domain of definition below.}
$q\in\Q^{d-1},n\in\N,z\in\R,y\in\R^e,\delta\in\{0,1\}$,
\begin{align*}
\tilde{g}_\delta(q,m,z,y)&=g(q,v_\delta(z,m),m,y),\\
\widetilde{\myop{sel}}(q,m,z,y)&=\myop{sel}(z,m),\\
\tilde{g}(q,m,z,y)&=\mix{\widetilde{\myop{sel}}}{\tilde{g}_0}{\tilde{g}_1}(q,m,z,y).
\end{align*}
For any $\alpha\in\Rp$ and $n\in\N$, define
\[\mho^*(\alpha,n)=\mho(\alpha,n)+1.\]
Let $(x,z,y)\in\mathcal{D}$ and $n,m\in\N,p\in\Z^{d-1}$, such that
\begin{equation}\label{eq:assumption_cont_redux_inductive}
\infnorm{\tfrac{p}{2^m}-x}\leqslant 2^{-m}\text{ and }m\geqslant\mho^*(\infnorm{(x,z,y)},n).
\end{equation}
Let $q=\tfrac{p}{2^m}$. There are three cases:
\begin{itemize}
\item \textbf{If $\mathbf{\widetilde{\myop{sel}}(q,m,z,y)=0}$:} then $\myop{sel}(z,m)=0$.
But then $\myop{sel}(z,m)<1$
so by \eqref{eq:dom_i_lt_1_cup_p}, $z\in 2^{-m}\Z+\left[0,\tfrac{2}{3}2^{-m}\right]$.
Write $z=p'2^{-m}+\varepsilon$ where $p'\in\Z$ and $\varepsilon\in[0,\tfrac{2}{3}2^{-m}]$.
Then
$v_0(z,,)=\frac{p'}{2^m}$ using \eqref{eq:prop_v_delta}. It follows that,
\begin{align*}
\infnorm{(x,z)-(q,v_0(z,m))}
    &=\max(\infnorm{x-q},|z-\tfrac{p'}{2^m}|)\\
    &=\max(\infnorm{x-q},|\varepsilon|)&&\text{since }z=p'2^{-m}+\varepsilon\\
    &\leqslant\max(2^{-\mho^*(\infnorm{(x,z,y)},n)},|\varepsilon|)&&\text{by assumption on }q\\
    &\leqslant\max(2^{-\mho^*(\infnorm{(x,z,y)},n)},\tfrac{2}{3}2^{-m})\\
    &\leqslant 2^{-\mho^*(\infnorm{(x,z,y)},n)}&&\text{since }m\geqslant\mho^*(\infnorm{(x,z,y)},n)\\
    &\leqslant 2^{-\mho(\infnorm{(x,z,y)},n)-1}&&\text{by definition of }\mho.
\end{align*}
It follows by assumption on $g$ that $\infnorm{g(q,v_0(z,m),m,y)-f(x,z,y)}\leqslant 2^{-n-1}$.
But since $\widetilde{\myop{sel}}(q,m,z,y)=0$, the
\[\tilde{g}(q,m,z,y)=\tilde{g}_0(q,m,z,y)=g(q,v_0(z,m),m,y),\]
thus $\infnorm{\tilde{g}(q,m,z,y)-f(x,z,y)}\leqslant 2^{-n-1}\leqslant 2^{-n}$.

\item \textbf{If $\mathbf{\widetilde{\myop{sel}}(q,m,z,y)=1}$:} then $\myop{sel}(z,m)=1$.
But then $\myop{sel}(z,m)>0$
so by \eqref{eq:dom_i_gt_0_cup_p}, $z\in 2^{-m}\Z+\left[-\tfrac{1}{2}2^{-m},\tfrac{1}{6}2^{-m}\right]$.
Write $z=p'2^{-m}+\varepsilon$ where $p'\in\Z$ and $\varepsilon\in\left[-\tfrac{1}{2}2^{-m},\tfrac{1}{6}2^{-m}\right]$.
Then
$v_1(z,m)=\frac{p'}{2^m}$ using \eqref{eq:prop_v_delta}. It follows that,
\begin{align*}
\infnorm{(x,z)-(q,v_1(z,m))}
    &=\max(\infnorm{x-q},|z-\tfrac{p'}{2^m}|)\\
    &\leqslant 2^{-\mho(\infnorm{(x,z,y)},n)-1}
\end{align*}
using the same chain of inequalities as in the previous case.
It follows by assumption on $g$ that $\infnorm{g(q,v_1(z,m),m,y)-f(x,z,y)}\leqslant 2^{-n}$.
But since $\widetilde{\myop{sel}}(q,m,z,y)=1$, then $\tilde{g}(q,m,z,y)=\tilde{g}_1(q,m,z,y)=g(q,v_1(z,m),m,y)$,
thus $\infnorm{\tilde{g}(q,m,z,y)-f(x,z,y)}\leqslant 2^{-n-1}\leqslant 2^{-n}$.

\item \textbf{If $\mathbf{0<\widetilde{\myop{sel}}(q,m,z,y)<1}$:} then
\[\tilde{g}(q,m,z,y)=(1-\alpha)\tilde{g}_0(q,m,z,y)+\alpha\tilde{g}_1(q,m,z,y)\]
where $\alpha=\myop{sel}(z,m)\in]0,1[$. Using the same reasoning as in the previous two
cases we get that
\[
\infnorm{\tilde{g}_0(q,m,z,y)-f(x,z,y)}\leqslant 2^{-n-1}
\text{ and }
\infnorm{\tilde{g}_1(q,m,z,y)-f(x,z,y)}\leqslant 2^{-n-1}.
\]
It easily follows that
\[\infnorm{\tilde{g}(q,m,z,y)-f(x,z,y)}\leqslant 2\alpha\cdot2^{-n-1}\leqslant2^{-n}.\]
\end{itemize}
To summarize, we have shown that under assumption \eqref{eq:assumption_cont_redux_inductive}
we have that
\[\infnorm{g(\tfrac{p}{2^m},m,z,y)-f(x,z,y)}\leqslant 2^{-n}.\]
And since since $\tilde{g}\in\mygpc$, we can apply the result inductively to $\tilde{g}$ (which
has only $d-1$ dyadic arguments) to conclude.
\end{proof}

\subsection{Proof of Theorem \ref{th:decoding}: Word decoding}
\label{proof_th:decoding}

\begin{proof}
We will iterate a function that works on tuple
of the form $(x,x',n,m,\mu)$ where $x$ is the remaining part to process, $x'$ is the processed part, $n$ the length
of the processed part, $m$ the number of nonzero symbols and $\mu$ will stay constant.
The function will remove the ``head'' of $x$, re-encode it with $\kappa$ and ``queue'' on $x'$,
increasing $n$ and $m$ if the head is not $0$.

In the remaining of this proof, we write $\ovl{0.x}{}^{k_i}$ to denote $0.x$ in basis $k_i$ instead of $k$.
{
Define for any $x,y\in\R$ and $n\in\N$:
\[g(x,y,n,m,\mu)=\big(\fracp^*(k_1x),y+k_2^{-n-1}\lagrange{\kappa}(\intp^*(k_1x)),n+1,m+\lagrneq{\idfun}{0}(\intp^*(k_1x)),\mu\big)\]
where
\[\intp^*(x)=\crnd\left(x-\tfrac{1}{2}+\tfrac{3e^{-\mu}}{4},\mu\right)\qquad\fracp^*(x)=x-\intp^*(x)\]
and $\crnd$ is defined in Definition~\ref{def:comp:round}.
Apply Lemma~\ref{lem:lagrange_interp} to get that $\lagrange{\kappa}\in\mygpc$
and Lemma~\ref{lem:char_interp} to get that $\lagrneq{\idfun}{0}\in\mygpc$. It follows that $g\in\mygpc$.
We need a small result about $\intp^*$ and $\fracp^*$. For any $w\in\intinterv{0}{k_1}^*$ and $x\in\R$, define
the following proposition:
\[A(x,w,\mu):-k_1^{-|w|}\frac{e^{-\mu}}{2}\leqslant x-\ovl{0.w}{}^{k_1}
\leqslant k_1^{-|w|}(1-e^{-\mu}).\]
We will show that:
\begin{equation}\label{eq:decode:1}
A(x,w,\mu)\Rightarrow\left\{\begin{array}{l}
\intp^*(k_1x)=\intp(k_1\ovl{0.w}{}^{k_1})\\
\left|\fracp^*(k_1x)-\fracp(k_1\ovl{0.w}{}^{k_1})\right|\leqslant k_1\left|x-\ovl{0.w}{}^{k_1}\right|
\end{array}\right..\end{equation}
Indeed, in this case, since $|w|\geqslant1$, we have that
\[\begin{array}{c}
-k_1^{1-|w|}\frac{e^{-\mu}}{2}\leqslant k_1x-k_1\ovl{0.w}{}^{k_1}\leqslant k_1^{1-|w|}(1-e^{-\mu})\\
-k_1^{1-|w|}\frac{e^{-\mu}}{2}\leqslant k_1x-w_1\leqslant k_1^{1-|w|}(1-e^{-\mu})+\ovl{0.w_{2..|w|}}{}^{k_1}\\
-\frac{e^{-\mu}}{2}\leqslant k_1x-w_1\leqslant k_1^{1-|w|}-e^{-\mu}+\sum_{i=1}^{|w|-1}(k_1-1)k_1^{-i}\\
-\frac{e^{-\mu}}{2}\leqslant k_1x-w_1\leqslant k_1^{1-|w|}-e^{-\mu}+1-k_1^{1-|w|}\\
-\frac{1}{2}-\frac{e^{-\mu}}{2}\leqslant k_1x-\frac{1}{2}-w_1\leqslant \frac{1}{2}-e^{-\mu}\\
-\frac{1}{2}+\frac{e^{-\mu}}{4}\leqslant k_1x-\frac{1}{2}+\frac{3e^{-\mu}}{4}-w_1\leqslant \frac{1}{2}-\frac{e^{-\mu}}{4}\\
\end{array}\]
And conclude by applying Theorem~\ref{th:comp:round} because $\intp(k_1\ovl{0.w}{}^{k_1})=w_1$. The result on $\fracp$ follows trivially.
It is then not hard to derive from \eqref{eq:decode:1} applied twice that:
\begin{equation}\label{eq:decode:1.1}
\begin{array}{c}
A(x,w,\mu)\quad\wedge\quad A(x',w,\mu')\\
\Downarrow\\
\infnorm{g(x,y,n,m,\mu)-g(x',y',n',m',\nu)}\leqslant2k_1\infnorm{(x,y,n,m,\mu)-(x',y',n',m',\mu')}.
\end{array}
\end{equation}
It also follows that proposition $A$ is preserved by applying $g$:
\begin{equation}\label{eq:decode:2}
A(x,w,\mu)\quad\Rightarrow\quad A(\fracp(k_1x),w_{2..|w|},\mu).
\end{equation}
Furthermore, $A$ is stronger for longer words:
\begin{equation}\label{eq:decode:3}
A(x,w,\mu)\quad\Rightarrow\quad A(x,w_{1..|w|-1},\mu).
\end{equation}
Indeed, if we have $A(x,w,\mu)$ then:
\[\begin{array}{c}
-k_1^{-|w|}\frac{e^{-\mu}}{2}\leqslant x-\ovl{0.w}{}^{k_1}\leqslant k_1^{-|w|}(1-e^{-\mu})\\
-k_1^{-|w|}\frac{e^{-\mu}}{2}\leqslant x-\ovl{0.w_{1..|w|-1}}{}^{k_1}\leqslant k_1^{-|w|}(1-e^{-\mu})+w_{|w|}k_1^{-|w|}\\
-k_1^{1-|w|}\frac{e^{-\mu}}{2}\leqslant x-\ovl{0.w_{1..|w|-1}}{}^{k_1}\leqslant k_1^{-|w|}(1-e^{-\mu})+(k_1-1)k_1^{-|w|}\\
-k_1^{1-|w|}\frac{e^{-\mu}}{2}\leqslant x-\ovl{0.w_{1..|w|-1}}{}^{k_1}\leqslant k_1^{-|w|}(k_1-e^{-\mu})\\
-k_1^{1-|w|}\frac{e^{-\mu}}{2}\leqslant x-\ovl{0.w_{1..|w|-1}}{}^{k_1}\leqslant k_1^{1-|w|}(1-e^{-\mu})\\
\end{array}\]
It also follows from the definition of $g$ that:
\begin{equation}\label{eq:decode:4}
A(x,w,\mu)\quad\Rightarrow\quad\infnorm{g(x,y,n,m,\mu)}\leqslant\max(k_1,1+\infnorm{x,y,n,m,\mu}).
\end{equation}
Indeed, if $A(x,w,\mu)$ then $\intp^*(k_1x)\in\intinterv{0}{k_1-1}$ thus $L_\kappa(\intp^*(k_1x))\in\intinterv{0}{k_2}$
and $\lagrneq{\idfun}{0}(\intp^*(k_1x))\in\{0,1\}$, the inequality follows easily.
A crucial property of $A$ is that it is open with respect to $x$:
\begin{equation}\label{eq:decode:5}
A(x,w,\mu)\quad\wedge\quad|x-y|\leqslant e^{-|w|\ln k_1-\mu-\nu}\quad\Rightarrow\quad A(y,w,\mu-\ln\tfrac{3}{2}).
\end{equation}
Indeed, if $A(x,w,\mu)$ and $|x-y|\leqslant e^{-|w|\ln k_1-\mu-\ln4}$
we have:
{
\[\begin{array}{c}
-k_1^{-|w|}\frac{e^{-\mu}}{2}\leqslant x-\ovl{0.w}{}^{k_1}\leqslant k_1^{-|w|}(1-e^{-\mu})\\
-k_1^{-|w|}\frac{e^{-\mu}}{2}+y-x\leqslant y-\ovl{0.w}{}^{k_1}\leqslant k_1^{-|w|}(1-e^{-\mu})+y-x\\
-k_1^{-|w|}\frac{e^{-\mu}}{2}-|y-x|\leqslant y-\ovl{0.w}{}^{k_1}\leqslant k_1^{-|w|}(1-e^{-\mu})+|y-x|\\
-k_1^{-|w|}\frac{e^{-\mu}}{2}-e^{-|w|\ln k_1-\mu-\ln4}\leqslant y-\ovl{0.w}{}^{k_1}\leqslant k_1^{-|w|}(1-e^{-\mu})+e^{-|w|\ln k_1-\mu-\ln4}\\
-k_1^{-|w|}(e^{-\mu-\ln4}+\frac{e^{-\mu}}{2})\leqslant y-\ovl{0.w}{}^{k_1}\leqslant k_1^{-|w|}(1-e^{-\mu}+e^{-\mu-\ln4})\\
-k_1^{-|w|}\frac{3e^{-\mu}}{4}\leqslant y-\ovl{0.w}{}^{k_1}\leqslant k_1^{-|w|}(1-\frac{3e^{-\mu}}{4})\\
-k_1^{-|w|}\frac{3e^{-\mu}}{4}\leqslant y-\ovl{0.w}{}^{k_1}\leqslant k_1^{-|w|}(1-\frac{6e^{-\mu}}{4})\\
-k_1^{-|w|}\frac{e^{\ln\tfrac{3}{2}-\mu}}{2}\leqslant y-\ovl{0.w}{}^{k_1}\leqslant k_1^{-|w|}(1-e^{\ln\tfrac{3}{2}-\mu})\\
\end{array}\]
}
In order to formally apply Theorem~\ref{th:gpac_comp_iter}, define for any $n\in\N$:
\[I_n=\left\{(x,y,\ell,m,\mu)\in\R^2\times\Rp^3:\exists w\in\intinterv{0}{k_1-1}^n, A(x,w,\mu)\right\}.\]
It follows from \eqref{eq:decode:3} that $I_{n+1}\subseteq I_n$. It follows from \eqref{eq:decode:2} that $g(I_{n+1})\subseteq I_n$.
It follows from \eqref{eq:decode:4} that $\infnorm{\fiter{g}{n}(x)}\leqslant\max(k_1,\infnorm{x}+n)$ for $x\in I_n$.
Now assume that $X=(x,y,n,m,\mu)\in I_n$, $\nu\in\Rp$ and\footnote{We use Remark~\ref{rem:gpac_iter_modulus_n} to allow a dependence of $\mho$ in $n$.}
$\infnorm{X-X'}\leqslant e^{-\infnorm{X}-n\ln k_1-\nu}$
where $X'=(x',y',n',m,\mu')$ then by definition $A(x,w,\mu)$ for some $w\in\intinterv{0}{k_1-1}^n$.
It follows from \eqref{eq:decode:5} that $A(y,w,\mu-\ln\tfrac{3}{2})$ since $\infnorm{X}+n\ln k_1\geqslant|w|\ln k_1+\mu$.
Thus by \eqref{eq:decode:1.1} we have $\infnorm{g(X)-g(X')}\leqslant 2k_1\infnorm{X-X'}$
which is enough by Remark~\ref{rem:gpac_iter_classic_err}. We are thus in good shape to apply Theorem~\ref{th:gpac_comp_iter}
and get $g^*_0\in\mygpc$.
Define:
\[\myop{decode}_\kappa(x,n,\mu)=\pi_{2,4}(g^*_0(x,0,0,0,\mu,n))\]
where $\pi_{2,4}(a,b,c,d,e,f,g)=(b,d)$. Clearly $\myop{decode}_\kappa\in\mygpc$,
it remains to see that it satisfies the theorem. We will prove this by induction on
the length of $|w|$. More precisely we will prove that for $|w|\geqslant0$:
\[\varepsilon\in[0,k_1^{-|w|}(1-e^{-\mu})]\quad\Rightarrow\quad
\fiter{g}{|w|}(\ovl{0.w}{}^{k_1}+\varepsilon,0,0,0,\mu)=(k_1^{|w|}\varepsilon,\ovl{0.\kappa(w)}{}^{k_2},|w|,\#\{i|w_i\neq0\},\mu).\]
The case of $|w|=0$ is trivial since it will act as the identity function:
\begin{align*}
\fiter{g}{|w|}(\ovl{0.w}{}^{k_1}+\varepsilon,0,0,0,\mu)&=\fiter{g}{0}(\varepsilon,0,0,0,\mu)\\
    &=(\varepsilon,0,0,0,\mu)\\
    &=(k_1^{|w|}\varepsilon,\ovl{0.\kappa(w)}{}^{k_2},|w|,\#\{i|w_i\neq0\},\mu).
\end{align*}
We can now show the induction step. 
Assume that $|w|\geqslant1$
and define $w'=w_{1..|w|-1}$. Let $\varepsilon\in[0,k_1^{-|w|}(1-e^{-\mu})]$
and define $\varepsilon'=k_1^{-|w|}w_{|w|}+\varepsilon$. It is clear
that $\ovl{0.w}{}^{k_1}+\varepsilon=\ovl{0.w'}{}^{k_1}+\varepsilon'$.
Then by definition $A(\ovl{0.w'}{}^{k_1}+\varepsilon',|w|,\mu)$ so
{
\begin{align*}
\fiter{g}{|w|}(\ovl{0.w}{}^{k_1}+\varepsilon,0,0,0,\mu)
    &=g(\fiter{g}{|w|-1}(\ovl{0.w'}{}^{k_1}+\varepsilon',0,0,0,\mu))\\
    &=g(k_1^{|w|-1}\varepsilon',\ovl{0.\kappa(w')}{}^{k_2},|w'|,\#\{i|w'_i\neq0\},\mu)\tag*{By induction}\\
    &=g(k_1^{-1}w_{|w|}+k_1^{|w|-1}\varepsilon,\ovl{0.\kappa(w')}{}^{k_2},|w'|,\#\{i|w'_i\neq0\},\mu)\\
    &=(\fracp^*(w_{|w|}+k_1^{|w|}\varepsilon),\tag*{Where $k_1^{-|w|}\varepsilon\in[0,1-e^{-\mu}]$}\\
    &\qquad\ovl{0.\kappa(w')}{}^{k_2}+k_2^{-|w'|-1}\lagrange{\kappa}(\intp^*(w_{|w|}+k_1^{|w|}\varepsilon)),\\
    &\qquad|w'|+1,\#\{i|w'_i\neq0\}+\lagrneq{\idfun}{0}(\intp^*(w_{|w|}+k_1^{|w|}\varepsilon)),\mu)\\
    &=(k_1^{|w|}\varepsilon,\ovl{0.\kappa(w')}{}^{k_2}+k_2^{-|w|}\lagrange{\kappa}(w_{|w|}),\\
    &\qquad|w|,\#\{i|w'_i\neq0\}+\lagrneq{\idfun}{0}(w_{|w|}),\mu)\\
    &=(k_1^{|w|}\varepsilon,\ovl{0.\kappa(w')}{}^{k_2},|w|,\#\{i|w_i\neq0\},\mu).
\end{align*}
}
We can now conclude to the result. Let $\varepsilon\in[0,k_1^{-|w|}(1-e^{-\mu})]$ then
$A(\ovl{0.0w}{}^{k_1}+\varepsilon,|w|,\mu)$ so
in particular $(\ovl{0.w}{}^{k_1}+\varepsilon,0,0,0,\mu)\in I_{|w|}$ so:
\begin{align*}
\myop{decode}_\kappa(\ovl{0.w}{}^{k_1}+\varepsilon,|w|,\mu)
    &=\pi_{2,4}(g^*_0(\ovl{0.w}{}^{k_1}+\varepsilon,0,0,0,\mu))\\
    &=\pi_{2,4}(\fiter{g}{|w|}(\ovl{0.w}{}^{k_1}+\varepsilon,0,0,0,\mu))\\
    &=\pi_{2,4}(\varepsilon,\ovl{0.\kappa(w)}{}^{k_2},|w|,\#\{i|w_i\neq0\},\mu)\\
    &=(\ovl{0.\kappa(w)}{}^{k_2},\#\{i|w_i\neq0\}).
\end{align*}}
\end{proof}

\subsection{Proof of Theorem \ref{th:fp_gpac_multi}: Multidimensional $\FP$ equivalence}
\label{proof:th:fp_gpac_multi}

\begin{proof}
First note that we can always assume that $m=1$ by applying the result componentwise.
Similarly, we can always assume that $n=2$ by applying the result repeatedly. Since $\FP$
is robust to the exact encoding used for pairs, we choose a particular encoding to prove the result.
Let $\#$ be a fresh symbol not found in $\Gamma$ and define $\Gamma^\#=\Gamma\cup\{\#\}$.
We naturally extend $\gamma$ to $\gamma^\#$ which maps $\Gamma^\#$ to $\N^*$ injectively.
Let $h:{\Gamma^\#}^*\rightarrow\Gamma^*$ and define for any $w,w'\in\Gamma^*$:
\[h^\#(w,w')=h(w\#w').\]
It follows\footnote{This is folklore, but mostly because this particular encoding of pairs is polytime computable.} that
\[f\in\FP\text{ if and only if }\exists h\in\FP\text{ such that }h^\#=f\]

Assume that $f\in\FP$. Then there exists $h\in\FP$ such that $h^\#=f$. Note that $h$ naturally
induces a function (still called) $h:{\Gamma^\#}^*\rightarrow{\Gamma^\#}^*$
so we can apply Theorem~\ref{th:fp_gpac} to get that $h$ is emulable over alphabet $\Gamma^\#$.
Apply Definition~\ref{def:fp_gpac_embed} to get $g\in\mygpc$ and $k\in\N$ that
emulate $h$. In the remaining of the proof, $\psi_k$ denotes encoding of Definition~\ref{def:fp_gpac_embed}
for this particular k, in other words:
\[\psi_k(w)=\left(\sum_{i=1}^{|w|}\gamma^\#(w_i)k^{-i},|w|\right)\]
Define for any $x,x'\in\R$ and $n,n'\in\N$:
\[\varphi(x,n,x',n)=\left(x+\left(\gamma^\#(\#)+x'\right)k^{-n-1},n+m+1\right).\]
We claim that $\varphi\in\mygpc$ and that for any $w,w'\in\Gamma^*$, $\varphi(\psi_k(w),\psi_k(w'))=\psi_k(w\#w')$.
The fact that $\varphi\in\mygpc$ is immediate using Theorem~\ref{th:gpac_comp_arith}
and the fact that $n\mapsto k^{-n-1}$ is analog-polytime-computable\footnote{Note that it works only because $n\geqslant0$.}. The second
fact is follows from a calculation:
\begin{align*}
\varphi(\psi_k(w),\psi_k(w'))&=\varphi\left(\sum_{i=1}^{|w|}\gamma^\#(w_i)k^{-i},|w|,\sum_{i=1}^{|w'|}\gamma^\#(w'_i)k^{-i},|w'|\right)\\
    &=\left(\sum_{i=1}^{|w|}\gamma^\#(w_i)k^{-i}+\left(\gamma^\#(\#)+\sum_{i=1}^{|w'|}\gamma^\#(w'_i)k^{-i}\right)k^{-|w|-1},|w|+|w'|+1\right)\\
    &=\left(\sum_{i=1}^{|w\#w'|}\gamma^\#((w\#w')_i)k^{-i},|w\#w'|\right)\\
    &=\psi_k(w\#w').
\end{align*}
Define $G=g\circ\varphi$. We claim that $G$ emulates $f$ with $k$. First $G\in\mygpc$ thanks to Theorem~\ref{th:gpac_comp_composition}.
Second, for any $w,w'\in\Gamma^*$, we have:
{
\begin{align*}
G(\psi_k(w,w'))&=g(\varphi(\psi_k(w),\psi_k(w')))\tag*{By definition of $G$ and $\psi_k$}\\
    &=g(\psi_k(w\#w'))\tag*{By the above equality}\\
    &=\psi_k(h(w\#w'))\tag*{Because $g$ emulates $h$}\\
    &=\psi_k(h^\#(w,w'))\tag*{By definition of $h^\#$}\\
    &=\psi_k(f(w,w')).\tag*{By the choice of $h$}
\end{align*}
}

Conversely, assume that $f$ is emulable.
Define $F:{\Gamma^\#}^*\rightarrow{\Gamma^\#}^*\times{\Gamma^\#}^*$ as follows for any ${w\in\Gamma^\#}^*$:
\[F(w)=\begin{cases}(w',w'')&\text{if }w=w'\#w''\text{ where }w',w''\in\Gamma^*\\(\emptyword,\emptyword)&\text{otherwise}\end{cases}.\]
Clearly $F_1,F_2\in\FP$ so apply Theorem~\ref{th:fp_gpac} to get that they are emulable.
Thanks to Lemma~\ref{lem:fp_gpac_embed_reenc}, there exists $h,g_1,g_2$ that emulate
$f,F_1,f_2$ respectively with the same $k$. Define:
\[H=h\circ(g_1,g_2).\]
Clearly $H\in\mygpc$ because $g_1,g_2,h\in\mygpc$. Furthermore, $H$ emulates $f\circ F$ because for any $w\in{\Gamma^\#}^*$:
\begin{align*}
H(\psi_k(w))&=h(g_1(\psi_k(w)),g_2(\psi_k(w)))\\
    &=h(\psi_k(g_1(w)),\psi_k(g_2(w)))\tag*{Because $g_i$ emulates $F_i$}\\
    &=h(\psi_k(F(w)))\tag*{By definition of $\psi_k$}\\
    &=\psi_k(f(F(w)))\tag*{Because $h$ emulates $f$}.
\end{align*}
Since $f\circ F:{\Gamma^\#}^*\rightarrow{\Gamma^\#}^*$ is emulable, we can apply Theorem~\ref{th:fp_gpac} to
get that $f\circ F\in\FP$. It is now trivial so see that $f\in\FP$ because for any $w,w'\in\Gamma^*$:
\[f(w,w')=(f\circ F)(w\#w')\]
and $((w,w')\mapsto w\#w')\in\FP$.
\end{proof}


\section{How to only use rational coefficients}
\label{sec:rationalelimination}

This section is devoted to prove that non-rational coefficients can be
eliminated. In other words, we prove that Definitions
\ref{def:discrete_rec_q} and \ref{def:discrete_rec} are defining the
same class, and that Definitions \ref{def:gplc} and \ref{def:glc} are defining the
same class.

Our main Theorems \ref{th:p_gpac_un} and \ref{th:eq_comp_analysis_un}
then clearly
follow.


To do so, we  introduce the following class. We write $\gpc[\Q]$ (resp.
$\gpc[\Rgen]$) for the class of functions $f$ satisfying item (2) of
Proposition \ref{prop:mainequivalence} considering that $\K=\Q$
(resp. $\K=\Rgen$) for some polynomials $\Upsilon$ and $\myOmega$.
Recall that $\Rgen$ denotes the smallest
generable field $\Rgen$ 
lying somewhere between $\Q$ and $\Rpoly$. 
  We
write $\gpwc[\Q]$ (resp.  $\gpwc[\Rgen]$) for the class of functions
$f$ satisfying item (3) of Proposition \ref{prop:mainequivalence}
considering that $\K=\Q$ (resp. $\K=\Rgen$) for some polynomials
$\Upsilon$ and $\myOmega$.

We actually show in this section  that $\gpwc[\Rgen]=\gpc[\Q]$.  As clearly  $\gplc[\K]=\gpc[\K]$ over any
field $\K$ \cite{\JOURNALOFCOMPLEXITY}, it follows that
$\gplc= \gpwc[\Rgen] = \gpc[\Q] =  \gplc[\Q]$
and hence all results
follow.

A particular difficulty in the proof is that none of
the previous theorems applies to $\gpwc[\Q]$ and $\gpc[\Q]$ because $\Q$ is not a
generable field. We thus have to reprove some theorems for the case of rational numbers.
In particular, when using rational numbers, we cannot use, in general, the fact that $y'=g(y)$
rewrites to a PIVP if $g$ is generable, because it may introduce some non-rational
coefficients.

\subsection{Composition in $\gpwc[\Q]$}

The first step is to show that $\gpwc[\Q]$ is stable under composition. This is
not immediate since $\Q$ is not a generable field and we do not have access to any generable function.
The only solution is to manually write a polynomial system with rational coefficients
and show that it works. This fact will be crucial for the remainder of the proof.

In order to compose functions, it will be useful always assume $\myOmega\equiv 1$ when
considering functions in $\gwc[\Q]{\Upsilon}{\myOmega}$.

\begin{lemma}\label{lem:gpwc_constant_omega_rat}
If $f\in\gpwc[\Q]$ then there exists $\Upsilon$ a polynomial such that $f\in\gwc[\Q]{\Upsilon}{\myOmega}$
where $\myOmega(\alpha,\mu)=1$ for all $\alpha$ and $\mu$.
\end{lemma}

\begin{proof}
Let $(f:\subseteq\R^n\to\R^m)\in\gpwc[\Q]$.
By definition, there exists $\myOmega$ and $\Upsilon$ polynomials such
that $f\in\gwc{\Upsilon}{\myOmega}$ with corresponding $d,p,q$.  Without loss of generality, we can assume that $\Upsilon$ and $\myOmega$ are increasing
and have rational coefficients. Let $x\in\dom{f}$ and $\mu\geqslant0$.
Then there exists $y$ such that for all $t\in\Rp$,
\[y(0)=q(x,\mu),\qquad y'(t)=p(y(t)).\]
Consider $(z,\psi)$ the solution to
\[\left\{\begin{array}{@{}l}z(0)=q(x,\mu)\\\psi(0)=\myOmega(1+x_1^2+\cdots+x_n^2,\mu)\end{array}\right.
\qquad
\left\{\begin{array}{@{}l}z'=p(z)\\\psi'=0\end{array}\right.
\]
Note that the system is polynomial with rational coefficients since $\myOmega$ is
a polynomial with rational coefficients. It is easy to see that $z$ and $\psi$ must exist over $\Rp$ and satisfy:
\[\psi(t)=\myOmega(\alpha,\mu),\qquad z(t)=y(\psi(t)t)\]
where $\alpha=1+x_1^2+\cdots+x_n^2$. But then for $t\geqslant1$,
\[\myOmega(\alpha,\mu)t\geqslant\myOmega(\infnorm{x},\mu)\]
since $\myOmega$ is increasing and $\alpha=1+x_1^2+\cdots+x_n^2\geqslant\infnorm{x}$.
It follows by definition that,
\[\infnorm{z_{1..m}(t)-f(x)}=\infnorm{y_{1..m}(\myOmega(\alpha,\mu)t)-f(x)}\leqslant e^{-\mu}\]
for any $t\geqslant 1$, by definition of $y$. Finally, since $\alpha\leqslant\poly(\infnorm{x})$,
\begin{align*}
\infnorm{(z,\psi)(t)}
    &=\max(\infnorm{y(\psi(t)t)},\psi(t))\\
    &\leqslant\max(\Upsilon(\infnorm{x},\mu,\myOmega(\alpha,\mu)t),\myOmega(\alpha,\mu)\\
    &\leqslant\poly(\infnorm{x},\mu,t).
\end{align*}
This proves that $f\in\gwc{\poly}{(\alpha,\mu)\mapsto 1}$ with
rational coefficients only.
\end{proof}

\begin{lemma}\label{lem:poly_weak_computable_compose_rat}
If $(f:\subseteq\R^n\to\R^m)\in\gpwc[\Q]$ and $r\in\Q^\ell[\R^m]$ then $r\circ f\in\gpwc[\Q]$.
\end{lemma}

\begin{proof}
Let $\myOmega,\Upsilon$ be polynomials such that
$f\in\gwc[\Q]{\Upsilon}{\myOmega}$ with corresponding $d,p,q$. 
Using Lemma~\ref{lem:gpwc_constant_omega_rat}, we can assume that $\myOmega\equiv 1$.
Without loss of generality we also assume that $\Upsilon$ has rational coefficients
and is non-decreasing in all variables. Let $x\in\dom{g}$ and $\mu\geqslant0$.
Let $\hat{q}$ be a polynomial with rational coefficients, to be defined later.
Consider the system
\begin{equation}\label{eq:sys_gpwc_poly_comp_sys_1}
y(0)=q(x,\hat{q}(x,\mu)),\qquad y'=p(y).
\end{equation}
Note that by definition $\infnorm{f(x)-y_{1..m}(t)}\leqslant e^{-\hat{q}(x,\mu)}$
for all $t\geqslant 1$.
Using a similar proof to Proposition~\ref{prop:gp_growth}, one can see that for any $t\geqslant 1$.
\begin{equation}\label{eq:sys_gpwc_poly_comp_sys_norm}
\max(\infnorm{f(x)},\infnorm{y_{1..m}(t)})\leqslant2+\Upsilon(\infnorm{x},0,1).
\end{equation}
Let $z(t)=r(y_{1..m(t)}(t))$ and observe that
\begin{equation}\label{eq:sys_gpwc_poly_comp_sys_2}
z(0)=r(q(x,\hat{q}(x,\mu)),\qquad z'(t)=\jacobian{r}(y_{1..m}(t))p_{1..m}(y(t)).
\end{equation}
Note that since $r, p$ and $\hat{q}$ are polynomials with rational coefficients,
the system \eqref{eq:sys_gpwc_poly_comp_sys_1},\eqref{eq:sys_gpwc_poly_comp_sys_2}
is of the form $w(0)=\poly(x,\mu)$, $w'=\poly(w)$ with rational coefficients, where $w=(y,z)$.
Let $k=\degp{r}$, then
\begin{align*}
    \infnorm{r(f(x))-z(t)}&=\infnorm{r(f(x))-r(y_{1..m}(t))}\\
        &\leqslant k\sigmap{r}\max(\infnorm{f(x)},\infnorm{y_{1..m}(t)})^{k-1}\infnorm{f(x)-y_{1..m}(t)}\\
        &\leqslant k\sigmap{r}\left(2+\Upsilon(\infnorm{x},0,1)\right)^{k-1}\infnorm{f(x)-y_{1..m}(t)}&&\text{using \eqref{eq:sys_gpwc_poly_comp_sys_norm}}\\
        &\leqslant k\sigmap{r}\left(2+\Upsilon(\infnorm{x},0,1)\right)^{k-1}e^{-\hat{q}(x,\mu)}&&\text{by definition of $y$}.\\
\end{align*}
We now define $\hat{q}(x,\mu)=\mu+k\sigmap{r}\left(2+\Upsilon(1+x_1^2+\cdots+x_n^2,0,1)\right)^{k-1}$.
Since $\Upsilon$ has rational coefficients, $\hat{q}$ is indeed a polynomial with rational
coefficients. Furthermore, $\infnorm{x}\leqslant1+\inorm{x}{2}^2$ and $\Upsilon$ is non-decreasing,
thus
\[
\hat{q}(x,\mu)=\mu+k\sigmap{r}\geqslant\mu+k\sigmap{r}\left(2+\Upsilon(\infnorm{x},0,1)\right)^{k-1}
\]
and we get that
\[\infnorm{r(f(x))-z(t)}
    \leqslant k\sigmap{r}\left(2+\Upsilon(\infnorm{x},0,1)\right)^{k-1}e^{-\mu+k\sigmap{r}\left(2+\Upsilon(\infnorm{x},0,1)\right)^{k-1}}
    \leqslant e^{-\mu}\]
using that $ue^{-u}\leqslant 1$ for any $u$. Finally, by construction we have
\[\infnorm{y(t)}\leqslant \Upsilon(\infnorm{x},\hat{q}(x,\mu),t)\leqslant\poly(\infnorm{x},\mu,t)\]
and
\[\infnorm{z(t)}=\infnorm{r(y_{1..m}(t))}\leqslant\poly(\infnorm{y_{1..m}(t)})\leqslant\poly(\infnorm{y(t)})\leqslant\poly(\infnorm{x},\mu,t).\]
Thus $r\circ f\in\gpwc[\Q]$.
\end{proof}

We also need a technical lemma to provide us with a simplified version of a
periodic switching function: a function that is periodically very small then very high
(like a clock). Figure~\ref{fig:theta_psi_per} gives the graphical intuition behind these functions.

\begin{lemma}\label{lem:rat_theta_per}
Let $\nu\in C^1(\R,\Rp)$ \textbf{with $\nu(0)=0$} and define for all $t\in\Z$,
\[\theta_{\nu}(t)=\tfrac{1}{2}+\tfrac{1}{2}\tanh(2\nu(t)(\sin(2t)-\tfrac{1}{2})).\]
Then
\[\theta_{\nu}(0)=0,
\qquad
\theta_{\nu}'(t)=p^{\theta}(\theta_\nu(t),\nu(t),\nu'(t),t,\sin(2t),\cos(2t))\]
where $p^{\theta}$ is a polynomial with rational coefficients. Furthermore,
for all $n\in\Z$,
\begin{itemize}
\item if $(n+\tfrac{1}{2})\pi\leqslant t\leqslant (n+1)\pi$ then $|\theta_\nu(t)|\leqslant e^{-\nu(t)}$,
\item if $n\pi+\frac{\pi}{12}\leqslant t\leqslant (n+\tfrac{1}{2})\pi$ then $\theta_\nu(t)\geqslant\tfrac{1}{2}$.
\end{itemize}
\end{lemma}

\begin{proof}
Check that
\[\theta_{\nu}'(t)=\big(\nu'(t)(\sin(2t)-\tfrac{1}{2})+2\nu(t)\cos(2t)\big)(1-(2\theta_\nu(t)-1)^2).\]
Recall that for all $x\in\R$, $|\sgn{x}-\tanh(x)|\leqslant e^{-x}$.
\begin{itemize}
\item If $t\in[(n+\tfrac{1}{2})\pi,(n+1)\pi]$, then $\sin(2t)\leqslant 0$ and since $\tanh$ is increasing,
\[\theta_\nu(t)\leqslant\tfrac{1}{2}+\tfrac{1}{2}\tanh(-\nu(t))\leqslant e^{-\nu(t)}.\]
\item If $t\in[n\pi+\frac{\pi}{12}, (n+\tfrac{1}{2})\pi-\tfrac{\pi}{12}]$ then
$\sin(2t)\geqslant\tfrac{1}{2}$ and $\theta_\nu(t)\geqslant\tfrac{1}{2}$.
\end{itemize}
\end{proof}

\begin{lemma}\label{lem:rat_psi_per}
Let $\nu\in C^1(\R,\Rp)$ \textbf{with $\nu(0)=0$} and define for all $t\in\Z$,
\[\begin{array}{r@{}lp{2cm}r@{}l}
\psi_{0,\nu}(t)&=\theta_\nu(2t)\theta_\nu(t),&&
\psi_{1,\nu}(t)&=\theta_\nu(-2t)\theta_\nu(t),\\\\
\psi_{2,\nu}(t)&=\theta_\nu(2t)\theta_\nu(-t),&&
\psi_{3,\nu}(t)&=\theta_\nu(-2t)\theta_\nu(-t).
\end{array}\]
Then
\[\psi_{i,\nu}(0)=0,
\qquad
\theta_{i,\nu}'(t)=p^{i,\psi}(\theta_\nu(t),\theta_\nu'(t),\theta_\nu(2t),\theta_\nu'(2t))\]
where $p^{i,\psi}$ is a polynomial with rational coefficients. Furthermore,
for all $i\in\{0,1,2,3\}$ and $n\in\Z$,
\begin{itemize}
\item if $(t\mod \pi)\notin[\frac{i\pi}{4},\tfrac{(i+1)\pi}{4}]$ then $|\psi_{i,\nu(t)}|\leqslant e^{-\nu(t)}$,
\item $m_\psi\leqslant\int_{n\pi+\tfrac{i\pi}{4}}^{n\pi+\tfrac{(i+1)\pi}{4}}\psi_{i,\nu(t)}dt\leqslant M_\psi$
for some constants $m_\psi,M_\psi$ that do not depend on $\nu$,
\item for any $\nu,\bar{\nu}$, $i\neq j$ and $(t\mod \pi)\in[\frac{i\pi}{4},\tfrac{(i+1)\pi}{4}]$,
if $\nu(t)\leqslant\bar{\nu}(t)$ then $\psi_{i,\nu}(t)\geqslant\psi_{j,\bar{\nu}}(t)$.
\end{itemize}
\end{lemma}

\begin{proof}
Note that $\psi_{i,\nu}(t)\in[0,1]$ for all $t\in\R$.
The first point is direct consequence of Lemma~\ref{lem:rat_theta_per} and the
fact that $\theta_\nu(-t)=\theta_\nu(t+\tfrac{\pi}{2})$. The second point requires more
work. We only show it for $\psi_{0,\nu}$ since the other cases are similar. Let $n\in\Z$,
if $t\in[n\pi+\tfrac{\pi}{12},n\pi+\tfrac{5\pi}{24}]$ then $t\in[n\pi+\tfrac{\pi}{12},n\pi+\tfrac{5\pi}{12}]$
thus $\gamma_\nu(t)\geqslant\tfrac{1}{2}$, and $2t\in[2n\pi+\tfrac{\pi}{12},2n\pi+\tfrac{5\pi}{12}]$
thus $\gamma_\nu(2t)\geqslant\tfrac{1}{2}$. It follows that $\psi_{0,\nu}(t)\geqslant\tfrac{1}{4}$
and thus
\[\int_{n\pi}^{n\pi+\tfrac{\pi}{4}}\psi_{0,\nu}(t)dt
    \geqslant\int_{n\pi+\tfrac{\pi}{12}}^{n\pi+\tfrac{5\pi}{24}}\tfrac{1}{4}dt
    \geqslant\tfrac{3\pi}{96}.\]
On the other hand,
\[\int_{n\pi}^{n\pi+\tfrac{\pi}{4}}\psi_{0,\nu}(t)dt
    \leqslant\int_{n\pi}^{n\pi+\tfrac{\pi}{4}}1dt\leqslant\frac{\pi}{4}.\]
\end{proof}

\newcommand{\fnthetaper}[2]{(1+tanh(2*(#1)*(sin(2*(#2))-1/2.)))/2.}
\newcommand{\fnpsizeroper}[2]{\fnthetaper{#1}{#2}*\fnthetaper{#1}{2*(#2)}}
\newcommand{\fnpsioneper}[2]{\fnthetaper{#1}{#2}*\fnthetaper{#1}{-2*(#2)}}
\newcommand{\fnpsitwoper}[2]{\fnthetaper{#1}{-#2}*\fnthetaper{#1}{2*(#2)}}
\newcommand{\fnpsithreeper}[2]{\fnthetaper{#1}{-#2}*\fnthetaper{#1}{-2*(#2)}}

{
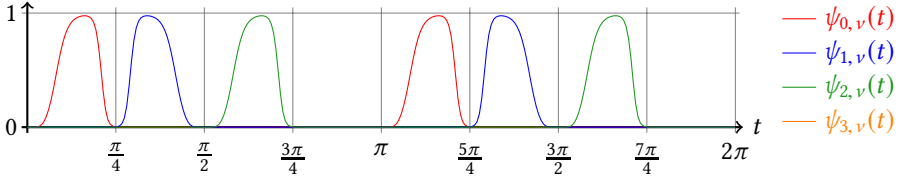
\begin{figure}
\begin{center}
\begin{tikzpicture}[domain=0:6.28,samples=300,scale=1.5]
\draw[very thin,color=gray] (-0.05,-0.05) grid[xstep=pi/4] (6.32,1.05);
\draw[thick,->] (0,-0.05) -- (0,1.1);
\draw[thick,->] (-0.05,0) -- (6.35,0) node[right] {$t$};
\foreach \i in {0,1}
{
    \draw[thick] (0,\i) node[left] {$\i$};
};
\foreach \i/\l in {{pi/4}/\tfrac{\pi}{4},{pi/2}/\tfrac{\pi}{2},{3*pi/4}/\tfrac{3\pi}{4},pi/\pi,
                    {5*pi/4}/\tfrac{5\pi}{4},{3*pi/2}/\tfrac{3\pi}{2},{7*pi/4}/\tfrac{7\pi}{4},2*pi/2\pi}
{
    \draw[thick] (\i,-0.05) node[below] {$\l$};
};
\draw[color=red] plot[id=fn_psi_0_per] function{\fnpsizeroper{3}{x}};
\draw[color=blue] plot[id=fn_psi_1_per] function{\fnpsioneper{3}{x}};
\draw[color=darkgreen] plot[id=fn_psi_2_per] function{\fnpsitwoper{3}{x}};
\draw[color=orange] plot[id=fn_psi_3_per] function{\fnpsithreeper{3}{x}};
\draw[color=red] (6.7,0.95) -- ++(0.3,0) node[right] {$\psi_{0,\nu}(t)$};
\draw[color=blue] (6.7,0.65) -- ++(0.3,0) node[right] {$\psi_{1,\nu}(t)$};
\draw[color=darkgreen] (6.7,0.35) -- ++(0.3,0) node[right] {$\psi_{2,\nu}(t)$};
\draw[color=orange] (6.7,0.05) -- ++(0.3,0) node[right] {$\psi_{3,\nu}(t)$};
\end{tikzpicture}
\end{center}
\caption{Graph of $\psi_{i,\nu}(t)$ for $\nu(t)=3$.}
\label{fig:theta_psi_per}
\end{figure}
}

Thanks to the switching functions defined above, the system will construct will
often be of a special form that we call ``reach''. The properties of this type
of system will be crucial for our proof.

\begin{lemma}\label{lem:bounded_reach_deriv_eq}
Let $d\in\N$, $[a,b]\subset\R$, $z_0\in\R^d$, $y\in C^1([a,b],\R^d)$ and $A,b\in C^0(\R^d\times[a,b],\R^d)$.
Assume that $A_i(x,t)>|b(x,t)|$ for all $t\in[a,b]$ and $x\in\R^d$. Then there exists
a unique $z\in C^1([a,b],\R^d)$ such that
\[z(a)=z_0,\qquad z_i'(t)=A_i(z(t),t)(y_i(t)-z_i(t))+b_i(z(t),t)\]
Furthermore, it satisfies
\[|z_i(t)-y_i(t)|\leqslant \max(1,|z_i(a)-y_i(a)|)+\sup_{s\in[a,t]}|y_i(s)-y_i(a)|,\quad\forall t\in[a,b].\]
\end{lemma}

\begin{proof}By the Cauchy-Lipschitz theorem, there exists a unique $z$ that satisfies
the equation over its maximum interval of life $[a,c)$ with $a<c$.
Let $u(t)=z(t)-y(t)$, then
\begin{align*}
u_i'(t)&=z_i'(t)-y_i'(t)\\
    &=-A_i(z(t),t)u_i(t)+b_i(z(t),t)-y_i'(t)\\
    &=-A_i(u(t)+y(t),t)u_i(t)+b_i(u(t)+y(t),t)-y_i'(t)\\
    &=F_i(u(t),y(t),t)
\end{align*}
where
\[F_i(x,t)=-A_i(y(t)+x,t)x_i+b_i(y(t)+x,t)-y_i'(t).\]
But now observe that for any $t\in[a,c]$, $i\in\{1,\ldots,d\}$ and $x\in\R^d$,
\begin{itemize}
\item if $x_i\geqslant 1$ then $F_i(x,t)<-y_i'(t)$,
\item if $x_i\leqslant -1$ then $F_i(x,t)>-y_i'(t)$.
\end{itemize}
Indeed, if $x_i\geqslant 1$ then
\begin{align*}
    F_i(x,t)
        &=A_i(y(t)+x,t)x_i+b_i(y(t)+x,t)-y_i'(t)\\
        &\geqslant A_i(y(t)+x,t)+b_i(y(t)+x,t)-y_i'(t)&&\text{using }x_i\geqslant 1\\
        &>|b_i(y(t)+x,t)|+b_i(y(t)+x,t)-y_i'(t)&&\text{using }A_i(x,t)>|b_i(x,t)|\\
        &\geqslant -y_i'(t)
\end{align*}
and similarly for $x_i\leqslant|y_i'(t)|$. It follows that for all $t\in[a,c)$,
\begin{equation}\label{eq:bounded_reach_deriv_eq:bound_u}
|u_i(t)|\leqslant \max(1,|u_i(a)|)+\sup_{s\in[a,t]}|y_i(s)-y_i(a)|.
\end{equation}
Indeed let $X_t=\{s\in[a,t]:|u_i(s)|\leqslant 1\}$. If $X_t=\varnothing$ then let $t_0=a$,
otherwise let $t_0=\max X_t$. Then for all $s\in(t_0,t]$,
$|u_i(t)|>1$ thus by continuity of $u$ there are two cases:
\begin{itemize}
\item either $u_i(s)>1$ for all $s\in(t_0,t]$, then $u_i'(s)=F_i(u(s),s)<-y_i'(s))$
thus
\[u_i(t)\leqslant u_i(t_0)-\int_{t_0}^sy_i'(u)du=u_i(t_0)+y_i(t)-y_i(t_0),\]
\item either $u_i(s)<-1$ for all $s\in(t_0,t]$, then $u_i'(s)=F_i(u(s),s)>-y_i'(s))$
thus
\[u_i(t)\geqslant u_i(t_0)-\int_{t_0}^sy_i'(u)du=u_i(t_0)+y_i(t)-y_i(t_0).\]
\end{itemize}
Thus in all cases
\[|u_i(t)|\leqslant |u_i(t_0)|+|y_i(t)-y_i(t_0)|.\]
But now notice that if $X_t=\varnothing$ then $t_0=a$ and $|u_i(t_0)|=|u_i(a)|$.
And otherwise, $t_0=\max X_t$ and $|u_i(t_0)|\leqslant 1$.

But note that the upper bound in \eqref{eq:bounded_reach_deriv_eq:bound_u} has
a finite limit when $t\rightarrow c$ since $y$ is continuous over $[a,b]\supset[a,c)$.
This implies that $u(c)$ exists and thus that $c=b$
because if it was not the case, by Cauchy-Lipschitz, we could extend the solution
to the right of $c$ and contradict the maximality of $[a,c)$.
\end{proof}

\begin{lemma}\label{lem:conv_reach_deriv}
Let $d\in\N$, $z_0,\varepsilon\in\R^d$, $[a,b]\subset\R$, $y\in C^1([a,b],\R^d)$ and $A,b\in C^0(\R^d\times[a,b],\R^d)$.
Assume that $A_i(x,t)\geqslant 0$ and $|b_i(x,t)|\leqslant\varepsilon_i$ for all $t\in[a,b]$, $x\in\R^d$
and $i\in\{1,\ldots,d\}$. Then there exists a unique $z\in C^1([a,b],\R^d)$ such that
\[z(a)=z_0,\qquad z_i'(t)=A_i(z(t),t)(y_i(t)-z_i(t))+b_i(z(t),t)\]
Furthermore, it satisfies
\[|z_i(t)-y_i(t)|\leqslant|z_i(a)-y_i(a)|\exp\left(-\int_a^tA_i(z(s),s)ds\right)
            +|y_i(t)-y_i(a)|+(t-a)\varepsilon_i.\]
\end{lemma}

\begin{proof}
The existence of a solution over $[a,b]$ is almost immediate since $b$ is bounded.
Let $u(t)=z(t)-y(t)$, then
\[u_i'(t)=z_i'(t)-y_i'(t)=-A_i(z(t),t)u_i(t)-y_i'(t)+b_i(z(t),t)\]
and thus we have a the following closed-form expression for $u_i$:
\[u_i(t)=e^{-\phi(t)}\left(\int_a^t e^{\phi(u)}(b_i(z(u),u)-y_i'(u))du+u_i(0)\right)\]
where
\[\phi(t)=\int_a^tA_i(z(s),s)ds\]
Thus
\[
|u_i(t)|
    \leqslant e^{-\phi(t)}|u_i(0)|+\int_a^te^{\phi(u)-\phi(t)}|b_i(z(u),u)|du
    +\left|\int_a^te^{\phi(u)-\phi(t)}y_i'(u)du\right|.\]
But by the Mean Value Theorem, there exists $c_t\in[a,t]$ such that
\[\int_a^te^{\phi(u)-\phi(t)}y_i'(u)du=e^{\phi(c_t)-\phi(t)}\int_a^ty_i'(u)du=e^{\phi(c_t)-\phi(t)}(y_i(t)-y_i(a)).\]
Thus by using that $\phi$ is increasing,
\begin{align*}
|u_i(t)|
    &\leqslant e^{-\phi(t)}|u_i(0)|+\int_a^t|b_i(z(u),u)|du
    +e^{\phi(c_t)-\phi(t)}|y_i(t)-y_i(a)|\\
    &\leqslant e^{-\phi(t)}|u_i(0)|+\int_a^t\varepsilon_idu
    +|y_i(t)-y_i(a)|\\
    &\leqslant e^{-\phi(t)}|u_i(0)|+(t-a)\varepsilon_i
    +|y_i(t)-y_i(a)|.
\end{align*}
\end{proof}

We can now show the major result of this subsection: the composition
of two functions of $\gpwc[\Q]$ is in $\gpc[\Q]$, that is computable using only rational coefficients.  Note that we are intuitively doing two things at once: showing that
the composition is computable, and that weak-computability implies computability;
none of which are obvious in the case of rational coefficients.

\begin{theorem}\label{th:weak_computable_compose_rat_comp}
If $f,g\in\gpwc[\Q]$ then $f\circ g\in\gpc[\Q]$.
\end{theorem}

\begin{proof}
Let $(f:\subseteq\R^m\to\R^\ell) \gpwc[\Q]$ with corresponding
$d,p^f,q^f$. Let $(g:\subseteq\R^n\to\R^m)\in\gpwc[\Q]$. Since $g\in\gpwc[\Q]$,
the function $(x,\mu)\mapsto(g(x),\mu)$ trivially belongs to $\gpwc[\Q]$.
Let $h(x,\mu)=q^f(g(x),\mu)$, then $h\in\gpwc[\Q]$ by Lemma~\ref{lem:poly_weak_computable_compose_rat}
since $q^f$ has rational coefficients. Using Lemma~\ref{lem:gpwc_constant_omega_rat},
we can assume that $f,h\in\gwc[\Q]{\Upsilon}{\myOmega}$ with $\myOmega\equiv 1$.
Note that we can always make the assumption that $\Upsilon$ is the same
for both $f$ and $h$ by taking the maximum. We have $h\in\gpwc[\Q]$ with
corresponding $d',p^h,q^h$.

To avoid any confusion, note that $q^h$ takes two ``$\mu$'' as input: the input of
$h$ is $(x,\mu)$ but $q^h$ adds a $\nu$ for the precision: $q^h((x,\mu),\nu)$.

To simplify notations, we will assume that $d=d'$, that is both
systems have the same number of variables, by adding useless variables to either system.

Let $x\in\dom{g}=\dom{h}$ and $\mu\geqslant0$.
Let $R, S$ and $Q$ be polynomials with rational coefficients, to be defined later,
but increasing in all variables.
Let $m_\psi,M_\psi$ be the constants from Lemma~\ref{lem:rat_psi_per}. Without loss
of generality, we can assume that the are rational numbers.
Consider the following system:
\[
\begin{array}{r@{}lp{1cm}r@{}l}
\mu(0)&=1,&&
\mu'(t)&=\psi_{3,\nu_\mu}(t)\alpha,\\\\
y(0)&=0,&&
y_i'(t)&=\psi_{0,\nu_0^i}(t)g_{0,i}(t)+\psi_{1,\nu_{1}^i}(t)g_{1,i}(t)+\psi_{2,\nu_{2}^i}(t)g_{2,i}(t)
\end{array}\]
where
\begin{align*}
    g_{0,i}(t)&=A_{0,i}(t)(r_i(t)-y_i(t)),\\
    g_{1,i}(t)&=\alpha p_i^h(y),\\
    g_{2,i}(t)&=\alpha p_i^f(y),\\
    A_{0,i}(t)&=\alpha Q(x,\mu(t))+2+g_{1,i}(t)^2+g_{1,2}(t)^2\\
    \alpha&=\max(1,\tfrac{1}{m_\psi})\\
    r_i(t)&=q_i^h(x,R(x,\mu(t)),S(x,\mu(t))),\\
    \nu_0^i(t)&=\big(1+g_{0,i}(t)^2+Q(x,\mu(t))\big)\beta t,\\
    \nu_1^i(t)&=\nu_0^i(t)+\big(1+g_{1,i}(t)^2+Q(x,\mu(t))\big)\beta t,\\
    \nu_2^i(t)&=\nu_0^i(t)+\big(1+g_{2,i}(t)^2+Q(x,\mu(t))\big)\beta t,\\
    \nu_\mu(t)&=\big(\alpha+\pi+Q(x,\mu(t))\big)\beta t,\\
    \beta&=4.
\end{align*}
Notice that we took the $\nu_{\ldots}(t)$ such that $\nu_{\ldots}(0)=0$ since it will
be necessary for the $\psi_{j,\nu}$. This explains the unexpected product
by $t$.

We start with the analysis of $\mu$, which is simplest. First note that $\mu'(t)\geqslant 0$
thus $\mu$ is increasing. And since $\mu'(t)\leqslant\alpha$ is bounded, it is clear
that $\mu$ must exist over $\R$. As a result, since $Q$ is increasing in $\mu$,
$\nu_\mu$ is also an increasing function.

Let $n\in\N$, then
\begin{align*}
\mu((n+1)\pi)
    &=\mu((n+\tfrac{3}{4})\pi)+\int_{(n+\tfrac{3}{4})\pi}^{(n+1)\pi}\mu'(t)dt\\
    &\geqslant \mu(n\pi)+\alpha\int_{(n+\tfrac{3}{4})\pi}^{(n+1)\pi}\psi_{3,\nu_\mu}(t)dt&&\text{since $\mu$ increasing}\\
    &\geqslant \mu(n\pi)+\alpha m_\psi&&\text{by Lemma~\ref{lem:rat_psi_per}}\\
    &\geqslant \mu(n\pi)+1&&\text{since }\alpha\geqslant m_\psi.
\end{align*}
It follows that for all $n\in\N$,
\begin{equation}\label{eq:wccrc:lower_bound_mu}
\mu(n\pi)\geqslant n+\mu(0)\geqslant n+1.
\end{equation}
But on the other hand,
\begin{align*}
\mu(t)
    &=\mu(0)+\int_{0}^{t}\mu'(u)du\\
    &=1+\alpha\int_{0}^{n\pi}\psi_{3,\nu_\mu}(u)du\\
    &\leqslant\alpha\int_{0}^{n\pi}1du\\
    &\leqslant 1+\alpha\pi t.\label{eq:wccrc:upper_bound_mu}\numberthis
\end{align*}
Let $n\in\N$, then by Lemma~\ref{lem:rat_psi_per},
for all $t\in[n\pi,(n+\tfrac{3}{4})\pi]$, $|\mu'(t)|\leqslant\alpha e^{-\nu_\mu(t)}$.
So in particular, if $t\geqslant\tfrac{1}{\beta}$ then $\mu_\nu(t)\geqslant\pi+\alpha+Q(x,\mu(t))
\geqslant\pi+\alpha+Q(x,\mu(n\pi))$. It follows that
\begin{equation}\label{eq:wccrc:stable_mu}
|\mu(t)-\mu(t')|
    \leqslant\tfrac{3}{4}\pi\alpha e^{-\pi-\alpha-Q(x,\mu(n\pi))}
    \leqslant e^{-Q(x,\mu(n\pi))},\qquad\forall t,t'\in[n\pi+\tfrac{1}{\beta},(n+\tfrac{3}{4})\pi].
\end{equation}

We can now start to analyze $y$. Let $n\in\N$, we will split the analysis in
several time intervals that correspond to different behaviors. Note that
we chose $\beta$ such that $\tfrac{1}{\beta}\leqslant\tfrac{\pi}{4}$. We use
the following fact many times during the proof: $|u|\leqslant 1+u^2$ for all $u\in\R$.

We will prove the following invariant by induction over $n\in\N$: there exists
a polynomial $M$ such that
\begin{equation}\label{eq:wccrc:prop_M}
\infnorm{y(n\pi)}\leqslant M(x,\mu(n\pi)).
\end{equation}
At this stage $M$ is still unspecified, but it is a very important requirement that
$M$ is \textbf{not allowed} to depend $Q$.
Note that \eqref{eq:wccrc:prop_M} is trivially satisfiable for $n=0$.

\textbf{Over $\mathbf{[n\pi,n\pi+\tfrac{1}{\beta}]}$:} this part is special
for $n=0$, the various $\nu_{\ldots}$ are still ``bootstrapping'' because
of the product by $t$ that we added to make $\mu_{\ldots}(0)=0$. The only
thing we show is that the solution exists, a non-trivial fact at this stage. First note
that by constrution, $\nu_1^i(t)\geqslant\nu_0^i(t)$ and $\nu_2^i(t)\geqslant\nu_0^i(t)$.
It follows for any $t\in[n\pi,n\pi+\tfrac{1}{\beta}]$, using Lemma~\ref{lem:rat_psi_per} that
\begin{equation}\label{eq:wccrc:rel_psi_boot}
\psi_{0,\nu_0^i}(t)\geqslant\psi_{1,\nu_1^i}(t)
\quad\text{and}\quad
\psi_{0,\nu_0^i}(t)\geqslant\psi_{2,\nu_1^i}(t).
\end{equation}
Furthermore, also by construction,
\begin{equation}\label{eq:wccrc:rel_A_g_boot}
A_{0,i}(t)\geqslant |g_{1,i}(t)|+|g_{2,i}(t)|.
\end{equation}
Putting \eqref{eq:wccrc:rel_psi_boot} and \eqref{eq:wccrc:rel_A_g_boot} we get that
\begin{equation}\label{eq:wccrc:bound_ode_boot}
A_{0,i}(t)\psi_{0,\nu_0^i}(t)\geqslant|\psi_{1,\nu_1^i}(t)g_{1,i}(t)|+|\psi_{2,\nu_2^i}(t)g_{2,i}(t)|.
\end{equation}
Since the system is of the form
\[y_i'(t)=\psi_{0,\nu_0^i}(t)A_{0,i}(t)(r(t)-y_i(t))+\psi_{1,\nu_{1}^i}(t)g_{1,i}(t)+\psi_{2,\nu_{2}^i}(t)g_{2,i}(t),\]
we can use \eqref{eq:wccrc:bound_ode_boot} to apply Lemma~\ref{lem:bounded_reach_deriv_eq}
to conclude that $y$ exists over $[n\pi,n\pi+\tfrac{1}{\beta}]$ and that
\begin{equation}\label{eq:wccrc:bound_y_boot_1}
|y_i(t)-r_i(t)|\leqslant\max(1,|y_i(n\pi)-r_i(n\pi)|)+\sup_{s\in[n\pi,t]}|r_i(s)-r_i(n\pi)|.
\end{equation}
Recall that $r_i(t)=q_i^h(x,R(x,\mu(t)),S(x,\mu(t)))$. So in particular,
using \eqref{eq:wccrc:upper_bound_mu},
\begin{equation}\label{eq:wccrc:bound_r_boot}
|r_i(t)|\leqslant q_i^h(x,R(x,1+\alpha\pi t),S(x,1+\alpha\pi t)).
\end{equation}
It follows that forall $t\in[n\pi,n\pi+\tfrac{1}{\beta}]$,
\begin{align*}
|y_i(t)-r_i(t)|
    &\leqslant\max(1,|y_i(n\pi)-r_i(n\pi)|)+\sup_{s\in[n\pi,t]}|r_i(s)-r_i(n\pi)|
        &&\text{using \eqref{eq:wccrc:bound_y_boot_1}}\\
    &\leqslant 1+|y_i(n\pi)|+|r_i(n\pi)|+2\sup_{s\in[n\pi,t]}|r_i(s)|\\
    &\leqslant 1+|y_i(n\pi)|+3\sup_{s\in[n\pi,t]}q_i^h(x,R(x,1+\alpha\pi s),S(x,1+\alpha\pi s))
        &&\text{using \eqref{eq:wccrc:bound_r_boot}}\\
    &\leqslant 1+M(x,\mu(n\pi))+3\sup_{s\in[n\pi,t]}q_i^h(x,R(x,1+\alpha\pi s),S(x,1+\alpha\pi s))
        &&\text{using \eqref{eq:wccrc:prop_M}}\\
    &\leqslant P_1(x,\mu(n\pi))\numberthis\label{eq:wccrc:bound_y_boot_2}
\end{align*}
for some polynomial\footnote{Note for later that $P_1$ depends on $q^h,M,R$ and $S$.} $P_1$.

\textbf{Over $\mathbf{[n\pi+\tfrac{1}{\beta},(n+\tfrac{1}{4})\pi]}$:} it is important
to note that in this case, and all remaining cases, $\beta t\geqslant 1$. Indeed
by construction we get for all $t\in[n\pi+\tfrac{1}{\beta},(n+\tfrac{1}{4})\pi]$ that
\[\nu_1^i(t)\geqslant|g_{1,i}(t)|+Q(x,\mu(t))
\quad\text{and}\quad
\nu_2^i(t)\geqslant|g_{2,i}(t)|+Q(x,\mu(t)).\]
It follows from Lemma~\ref{lem:rat_psi_per} and the fact that $\mu$ is increasing that
\begin{equation}\label{eq:wccrc:phase_0:bound_psi_g_1}
|\psi_{1,\nu_1^i}(t)g_{1,i}(t)|\leqslant e^{-\nu_1^i(t)}|g_{1,i}(t)|\leqslant e^{-Q(x,\mu(t))}\leqslant e^{-Q(x,\mu(n\pi))}
\end{equation}
and
\begin{equation}\label{eq:wccrc:phase_0:bound_psi_g_2}
|\psi_{2,\nu_1^i}(t)g_{2,i}(t)|\leqslant e^{-\nu_2^i(t)}|g_{2,i}(t)|\leqslant e^{-Q(x,\mu(t))}\leqslant e^{-Q(x,\mu(n\pi))}.
\end{equation}
Thus we can apply Lemma~\ref{lem:conv_reach_deriv} and get that
\begin{equation}\label{eq:wccrc:bound_y_phase_0_1}
|y_i(t)-r_i(t)|\leqslant |y_i(n\pi+\tfrac{1}{\beta})-r_i(n\pi+\tfrac{1}{\beta})|e^{-B(t)}
    +2e^{-Q(x,\mu(n\pi))}+|r_i(t)-r_i(n\pi+\tfrac{1}{\beta})|
\end{equation}
where
\begin{align*}
B(t)
    &=\int_{n\pi+\tfrac{1}{\beta}}^{(n+\tfrac{1}{4})\pi}\psi_{0,\nu_0^i}(u)A_{0,i}(u)du\\
    &\geqslant \int_{n\pi+\tfrac{1}{\beta}}^{(n+\tfrac{1}{4})\pi}\psi_{0,\nu_0^i}(u)\alpha Q(x,\mu(u))du\\
    &\geqslant \alpha Q(x,\mu(n\pi))\int_{n\pi+\tfrac{1}{\beta}}^{(n+\tfrac{1}{4})\pi}\psi_{0,\nu_0^i}(u)du&&\text{since $Q$ and $\mu$ increasing}\\
    &\geqslant \alpha Q(x,\mu(n\pi))m_\psi&&\text{using Lemma~\ref{lem:rat_psi_per}}\\
    &\geqslant Q(x,\mu(n\pi))&&\text{since }\alpha m_\psi\geqslant 1.\numberthis\label{eq:wccrc:phase_0_lower_bound_B}
\end{align*}
Recall that $r_i(t)=q_i^h(x,R(x,\mu(t)),S(x,\mu(t))$ where $q_i^h$ and $R$ are polynomials.
It follows that there exists a polynomial\footnote{Note for later that $\Delta_r$ depends on $q^h,R$ and $S$.}
$\Delta_r$ such that for all $t,t'\geqslant 0$,
\[|r_i(t)-r_i(t')|\leqslant\Delta_r(x,\max(|\mu(t)|,|\mu(t')|))|\mu(t)-\mu(t')|.\]
And using \eqref{eq:wccrc:upper_bound_mu}, and \eqref{eq:wccrc:stable_mu} we get that
\begin{equation}\label{eq:wccrc:phase_0_bound_r}
|r_i(t)-r_i(t')|\leqslant\Delta_r(x,1+\alpha\pi t)e^{-Q(x,\mu(n\pi))}.
\end{equation}
It follows that
Putting , \eqref{eq:wccrc:phase_0_lower_bound_B}
and \eqref{eq:wccrc:phase_0_bound_r} we get that
\begin{align*}
|y_i(t)-r_i(t)|
    &\leqslant |y_i(n\pi+\tfrac{1}{\beta})-r_i(n\pi+\tfrac{1}{\beta})|e^{-B(t)}
        &&\text{using \eqref{eq:wccrc:bound_y_phase_0_1}}\\
    &\hspace{1em}+2e^{-Q(x,\mu(n\pi))}+|r_i(t)-r_i(n\pi+\tfrac{1}{\beta})|\\
    &\leqslant P_1(x,\mu(n\pi))e^{-B(t)}+2e^{-Q(x,\mu(n\pi))}+|r_i(t)-r_i(n\pi+\tfrac{1}{\beta})|
        &&\text{using \eqref{eq:wccrc:bound_y_boot_2}}\\
    &\leqslant P_1(x,\mu(n\pi))e^{-Q(x,\mu(n\pi))}+2e^{-Q(x,\mu(n\pi))}+|r_i(t)-r_i(n\pi+\tfrac{1}{\beta})|
        &&\text{using \eqref{eq:wccrc:phase_0_lower_bound_B}}\\
    &\leqslant P_1(x,\mu(n\pi))e^{-Q(x,\mu(n\pi))}+2e^{-Q(x,\mu(n\pi))}+\Delta_r(x,1+\alpha\pi t)e^{-Q(x,\mu(n\pi))}
        &&\text{using \eqref{eq:wccrc:phase_0_bound_r}}\\
    &\leqslant P_2(x,\mu(n\pi))e^{-Q(x,\mu(n\pi))}\numberthis\label{eq:wccrc:phase_0_prec_y}
\end{align*}
for some polynomial\footnote{Note that $P_2$ depends on $P_1$ and $\Delta_r$.
In particular it does not depend, even indirectly, on $Q$.} $P_2$.

\textbf{Over $\mathbf{[(n+\tfrac{1}{4})\pi,(n+\tfrac{1}{2})\pi]}$:} for all
$t$ in this interval,
\[\nu_0^i(t)\geqslant|g_{0,i}(t)|+Q(x,\mu(t))
\quad\text{and}\quad
\nu_2^i(t)\geqslant|g_{2,i}(t)|+Q(x,\mu(t)).\]
It follows from Lemma~\ref{lem:rat_psi_per} and the fact that $\mu$ is increasing that
\begin{equation}\label{eq:wccrc:phase_1:bound_psi_g_1}
|\psi_{0,\nu_1^i}(t)g_{0,i}(t)|\leqslant e^{-\nu_0^i(t)}|g_{0,i}(t)|\leqslant e^{-Q(x,\mu(t))}\leqslant e^{-Q(x,\mu(n\pi))}
\end{equation}
and
\begin{equation}\label{eq:wccrc:phase_1:bound_psi_g_2}
|\psi_{2,\nu_1^i}(t)g_{2,i}(t)|\leqslant e^{-\nu_2^i(t)}|g_{2,i}(t)|\leqslant e^{-Q(x,\mu(t))}\leqslant e^{-Q(x,\mu(n\pi))}.
\end{equation}
Consequently, the system is of the form
\begin{equation}\label{eq:wccrc:phase_1_sys}
y_i'(t)=\alpha \psi_{1,\nu^i_1}(t)p^h_i(y(t))+\varepsilon_i(t)
\quad\text{where}\quad
|\varepsilon_i(t)|\leqslant 2e^{-Q(x,\mu(n\pi))}.
\end{equation}
For any $t\in[(n+\tfrac{1}{4})\pi,(n+\tfrac{1}{2})\pi]$, let
\[\xi(t)=(n+\tfrac{1}{4})\pi+\int_{(n+\tfrac{1}{4})\pi}^t\alpha\psi_{1,\nu^i_1}(u)du.\]
Since $\psi_{1,\nu^i_1}>0$, $\xi$ is increasing and invertible. Now consider the following
system:
\begin{equation}\label{eq:wccrc:phase_1_rescaled_sys}
z_i((n+\tfrac{1}{4})\pi)=y_i((n+\tfrac{1}{4})\pi),
\qquad z_i'(u)=p^h_i(z(u))+\varepsilon(\xi^{-1}(u)).
\end{equation}
It follows that, on the interval of life,
\begin{equation}\label{eq:wccrc:phase_1_rel_y_z}
y_i(t)=z_i(\xi(t)).
\end{equation}
Note using Lemma~\ref{lem:rat_psi_per} that
\begin{equation}\label{eq:wccrc:phase_1_bounded_xi}
1\leqslant\alpha m_\psi\leqslant \xi((n+\tfrac{1}{2})\pi)-\xi((n+\tfrac{1}{4})\pi)\leqslant \alpha M_\psi.
\end{equation}
Now consider the following system:
\begin{equation}\label{eq:wccrc:phase_1_perfect_sys}
w_i((n+\tfrac{1}{4})\pi)=q_i^h(x,R(x,\mu(n\pi)),S(x,\mu(n\pi))),
\qquad w_i'(u)=p^h_i(z(u)).
\end{equation}
By definition of $q^h$ and $p^h$, the solution $w$ exists over $\R$ and satisfies
that
\begin{equation}\label{eq:wccrc:prec_w}
|w_i(u)-h_i(x,R(x,\mu(n\pi)))|\leqslant e^{-S(x,\mu(n\pi))}\quad\text{for all }u-(n+\tfrac{1}{4})\pi\geqslant1
\end{equation}
since $h\in\gwc[\Q]{\Upsilon}{\myOmega}$ with $\myOmega\equiv 1$, and
\begin{align*}
|w_i(u)|
    &\leqslant\Upsilon(\infnorm{(x,R(x,\mu(n\pi)))},S(x,\mu(n\pi)),u-(n+\tfrac{1}{4})\pi)\\
    &\leqslant P_3(x,\mu(n\pi),u-(n+\tfrac{1}{3})\pi)&&\text{for all }u\in\R\numberthis\label{eq:wccrc:bound_w}
\end{align*}
for some polynomial\footnote{Note that $P_3$ depends on $\Upsilon,R$ and $S$.} $P_3$.
Following Theorem~16 of \cite{BournezGP16b}, let $\eta>0$ and $a=(n+\tfrac{1}{4})\pi$
and let
\begin{equation}\label{eq:wccrc:def_delta}
\delta_\eta(u)=\left(\infnorm{z(a)-w(a)}+\int_a^u\infnorm{\varepsilon(\xi^{-1}(s))}ds\right)
    \exp\left(k\sigmap{p^h}\int_a^u(\infnorm{w(s)}+\eta)^{k-1}ds\right)
\end{equation}
where $k=\degp{p^h}$. Let $u\in[a,b]$ where $b=\xi((n+\tfrac{1}{2})\pi)$,
then
\begin{align*}
\int_a^b\infnorm{\varepsilon(\xi^{-1}(s))}ds
    &\leqslant 2(b-a)e^{-Q(x,\mu(n\pi))}&&\text{using \eqref{eq:wccrc:phase_1_sys}},\\
\infnorm{z(a)-w(a)}
    &=\infnorm{q^h(x,R(x,\mu(n\pi),S(x,\mu(n\pi))))-y(a)}\\
    &=\infnorm{r(n\pi)-y(a)}\\
    &\leqslant P_2(x,\mu(n\pi))e^{-Q(x,\mu(n\pi))}&&\text{using \eqref{eq:wccrc:phase_0_prec_y}},\\
k\sigmap{p^h}\int_a^b(\infnorm{w(s)}+\eta)^{k-1}ds
    &\leqslant k\sigmap{p^h}(b-a)\big(\eta+P_3(x,\mu(n\pi),b)\big)^{k-1}
        &&\text{using \eqref{eq:wccrc:bound_w}},\\
b&\leqslant a+\alpha M_\psi&&\text{using \eqref{eq:wccrc:phase_1_bounded_xi}}.
\end{align*}
Plugging everything into \eqref{eq:wccrc:def_delta} we get that for all $u\in[a,b]$,
\begin{equation}\label{eq:wccrc:bound_delta}
\delta_1(u)\leqslant P_4(x,\mu(n\pi))e^{-Q(x,\mu(n\pi))}e^{P_5(x,\mu(n\pi))}
\end{equation}
for some polynomials\footnote{Note that $P_4$ depends on $P_2$ and $P_5$ on $\Upsilon$ and $p^h$.}
$P_4$ and $P_5$. Since we have no chosen $Q$ yet, we now let
\begin{equation}\label{eq:wccrc:def_Q}
Q(x,\nu)=P_5(x,\nu)+P_4(x,\nu)+Q^*(x,\nu)
\end{equation}
where $Q^*$ is some unspecified polynomial to be fixed later. Note that this definition
makes sense because $P_4$ and $P_5$ do not (even indirectly) depend on $Q$.
It then follows from \eqref{eq:wccrc:bound_delta} that
\[\delta_1(u)\leqslant e^{-Q^*(x,\mu(n\pi))}\leqslant 1\]
and thus we can apply Theorem~16 of \cite{BournezGP16b} to get that
\begin{equation}\label{eq:wccrc:rel_w_z}
|z_i(u)-w_i(u)|\leqslant\delta_1(u)\leqslant e^{-Q^*(x,\mu(n\pi))}\qquad\text{for all }u\in[a,b].
\end{equation}
But in particular, \eqref{eq:wccrc:phase_1_bounded_xi} implies that $b-a\geqslant 1$
so by \eqref{eq:wccrc:prec_w}
\begin{equation}\label{eq:wccrc:phase_1_prec_z}
|z_i(b)-h_i(x,R(x,\mu(n\pi)))|\leqslant e^{-Q^*(x,\mu(n\pi))}+e^{-S(x,\mu(n\pi))}.
\end{equation}
And finally, using \eqref{eq:wccrc:phase_1_rel_y_z} we get that
\begin{equation}\label{eq:wccrc:phase_1_prec_y_pre}
|y_i((n+\tfrac{1}{2})\pi)-h_i(x,R(x,\mu(n\pi)))|\leqslant e^{-Q^*(x,\mu(n\pi))}+e^{-S(x,\mu(n\pi))}.
\end{equation}
At this stage, we let
\begin{equation}\label{eq:wccrc:def_Qs}
Q^*(x,\nu)=S(x,\nu)+R(x,\nu)
\end{equation}
so that
\begin{equation}\label{eq:wccrc:phase_1_prec_y}
|y_i((n+\tfrac{1}{2})\pi)-h_i(x,R(x,\mu(n\pi)))|\leqslant 2e^{-S(x,\mu(n\pi))}.
\end{equation}

\textbf{Over $\mathbf{[(n+\tfrac{1}{2})\pi,(n+\tfrac{3}{4})\pi]}$:} the situation
is very similar to the previous case so we omit some proof steps.
The system is of the form
\begin{equation}\label{eq:wccrc:phase_2_sys}
y_i'(t)=\alpha \psi_{2,\nu^i_1}(t)p^f_i(y(t))+\varepsilon_i(t)
\quad\text{where}\quad
|\varepsilon_i(t)|\leqslant 2e^{-Q(x,\mu(n\pi))}.
\end{equation}
We let
\[\xi(t)=(n+\tfrac{1}{2})\pi+\int_{(n+\tfrac{1}{2})\pi}^t\alpha\psi_{1,\nu^i_1}(u)du\]
and consider the following system:
\begin{equation}\label{eq:wccrc:phase_2_rescaled_sys}
z_i((n+\tfrac{1}{2})\pi)=y_i((n+\tfrac{1}{2})\pi),
\qquad z_i'(u)=p^f_i(z(u))+\varepsilon(\xi^{-1}(u)).
\end{equation}
It follows that, on the interval of life,
\begin{equation}\label{eq:wccrc:phase_2_rel_y_z}
y_i(t)=z_i(\xi(t)).
\end{equation}
It is again the case that
\begin{equation}\label{eq:wccrc:phase_2_bounded_xi}
1\leqslant\alpha m_\psi\leqslant \xi((n+\tfrac{3}{4})\pi)-\xi((n+\tfrac{1}{2})\pi)\leqslant \alpha M_\psi.
\end{equation}
We introduce the following system:
\begin{equation}\label{eq:wccrc:phase_2_perfect_sys}
w_i((n+\tfrac{1}{2})\pi)=q_i^f(g(x),R(x,\mu(n\pi))),
\qquad w_i'(u)=p^f_i(z(u)).
\end{equation}
By definition of $q^f$ and $p^f$, the solution $w$ exists over $\R$ and satisfies
that
\begin{equation}\label{eq:wccrc:phase_2_prec_w}
|w_i(u)-f_i(g(x))|\leqslant e^{-R(x,\mu(n\pi))}\quad\text{for all }u-(n+\tfrac{1}{2})\pi\geqslant 1
\end{equation}
since $f\in\gwc[\Q]{\Upsilon}{\myOmega}$ with $\myOmega\equiv 1$, and
\begin{align*}
|w_i(u)|
    &\leqslant\Upsilon(\infnorm{g(x)},R(x,\mu(n\pi)),u-(n+\tfrac{1}{2})\pi)\\
    &\leqslant P_6(x,\mu(n\pi),u-(n+\tfrac{1}{2})\pi)&&\text{for all }u\in\R\numberthis\label{eq:wccrc:phase_2_bound_w}
\end{align*}
for some polynomial\footnote{Note that $P_6$ depends on $\Upsilon$ and $R$.} $P_6$
since $\infnorm{g(x)}\leqslant 1+\Upsilon(\infnorm{x},0,1)$.
Following Theorem~16 of \cite{BournezGP16b}, let $\eta>0$ and $a=(n+\tfrac{1}{2})\pi$
and let
\begin{equation}\label{eq:wccrc:phase_2_def_delta}
\delta_\eta(u)=\left(\infnorm{z(a)-w(a)}+\int_a^u\infnorm{\varepsilon(\xi^{-1}(s))}ds\right)
    \exp\left(k\sigmap{p^f}\int_a^u(\infnorm{w(s)}+\eta)^{k-1}ds\right)
\end{equation}
where $k=\degp{p^f}$. Let $u\in[a,b]$ where $b=\xi((n+\tfrac{1}{2})\pi)$,
then
\begin{align*}
\int_a^b\infnorm{\varepsilon(\xi^{-1}(s))}ds
    &\leqslant 2(b-a)e^{-Q(x,\mu(n\pi))}&&\text{using \eqref{eq:wccrc:phase_2_sys}},\\
    &\leqslant 2(b-a)e^{-S(x,\mu(n\pi))}&&\text{using \eqref{eq:wccrc:def_Q} and \eqref{eq:wccrc:def_Qs}},\\
\infnorm{z(a)-w(a)}
    &=\infnorm{q^f(g(x),R(x,\mu(n\pi)))-y(a)}\\
    &=\infnorm{h(x,R(x,\mu(n\pi)))-y(a)}\\
    &\leqslant 2e^{-S(x,\mu(n\pi))}&&\text{using \eqref{eq:wccrc:phase_1_prec_y}},\\
k\sigmap{p^f}\int_a^b(\infnorm{w(s)}+\eta)^{k-1}ds
    &\leqslant k\sigmap{p^h}(b-a)\big(\eta+P_6(x,\mu(n\pi),b)\big)^{k-1}
        &&\text{using \eqref{eq:wccrc:phase_2_bound_w}},\\
b&\leqslant a+\alpha M_\psi&&\text{using \eqref{eq:wccrc:phase_2_bounded_xi}}.
\end{align*}
Plugging everything into \eqref{eq:wccrc:phase_2_def_delta} we get that for all $u\in[a,b]$,
\begin{equation}\label{eq:wccrc:phase_2_bound_delta}
\delta_1(u)\leqslant P_7(x,\mu(n\pi))e^{-S(x,\mu(n\pi))}e^{P_8(x,\mu(n\pi))}
\end{equation}
for some polynomials\footnote{Note that $P_7$ depends on $P_6$ and $P_8$ on $\Upsilon$ and $p^f$.}
$P_7$ and $P_8$. Since we have no chosen $S$ yet, we now let
\begin{equation}\label{eq:wccrc:def_S}
S(x,\nu)=P_7(x,\nu)+P_8(x,\nu)+S^*(x,\nu)
\end{equation}
where $S^*$ is some unspecified polynomial to be fixed later. Note that this definition
makes sense because $P_7$ and $P_8$ do not (even indirectly) depend on $S$.
It then follows from \eqref{eq:wccrc:phase_2_bound_delta} that
\[\delta_1(u)\leqslant e^{-S^*(x,\mu(n\pi))}\leqslant 1\]
and thus we can apply Theorem~16 of \cite{BournezGP16b} to get that
\begin{equation}\label{eq:wccrc:phase_2_rel_w_z}
|z_i(u)-w_i(u)|\leqslant\delta_1(u)\leqslant e^{-S^*(x,\mu(n\pi))}\qquad\text{for all }u\in[a,b].
\end{equation}
But in particular, \eqref{eq:wccrc:phase_2_bounded_xi} implies that $b-a\geqslant 1$
so by \eqref{eq:wccrc:phase_2_prec_w}
\begin{equation}\label{eq:wccrc:phase_2_prec_z}
|z_i(b)-f_i(g(x))|\leqslant e^{-S^*(x,\mu(n\pi))}+e^{-R(x,\mu(n\pi))}.
\end{equation}
And finally, using \eqref{eq:wccrc:phase_2_rel_y_z} we get that
\begin{equation}\label{eq:wccrc:phase_2_prec_y_pre}
|y_i((n+\tfrac{3}{4})\pi)-f_i(g(x))|\leqslant e^{-S^*(x,\mu(n\pi))}+e^{-R(x,\mu(n\pi))}.
\end{equation}
Finally we let
\begin{equation}\label{eq:wccrc:def_Ss}
S^*(x,\nu)=R(x,\nu)
\end{equation}
so that
\begin{equation}\label{eq:wccrc:phase_2_prec_y}
|y_i((n+\tfrac{3}{4})\pi)-f_i(g(x))|\leqslant 2e^{-R(x,\mu(n\pi))}.
\end{equation}
Also note using \eqref{eq:wccrc:phase_2_rel_y_z}, \eqref{eq:wccrc:phase_2_bound_w} and \eqref{eq:wccrc:phase_2_rel_w_z}
that
\begin{equation}\label{eq:wccrc:phase_2_bound_y}
|y_i(t)|\leqslant 1+P_6(x,\mu(n\pi),b-a)\leqslant P_9(x,\mu(n\pi))
\end{equation}
for some polynomial\footnote{Note that $P_9$ depends on $P_6$.} $P_9$.

\textbf{Over $\mathbf{[(n+\tfrac{3}{4})\pi,(n+1)\pi]}$:}
for all $j\in\{0,1,2\}$, apply Lemma~\ref{lem:rat_psi_per} to get that
\begin{equation}\label{eq:wccrc:phase_3_bound_psi}
|\psi_{j,\nu^i_j}(t)|\leqslant e^{-\nu^i_j(t)}\quad\text{and}\quad \nu^i_j(t)\geqslant|g_{j,i}(t)|+Q(x,\mu(t)).
\end{equation}
It follows that
\begin{equation}\label{eq:wccrc:phase_3_bound_der_y}
|y_i'(t)|\leqslant 3e^{-Q(x,\mu(t))}\leqslant 3e^{-Q(x,\mu(n\pi))}
\end{equation}
and
\begin{equation}\label{eq:wccrc:phase_3_bound_diff_y}
|y_i(t)-y_i((n+\tfrac{3}{4})\pi)|\leqslant\int_{(n+\tfrac{3}{4})\pi}^t|y_i'(u)|du\leqslant
    3e^{-Q(x,\mu(t))}\leqslant 5e^{-Q(x,\mu(n\pi))}.
\end{equation}
And thus
\begin{align*}
|y_i(t)-f_i(g(x))|
    &\leqslant|y_i(t)-y_i((n+\tfrac{3}{4})\pi)|+|y_i((n+\tfrac{3}{4})\pi)-f_i(g(x))|\\
    &\leqslant 3e^{-Q(x,\mu(n\pi))}+|y_i((n+\tfrac{3}{4})\pi)-f_i(g(x))||&&\text{using \eqref{eq:wccrc:phase_3_bound_diff_y}}\\
    &\leqslant 3e^{-Q(x,\mu(n\pi))}+2e^{-R(x,\mu(n\pi))}&&\text{using \eqref{eq:wccrc:phase_2_prec_y}}\\
    &\leqslant 5e^{-R(x,\mu(n\pi))}&&\hspace{-3em}\text{using \eqref{eq:wccrc:def_Q} and \eqref{eq:wccrc:def_Qs}}.
    \numberthis\label{eq:wccrc:phase_3_stable_y}
\end{align*}
It follows using \eqref{eq:wccrc:phase_3_bound_diff_y} and \eqref{eq:wccrc:phase_2_bound_y} that
\begin{align*}
\infnorm{y((n+1)\pi)}
    &\leqslant 1+\infnorm{y((n+\tfrac{3}{4})\pi)}\\
    &\leqslant 1+P_9(x,\mu(n\pi)).
\end{align*}
We can thus let
\begin{equation}\label{eq:wccrc:def_M}
M(x,\nu)=1+P_9(x,\nu)
\end{equation}
to get the induction invariant. Note, as this is crucial for the proof, that
$M$ does not depend, even indirectly, on $Q$.
\bigskip

We are almost done: the system for $y$ computes $f(g(x))$ with increasing precision
in the time intervals $[(n+\tfrac{3}{4})\pi),(n+1)\pi]$ but the value could be anything
during the rest of the time. To solve this issue, we create an extra system that
``samples'' $y$ during those time intervals, and does nothing the rest of the time.
Consider the system
\[z_i(0)=0,
\qquad z_i'(t)=\psi_{3,\nu^i_3}(t)g_{3,i}(t)\]
where
\begin{align*}
    g_{0,i}(t)&=A_{3,i}(t)(y_i(t)-z_i(t)),\\
    A_{3,i}(t)&=\alpha R(x,\mu(t))+\alpha N(x,\mu(t))\\
    \nu_3^i(t)&=\big(3+g_{3,i}(t)^2+R(x,\mu(t))\big)\beta t.\\
\end{align*}
We will show the following invariant by induction $n$:
\begin{equation}\label{eq:wccrc:prop_N}
\infnorm{z(n\pi)}\leqslant N(x,\mu(n\pi))
\end{equation}
for some polynomial $N$ to be fixed later that \textbf{is not allowed to depend on $R$}. Note that since $z(0)=0$, it is trivially
satisfiable for $n=0$.

\textbf{Over $\mathbf{[0.\tfrac{1}{\beta}]}$:} similarly to $y$, the existence of $z$
is not clear over this time interval because of the bootstrap time of $\nu_3^i$.
Since the argument is very similar to that of $y$ (simpler in fact), we do not repeat it.

\textbf{Over $\mathbf{[n\pi,(n+\tfrac{3}{4})\pi]}$ for $\mathbf{n\geqslant 1}$:}
apply Lemma~\ref{lem:rat_psi_per} to get that
\[|\psi_{3,\nu^i_3}(t)|\leqslant e^{-\nu^i_3(t)}\leqslant e^{-|g_{i,3}(t)|-Q(x,\mu(n\pi))-2}.\]
It follows that for all $t\in[n\pi,(n+\tfrac{3}{4})\pi]$,
\begin{equation}\label{eq:wccrc:copy_stable_phase}
|z_i(t)-z_i(n\pi)|
    \leqslant \tfrac{3}{4}\pi e^{-R(x,\mu(n\pi))-2}
    \leqslant e^{-R(x,\mu(n\pi))}\leqslant 1.
\end{equation}

\textbf{Over $\mathbf{[(n+\tfrac{3}{4})\pi,(n+1)\pi]}$:} apply Lemma~\ref{lem:conv_reach_deriv}
to get that
\begin{equation}\label{eq:wccrc:copy_phase_bound_z_y}
|z_i(t)-y_i(t)|\leqslant
    |z_i((n+\tfrac{3}{4})\pi)-y_i((n+\tfrac{3}{4})\pi)|e^{-B(t)}+|y_i(t)-y_i((n+\tfrac{3}{4})\pi)|
\end{equation}
where
\[B(t)=\int_{(n+\tfrac{3}{4})\pi}^{t}A_{3,i}(u)\psi_{3,\nu^i_3}(u)du.\]
Let $b=(n+1)\pi$, then
\begin{align*}
B(b)&=\int_{(n+\tfrac{3}{4})\pi}^{(n+1)\pi}A_{3,i}(u)\psi_{3,\nu^i_3}(u)du\\
    &\geqslant \alpha(R(x,\mu(n\pi))+N(x,\mu(n\pi)))\int_{(n+\tfrac{3}{4})\pi}^{(n+1)\pi}\psi_{3,\nu^i_3}(u)du
        &&\text{using Lemma~\ref{lem:rat_psi_per}}\\
    &\geqslant (R(x,\mu(n\pi))+N(x,\mu(n\pi)))\alpha m_\psi\\
    &\geqslant R(x,\mu(n\pi))+N(x,\mu(n\pi))&&\text{using }\alpha m_\psi\geqslant 1.
\end{align*}
It follows that
\begin{align*}
|z_i(b)-y_i(b)|
    &\leqslant |z_i((n+\tfrac{3}{4})\pi)-y_i((n+\tfrac{3}{4})\pi)|e^{-R(x,\mu(n\pi))-N(x,\mu(n\pi))}\\
    &\hspace{1em}+|y_i(t)-y_i((n+\tfrac{3}{4})\pi)|\\
    &\leqslant \big(|z_i((n+\tfrac{3}{4})\pi)|+|y_i((n+\tfrac{3}{4})\pi)|\big)e^{-R(x,\mu(n\pi))-N(x,\mu(n\pi))}\\
    &\hspace{1em}+|y_i(t)-y_i((n+\tfrac{3}{4})\pi)|\\
    &\leqslant \big(|z_i(n\pi)|+1+|y_i((n+\tfrac{3}{4})\pi)|\big)e^{-R(x,\mu(n\pi))-N(x,\mu(n\pi))}\\
    &\hspace{1em}+5e^{-R(x,\mu(n\pi))}&&\text{using \eqref{eq:wccrc:phase_3_bound_diff_y}}\\
    &\leqslant \big(N(x,\mu(n\pi))+1+|y_i((n+\tfrac{3}{4})\pi)|\big)e^{-R(x,\mu(n\pi))-N(x,\mu(n\pi))}&&\text{using \eqref{eq:wccrc:prop_N}}\\
    &\hspace{1em}+5e^{-R(x,\mu(n\pi))}\\
    &\leqslant |y_i((n+\tfrac{3}{4})\pi)|e^{-R(x,\mu(n\pi))-N(x,\mu(n\pi))}
        +7e^{-R(x,\mu(n\pi))}\\
    &\leqslant (|y_i((n+1)\pi)|+1)e^{-R(x,\mu(n\pi))-N(x,\mu(n\pi))}
        +7e^{-R(x,\mu(n\pi))}&&\text{using \eqref{eq:wccrc:phase_3_bound_diff_y}}\\
    &\leqslant (M(x,\mu(n\pi))+1)e^{-R(x,\mu(n\pi))-N(x,\mu(n\pi))}
        +7e^{-R(x,\mu(n\pi))}&&\text{using \eqref{eq:wccrc:prop_M}}.
\end{align*}
Since we have not specified $N$ yet, we can take
\begin{equation}\label{eq:wccrc:def_N}
N(x,\nu)=M(x,\mu)
\end{equation}
so that
\begin{equation}\label{eq:wccrc:copy_phase_conv}
|z_i(b)-y_i(b)|\leqslant 8e^{-R(x,\mu(n\pi))}.
\end{equation}
It follows that
\begin{align*}
|z_i(b)-f_i(g(x))|
    &\leqslant |z_i(t)-y_i(t)|+|y_i(t)-f_i(g(x))|\\
    &\leqslant 8e^{-R(x,\mu(n\pi))}+|y_i(t)-f_i(g(x))|&&\text{using \eqref{eq:wccrc:copy_phase_conv}}\\
    &\leqslant 8e^{-R(x,\mu(n\pi))}+5e^{-R(x,\mu(n\pi)}&&\text{using \eqref{eq:wccrc:phase_3_stable_y}}\\
    &\leqslant 13e^{-R(x,\mu(n\pi))}.\numberthis\label{eq:wccrc:copy_phase_conv_2}
\end{align*}
Furthermore \eqref{eq:wccrc:copy_phase_bound_z_y} gives that for all $t\in[(n+\tfrac{3}{4})\pi,(n+1)\pi]$,
\begin{align*}
|z_i(t)-y_i(t)|
    &\leqslant|z_i((n+\tfrac{3}{4})\pi)-y_i((n+\tfrac{3}{4})\pi)|+|y_i(t)-y_i((n+\tfrac{3}{4})\pi)|\\
    &\leqslant|z_i((n+\tfrac{3}{4})\pi)-y_i((n+\tfrac{3}{4})\pi)|+5e^{-Q(x,\mu(n\pi))}&&\text{using \eqref{eq:wccrc:phase_3_bound_diff_y}}\\
    &\leqslant|z_i((n+\tfrac{3}{4})\pi)-z_i(n\pi)|
        +|z_i(n\pi)-f_i(g(x))|\\
    &\hspace{1em}+|f_i(g(x))-y_i((n+\tfrac{3}{4})\pi)|
        +5e^{-Q(x,\mu(n\pi))}\\
    &\leqslant e^{-R(x,\mu(n\pi))}
        +|z_i(n\pi)-f_i(g(x))|&&\text{using \eqref{eq:wccrc:copy_stable_phase}}\\
    &\hspace{1em}+2e^{-R(x,\mu(n\pi))}
        +5e^{-Q(x,\mu(n\pi))}&&\text{using \eqref{eq:wccrc:phase_2_prec_y}}\\
    &\leqslant 8e^{-R(x,\mu(n\pi))}
        +|z_i(n\pi)-f_i(g(x))|.\numberthis\label{eq:wccrc:copy_phase_transition}
\end{align*}
\bigskip

We can now leverage this analysis to conclude: putting \eqref{eq:wccrc:copy_stable_phase} and \eqref{eq:wccrc:copy_phase_conv_2}
together we get that
\begin{equation}\label{eq:wccrc:copy_phase_prec_z}
|z_i(t)-f_i(g(x))|\leqslant 14e^{-R(x,\mu(n\pi))}\qquad\text{for all }t\in[(n+1)\pi,(n+\tfrac{7}{4})\pi]
\end{equation}
and for all $t\in[(n+\tfrac{7}{4})\pi,(n+2)\pi]$,
\begin{align*}
|z_i(t)-f_i(g(x))|
    &\leqslant|z_i(t)-y_i(t)|+|y_i(t)-f_i(g(x))|\\
    &\leqslant 8e^{-R(x,\mu(n\pi))}
        +|z_i(n\pi)-f_i(g(x))|+|y_i(t)-f_i(g(x))|&&\text{using \eqref{eq:wccrc:copy_phase_transition}}\\
    &\leqslant 8e^{-R(x,\mu(n\pi))}+|z_i(n\pi)-f_i(g(x))|
        +2e^{-R(x,\mu(n\pi))}&&\text{using \eqref{eq:wccrc:phase_2_prec_y}}\\
    &\leqslant 10e^{-R(x,\mu(n\pi))}+|z_i(n\pi)-f_i(g(x))|.\numberthis\label{eq:wccrc:copy_phase_rel_y_z}
\end{align*}
And finally, putting \eqref{eq:wccrc:copy_phase_prec_z} and \eqref{eq:wccrc:copy_phase_rel_y_z}
together, we get that
\begin{equation}
|z_i(t)-f_i(g(x))|\leqslant 24e^{-R(x,\mu(n\pi))}\qquad\text{for all }t\in[(n+1)\pi,(n+2)\pi]
\end{equation}
Since we have not specified $R$ yet, we can
take
\begin{equation}\label{eq:wccrc:def_R}
R(x,\nu)=24+\nu
\end{equation}
so that
\begin{equation}\label{eq:wccrc:conclusion}
|z_i(t)-f_i(g(x))|\leqslant e^{-\mu(n\pi)}\leqslant e^{-n}\qquad\text{for all }t\in[(n+1)\pi,(n+2)\pi].
\end{equation}
This concludes the proof that $f\circ g\in\gpc[\Q]$ since we have proved that the system
converges quickly, has bounded values and the entire system has a polynomial right-hand side
using rational numbers only.
\end{proof}

\subsection{From $\gpwc[\Rgen]$ to $\gpwc[\Q]$}

The second step of the proof is to recast the problem entirely in the language
of $\gpwc[\Q]$. The observation is that given a system, corresponding to $f\in\gpwc[\Rgen]$,
we can abstract away the coefficients and make them part of the input, so that
$f(x)=g(x,\alpha)$ where $g\in\gpwc[\Q]$ and $\alpha\in\Rgen^k$. We then show that
we can see $\alpha$ as the result of a computation itself: we build $h\in\gpwc[\Q]$
such that $\alpha=h(1)$. Now we are back to $x=g(x,h(1))$, in other words a composition
of functions in $\gpwc[\Q]$.

First, let us recall the definition of $\Rgen$ from \cite{BournezGP16a}:
\[\Rgen=\bigcup_{n\geqslant 0}\fiter{G}{n}(\Q)\]
where
\[G(X)=\{f(1):(f:\R\to\R)\in\gpval[X]{}\}.\]
Note that in \cite{BournezGP16a}, we defined $G$ slightly differently using
$\myclass{GVAL}$, the class of generable functions, instead of $\gpval$. Those two
definitions are equivalent because if $f\in\myclass{GVAL}[X]{}$, we can define $h(t)=f(\tfrac{2t}{1+t^2})$
that is such that $h(0)=f(0)$, $h(1)=f(1)$ and belongs to $\gpval[X]$.

\begin{lemma}\label{lem:decompose_f_gpwc_rat_nonrat}
Let $(f:\in\gpwc[X]$ where $\Q\subseteq X$, then there exists $\ell\in\N$, $\beta\in X^\ell$ and $h\in\gpwc[\Q]$
with $\dom{h}=\dom{f}$ such that $f=h\circ g$ where $g(x)=(x,\beta)$ for all $x\in\dom{f}$.
\end{lemma}

\begin{proof}
Let $\myOmega$ and $\Upsilon$ polynomials such that
$(f:\subseteq\R^n\to\R^m)\in\gwc{\Upsilon}{\myOmega}$ with
corresponding $d,q$ and $p$. Let $x\in\dom{f}$ and $\mu\geqslant 0$
and consider the following system:
\[y(0)=q(x,\mu),\qquad y'(t)=q(y(t)).\]
By definition, for any $t\geqslant\myOmega(\infnorm{x},\mu)$,
\[\infnorm{y_{1..m}(t)-f(x)}\leqslant e^{-\mu}\]
and for all $t\geqslant 0$,
\[\infnorm{y(t)}\leqslant\Upsilon(\infnorm{x},\mu,t).\]
Let $\ell$ be the number of nonzero coefficients of $p$ and $q$. Then there exists $\beta\in R^\ell$,
$\hat{p}\in\Q^d[\R^{d+\ell}]$ and $\hat{q}\in\Q^d[\R^{n+1+\ell}]$
such that for all $x\in\R^n$, $\mu\geqslant 0$ and $u\in\R^d$,
\[q(x,\mu)=\hat{q}(x,\beta,\mu)\quad\text{and}\quad p(u)=\hat{p}(u,\beta).\]
Now consider the following system for any $(x,w)\in\dom{f}\times\{\beta\}$ and $\mu\geqslant 0$:
\[\begin{array}{l@{,\qquad}l}
    u(0)=w& u'(t)=0,\\
    z(0)=\hat{q}(x,w,\mu)& z'(t)=\hat{p}(z(t),u(t)).
    \end{array}\]
Note that this system only has rational coefficients because $\hat{q}$ and $\hat{p}$
have rational coefficients. Also $u(t)$ is the constant function equal to $w$ and $w=\beta$ since $(x,w)\in\dom{f}\times\{\beta\}$.
Thus $z'(t)=\hat{p}(z(t),\beta)=p(z(t))$, and $z(0)=\hat{q}(x,\beta,\mu)=q(x,\mu)$.
It follows that $z\equiv y$ and thus this system weakly-computes $h(x,w)=f(x)$:
\[\infnorm{z_{1..m}(t)-f(x)}=\infnorm{y_{1..m}(t)-f(x)}\leqslant e^{-\mu}\]
and
\[\infnorm{(u(t),z(t))}\leqslant\infnorm{u(t)}+\infnorm{z(t)}\leqslant\infnorm{w}+\Upsilon(\infnorm{x},\mu,t).\]
Thus $h\in\gpwc[\Q]$. It is clear that if $g(x)=(x,\beta)$ then $(h\circ g)(x)=h(x,\beta)=f(x)$.
\end{proof}

\begin{lemma}\label{lem:rgen_to_gpwc_q}
For any $X\supseteq\Q$ and $(f:\R\to\R)\in\gpval[X]{}$, $(x\in\R\mapsto f(1))\in\gpwc[X]$.
\end{lemma}

\begin{proof}
Expand the definition of $f$ to get $d\in\N$,
$y_0\in X^{d}$ and $p\in X[\R^{d}]$
such that
\[y(0)=y_0,\qquad y'(t)=p(y(t))\]
satisfies for all $t\in\R$,
\[f(t)=y_1(t).\]
Now consider the following system for $x\in\R$ and $\mu\geqslant 0$:
\[\begin{array}{l@{,\qquad}l}
    \psi(0)=1& \psi'(t)=-\psi(t),\\
    z(0)=y_0& z'(t)=\psi(t)p(z(t)).
    \end{array}\]
This system only has coefficients in $X$ and it is not hard to see that
\[\psi(t)=e^{-t}\quad\text{and}\quad z(t)=y\left(\int_0^t\psi(s)ds\right)=y\left(1-e^{-t}\right).\]
Furthermore, since $f(1)=y_1(1)$,
\begin{align*}
|f(1)-z_1(t)|
    &=|y_1(1)-y_1(1-e^{-t})|\\
    &=\left|\int_{1-e^{-t}}^1y_1'(s)ds\right|\\
    &\leqslant\int_{1-e^{-t}}^1|p_1(y(s))|ds\\
    &\leqslant e^{-t}\sup_{s\in[0,1]}|p_1(y(s))|.
\end{align*}
Let $A=\sup_{s\in[0,1]}|p_1(y(s))|$ which is finite because $y$ is continuous and $[0,1]$
is compact, and let $\myOmega(x,\mu)=\mu+A$. Then for any $\mu\geqslant 0$ and $t\geqslant\myOmega(\infnorm{x},\mu)$,
\[|f(1)-z_1(t)|\leqslant e^{-t}A\leqslant e^{-\myOmega(\infnorm{x},\mu)}A\leqslant e^{-\mu-A}A\leqslant e^{-\mu}.\]
Furthermore,
\[\infnorm{z(t)}=\infnorm{y(1-e^t)}\leqslant\sup_{s\in[0,1]}\infnorm{y(s)}\]
where the right-hand is a finite constant because $y$ is continuous and $[0,1]$.
This shows that $(x\in\R\mapsto f(1))\in\gpwc[X]$.
\end{proof}

\begin{proposition}\label{prop:gpwc_rgen_to_gpc}
For all $n\in\N$, $\gpwc[\fiter{G}{n}(\Q)]\subseteq\gpc[\Q]$.
\end{proposition}

\begin{proof}
When $n=0$, the result is trivial because $\fiter{G}{0}(\Q)=\Q$.

Assume the result is true for $n$ and take $f\in\gpwc[\fiter{G}{n+1}(\Q)]$.
Apply Lemma~\ref{lem:decompose_f_gpwc_rat_nonrat} to get $h\in\gpwc[\Q]$ such that $f=h\circ g$
where $g(x)=(x,\beta)$ where $\beta\in\fiter{G}{n+1}(\Q)^\ell$ for some $\ell\in\N$.
Let $i\in\{1,\ldots,\ell\}$, by definition of $\beta_i$, there exists $y_i\in\gpval[\fiter{G}{n}(\Q)]{}$
such that $\beta_i=y_i(1)$. Apply Lemma~\ref{lem:rgen_to_gpwc_q} to get that
$(x\in\R\mapsto y_i(1))\in\gpwc[\fiter{G}{n}(\Q)]$. Now by induction,
$(x\in\R\mapsto \beta_i)=(x\in\R\mapsto y_i(1))\in\gpwc[\Q]$. Putting all those
systems together, and adding variables to keep a copy of the input, it easily follows
that $g\in\gpwc[\Q]$. Apply Theorem~\ref{th:weak_computable_compose_rat_comp}
to conclude that $f=h\circ g\in\gpc[\Q]$.
\end{proof}

We can now prove the main theorem of this section. 

\begin{theorem}$\gpwc[\Rgen]=\gpc[\Q]$.
\end{theorem}

\begin{proof}
The inclusion $\gpwc[\Rgen]\subseteq\gpwc[\Q]$ is trivial. Conversely, take $f\in\gpwc[\Rgen]$.
The system that computes $f$ only has a finite number of coefficients, all in $\Rgen$.
Thus there exists $n\in\N$ such that all the coefficients belong to $\fiter{G}{n}(\Q)$
and then $f\in\gpwc[\fiter{G}{n}(\Q)]$. Apply Proposition~\ref{prop:gpwc_rgen_to_gpc}
to conclude.
\end{proof}

As clearly  $\gplc[\K]=\gpc[\K]$ over any
field $\K$ \cite{\JOURNALOFCOMPLEXITY}, it follows that
$\gplc= \gpwc[\Rgen] = \gpc[\Q] =  \gplc[\Q]$
and hence Definitions \ref{def:gplc} and \ref{def:glc} are defining the
same class. Similarly, and consequently,  Definitions
\ref{def:discrete_rec_q} and \ref{def:discrete_rec} are also defining the
same class.

\newpage
\appendix

\section*{APPENDIX}
\section{Notations}\label{sec:notations}

\renewcommand{\arraystretch}{1.2}

\subsection*{Sets}
\begin{longtable}{@{}p{4cm}lp{6.8cm}@{}}
Concept & Notation & Comment \\
\midrule
\endhead
Real interval & $[a,b]$ & $\{x\in\R | \thinspace a \leqslant x \leqslant b \}$ \\
    & $[a,b[$ & $\{x\in\R | \thinspace a \leqslant x < b \}$ \\
    & $]a,b]$ & $\{x\in\R | \thinspace a < x \leqslant b \}$ \\
    & $]a,b[$ & $\{x\in\R | \thinspace a < x < b \}$ \\
Line segment & $[x,y]$ & $\{(1-\alpha)x+\alpha y\in\R^n,\alpha\in[0,1]\}$\\
    & $[x,y[$ & $\{(1-\alpha)x+\alpha y\in\R^n,\alpha\in[0,1[\}$\\
    & $]x,y]$ & $\{(1-\alpha)x+\alpha y\in\R^n,\alpha\in]0,1]\}$\\
    & $]x,y[$ & $\{(1-\alpha)x+\alpha y\in\R^n,\alpha\in]0,1[\}$\\
Integer interval & $\intinterv{a}{b}$ & $\{a,a+1,\ldots,b\}$ \\
Natural numbers & $\N$ & $\{0, 1, 2, \ldots\}$\\
Integers & $\Z$ & $\{\ldots,-2,-1,0,1,2,\ldots\}$\\
Rational numbers & $\Q$ & \\
Dyadic rationnals & $\D$ & $\{m2^{-n},m\in\Z,n\in\N\}$ \\
Real numbers & $\R$ & \\
Non-negative numbers & $\Rp$ & $\Rp=[0,+\infty[$\\
Non-zero numbers & $\Rs$ & $\Rs=\R\setminus\{0\}$ \\
Positive numbers & $\Rps$ & $\Rps=]0,+\infty[$ \\
Set shifting & $x+Y$ & $\{x+y,y\in Y\}$\\
Set addition & $X+Y$ & $\{x+y,x\in X, y\in Y\}$\\
Matrices & $\MAT{n}{m}{\K}$ & Set of $n\times m$ matrices over field $\K$\\
& $\MAT{n}{}{\K}$ & Shorthand for $\MAT{n}{n}{\K}$\\
& $\MAT{n}{m}{}$ & Set of $n\times m$ matrices over a field is deduced from the context\\
Polynomials & $\K[X_1,\ldots,X_n]$ & Ring of polynomials with variables $X_1,\ldots,X_n$
    and coefficients in $\K$\\
& $\K[\A^n]$ & Polynomial functions with $n$ variables, coefficients in $\K$
    and domain of definition $\A^n$ \\
Fractions & $\K(X)$ & Field of rational fractions with coefficients in $\K$\\
Power set & $\powerset{X}$ & The set of all subsets of $X$\\
Domain of definition & $\dom{f}$ & If $f:I\rightarrow J$ then $\dom{f}=I$ \\
Cardinal & $\card{X}$ & Number of elements\\
Polynomial vector & $\K^n[\A^d]$ & Polynomial in $d$ variables with coefficients in $\K^n$ \\
    & $\K[\A^d]^n$ & Isomorphic $\K^n[\A^d]$\\
Polynomial matrix & $\MAT{n}{m}{\K}[\A^n]$ & Polynomial in $n$ variables with matrix coefficients \\
    & $\MAT{n}{m}{\K[\A^n]}$ & Isomorphic $\MAT{n}{m}{\K}[\A^n]$\\
Smooth functions & $C^k$ & Partial derivatives of order $k$ exist and are continuous\\
    & $C^\infty$ & Partial derivatives exist at all orders\\
\bottomrule
\end{longtable}

\subsection*{Complexity classes}
\begin{longtable}{@{}p{5cm}lp{5.5cm}@{}}
Concept & Notation & Comment \\
\midrule
\endhead
Polynomial Time & $\PTIME$ & Class of decidable languages\\
    & $\FP$ & Class of computable functions\\
Polytime computable numbers & $\Rpoly$ & 
\\
Polytime computable real functions & $\pcab{a}{b}$ & Over compact interval $[a,b]$\\
Generable reals & $\Rgen$ & See \cite{BournezGP16a} \\
Poly-length-computability &  $\mygpc$& See Definition~\ref{def:gplc}\\
    & $\gc{\Upsilon}{\myOmega}$ & Notation defined page \pageref{page:def:gc} \\
        & $\goc{\Upsilon}{\myOmega}{\Lambda}$ & Notation defined page \pageref{page:def:goc} \\
        & $\guc{\Upsilon}{\myOmega}{\Lambda}{\Theta}$ & Notation defined page \pageref{page:def:guc} \\
\bottomrule
\end{longtable}

\subsection*{Metric spaces and topology}
\begin{longtable}{@{}p{5cm}lp{6.5cm}@{}}
Concept & Notation & Comment \\
\midrule
\endhead
$p$-norm & $\inorm{x}{p}$ & $\displaystyle\left(\sum_{i=1}^{n}|x_i|^p\right)^{\frac{1}{p}}$ \\
Infinity norm & $\infnorm{x}$ & $\max(|x_1|,\ldots,|x_n|)$ \\
\bottomrule
\end{longtable}

\subsection*{Polynomials}
\begin{longtable}{@{}p{4.5cm}lp{6.8cm}@{}}
Concept & Notation & Comment \\
\midrule
\endhead
Univariate polynomial & $\displaystyle\sum_{i=0}^{d}a_iX^i$ &\\
Multi-index & $\alpha$& $(\alpha_1,\ldots,\alpha_k)\in\N^k$ \\
& $|\alpha|$ & $\alpha_1+\cdots+\alpha_k$ \\
& $\alpha!$ & $\alpha_1!\alpha_2!\cdots\alpha_k!$ \\
Multivariate polynomial & $\displaystyle\sum_{|\alpha|\leqslant d}a_\alpha X^\alpha$ & where $X^\alpha=X_1^{\alpha_1}\cdots X_k^{\alpha_k}$\\
Degree & $\degp{P}$ & Maximum degree of a monomial, $X^\alpha$ is of degree $|\alpha|$, conventionally $\degp{0}=-\infty$\\
& $\degp{P}$ & $\max(\degp{P_i})$ if $P=(P_1,\ldots,P_n)$ \\
& $\degp{P}$ & $\max(\degp{P_{ij}})$ if $P=(P_{ij})_{i\in\intinterv{1}{n},j\in\intinterv{1}{m}}$ \\
Sum of coefficients & $\sigmap{P}$ & $\sigmap{P}=\sum_{\alpha}|a_\alpha|$ \\
& $\sigmap{P}$ & $\max(\sigmap{P_1},\ldots,\sigmap{P_n})$ if $P=(P_1,\ldots,P_n)$ \\
& $\sigmap{P}$ & $\max(\sigmap{P_{ij}})$ if $P=(P_{ij})_{i\in\intinterv{1}{n},j\in\intinterv{1}{m}}$ \\
A polynomial & $\poly$ & An unspecified polynomial \\
\bottomrule
\end{longtable}

\subsection*{Miscellaneous functions}
\begin{longtable}{@{}p{5cm}lp{6.5cm}@{}}
Concept & Notation & Comment \\
\midrule
\endhead
Sign function & $\sgn{x}$ & Conventionally $\sgn{0}=0$ \\
Ceiling function & $\ceil{x}$ & $\min\{n\in\Z: x\leqslant n\}$ \\
Rounding function & $\round{x}$ & $\argmin_{n\in\Z}|n-x|$, undefined for $x=n+\frac{1}{2}$ \\
Integer part function & $\intp(x)$ & $\max(0,\lfloor x\rfloor)$\\
    & $\intp_n(x)$ & $\min(n,\intp(x))$ \\
Fractional part function & $\fracp(x)$ & $x-\intp{x}$ \\
    & $\fracp_n(x)$ & $x-\intp_n(x)$\\
Composition operator & $f\circ g$ & $(f\circ g)(x)=f(g(x))$ \\
Identity function & $\idfun$ & $\idfun(x)=x$ \\
Indicator function & $\indicator{X}$ & $\indicator{X}(x)=1$ if $x\in X$ and $\indicator{X}(x)=0$ otherwise \\
$n^{th}$ iterate & $\fiter{f}{n}$ & $\fiter{f}{0}=\idfun$ and $\fiter{f}{n+1}=\fiter{f}{n}\circ f$\\
\bottomrule
\end{longtable}

\subsection*{Calculus}
\begin{longtable}{@{}p{4cm}lp{6.6cm}@{}}
Concept & Notation & Comment \\
\midrule
\endhead
Derivative & $f'$ &\\
$n^{th}$ derivative & $f^{(n)}$ & $f^{(0)}=f$ and $f^{(n+1)}={f^{(n)}}'$\\
Partial derivative & $\partial_if,\frac{\partial f}{\partial x_i}$ & with respect to the $i^{th}$ variable\\
Scalar product & $\scalarprod{x}{y}$ & $\sum_{i=1}^nx_iy_i$ in $\R^n$\\
Gradient & $\grad{f}(x)$ & $(\partial_1 f(x), \ldots, \partial_nf(x))$ \\
Jacobian matrix & $\jacobian{f}(x)$ & $(\partial_jf_i(x))_{i\in\intinterv{1}{n},j\in\intinterv{1}{m}}$\\
Taylor approximation & $\taylor{a}{n}{f}(t)$ & $\displaystyle\sum_{k=0}^{n-1}\frac{f^{(k)}(a)}{k!}(t-a)^k$ \\
Big O notation & $f(x)=\bigO{g(x)}$ & $\exists M,x_0\in\R$, $|f(x)|\leqslant M|g(x)|$ for all $x\geqslant x_0$ \\
Soft O notation & $f(x)=\softO{g(x)}$ & Means $f(x)=\bigO{g(x)\log^kg(x)}$ for some $k$\\
Subvector & $x_{i..j}$ & $(x_i,x_{i+1},\ldots,x_j)$\\
Matrix transpose & $\transpose{M}$ &\\
Past supremum & $\pastsup{\delta}{f}(t)$ & $\sup_{u\in[t,t-\delta]\cap\Rp}f(t)$\\
Partial function & $f:\subseteq X\rightarrow Y$ & $\dom{f}\subseteq X$\\
Restriction & $\frestrict{f}{I}$ & $\frestrict{f}{I}(x)=f(x)$ for all $x\in\dom{f}\cap I$\\
\bottomrule
\end{longtable}

\subsection*{Words}
\begin{longtable}{@{}p{4cm}lp{7.6cm}@{}}
Concept & Notation & Comment \\
\midrule
\endhead
Alphabet & $\Sigma,\Gamma$ & A finite set\\
Words & $\Sigma^*$ & $\bigcup_{n\geqslant0}\Sigma^n$\\
Empty word & $\emptyword$ &\\
Letter & $w_i$ & $i^{th}$ letter, starting from one\\
Subword & $w_{i..j}$ & $w_iw_{i+1}\cdots w_j$\\
Length & $|w|$ &\\
Repetition & $w^k$ & $\underbrace{ww\cdots w}_{k\text{ times}}$\\
\bottomrule
\end{longtable}

\bibliographystyle{ACM-Reference-Format-Journals}
\bibliography{extracted}

\end{document}